\documentclass[aap,MSNbibl,dvips]{arximspdf}
\usepackage{mathbh}

\doi{10.1214/12-AAP917} %kopijuoti is PTS
\volume{24}
\issue{2}
\pubyear{2014}
\firstpage{476}
\lastpage{525}

\makeatletter
\newcommand{\rrVert}{\Vert}
\newcommand{\rrvert}{\vert}
\newcommand{\llVert}{\Vert}
\newcommand{\llvert}{\vert}
\newcommand{\eqref}[1]{(\ref{#1})}
\newtheorem{Theorem}{Theorem}[section]
\newtheorem{Lemma}[Theorem]{Lemma}
\newtheorem{Proposition}[Theorem]{Proposition}
\newtheorem{Corollary}[Theorem]{Corollary}
\newproclaim{remark}[Theorem]{Remark}
\newtheorem{claim}[Theorem]{Claim}
\newproclaim{definition}[Theorem]{Definition}
\makeatother

\begin{document}
\begin{frontmatter}

\title{Universality for one-dimensional hierarchical coalescence
processes with double and triple~merges\thanksref{T1}}
\runtitle{1D hierarchical coalescence processes}

\begin{aug}
\author[A]{\fnms{A.} \snm{Faggionato}\ead[label=e1]{faggiona@mat.uniroma1.it}},
\author[B]{\fnms{C.} \snm{Roberto}\corref{}\ead[label=e2]{croberto@math.cnrs.fr}}
\and
\author[C]{\fnms{C.} \snm{Toninelli}\ead[label=e3]{cristina.toninelli@upmc.fr}\thanksref{t2}}
\runauthor{A. Faggionato, C. Roberto and C. Toninelli}
\affiliation{University La Sapienza, Universit\'{e} Paris Ouest Nanterre La D\'{e}fense and Universit\'{e} Paris VI-VII}
\address[A]{A. Faggionato\\
Dip. Matematica G.~Castelnuovo\\
University La Sapienza\\
P.le Aldo Moro 2\\
00185 Roma\\
Italy\\
\printead{e1}} %adresu isvedimo komanda gale!
\address[B]{C. Roberto\\
MODAL'X\\
Universit\'{e} Paris Ouest Nanterre La D\'{e}fense\\
200 avenue de la R\'{e}publique 92000 Nanterre\\
France\\
\printead{e2}}
\address[C]{C. Toninelli\\
L.P.M.A. and CNRS-UMR 7599\\
Universit\'{e} Paris VI-VII\\
4 Pl. Jussieu 75252 Paris\\
France\\
\printead{e3}}
\end{aug}
\thankstext{T1}{Supported by the European Research Council through the
``Advanced Grant'' PTRELSS 228032.}
\thankstext{t2}{Supported in part by the French Ministry of
Education through the Grant ANR-2010-BLAN-0108.}

% HISTORY:
\received{\smonth{1} \syear{2012}}
\revised{\smonth{12} \syear{2012}}

% ABSTRACT
%
\begin{abstract}
We consider one-dimensional hierarchical coalescence processes (in
short HCPs) where two or three neighboring domains can merge. An HCP
consists of an infinite sequence of stochastic coalescence processes:
each process occurs in a~different ``epoch'' and evolves for an
infinite time, while the evolutions in subsequent epochs are linked in
such a~way that the initial distribution of epoch $n+1$ coincides with
the final distribution of epoch $n$. Inside each epoch a~domain can
incorporate one of its neighboring domains or both of them if its
% its left or its right or both its left and right
%neighboring interval. Merging events occur with quite general rates
%with the constraint that only intervals whose
length belongs to a~certain epoch-dependent finite range.
% are \emph{active}, \ie they can
%C merge with their left or right neighboring interval with quite
%general rates. Inactive intervals are frozen.
%incorporate other intervals.
%Inactive intervals cannot incorporate their neighbors and can increase
%their length only if they
% are incorporated by active neighbors.
%The activity ranges are such that after a~merging
%step the newly produced interval
%always becomes inactive for that epoch but active for some future
%epoch.

Assuming that the distribution at the beginning of the first epoch is
described by a~renewal simple point process, we prove limit theorems
for the domain length and for the position of the leftmost
point (if any). %These results show the existence of universality
%classes which explain why very different models and very different
%initial distribution lead to the same limiting behavior.
Our analysis extends the results obtained in [\textit{Ann. Probab.}
\textbf{40} (2012) 1377--1435] to a~larger family of models, including
relevant examples from the physics literature [\textit{Europhys. Lett.}
\textbf{27} (1994) 175--180, \textit{Phys. Rev. E} (3) \textbf{68}
(2003) 031504]. It reveals the presence of a~common abstract structure
behind models which are apparently very different, thus leading to very
similar limit \mbox{theorems}. Finally, we give here a~full
characterization of the infinitesimal generator for the dynamics inside
each epoch, thus allowing us to describe the time evolution of the
expected value of regular observables in terms of an ordinary
differential equation.
\end{abstract}

% KEYWORDS

\end{frontmatter}

%s1 #&#
\section{Introduction}
\setcounter{footnote}{2}

%We consider a~large class of one-dimensional hierarchical
%coalescence processes, HCPs.
A one-dimensional hierarchical
coalescence process\break  (HCP) consists of an infinite sequence of
one-dimensional coalescence processes:
%$\{\xi^{(n)}(\cdot)\}_{n\ge1}$:
each process occurs in a~different \textit{epoch}
%(indexed by $n$)
and evolves for an infinite time, while the evolution in subsequent
epochs are linked in such a~way that the initial distribution of epoch
$n+1$
%$\xi^{(n+1)}$
coincides with
the final distribution of epoch $n$.
%Then the corresponding HCP is such that
%$\xi^{(n)}$. I
%**need informal descriotion**.
At a~given time inside epoch $n$ the state of the process is described
by a~simple point process on $\mathbb{R}$, that is, by a~random locally
finite subset of $\mathbb{R}$, such that the intervals among
consecutive points (\textit{domains}) are not smaller than $d^{(n)}$,
where $\{d^{(n)}\}_{n\geq1}$ is an a~priori fixed sequence of strictly
increasing and diverging positive numbers. The evolution inside epoch
$n$ can be informally described as follows.
% process evolves as follows.
%Inside each epoch the process is
%described by a~suitable simple point process on $\bbR$ or on $\bbZ$
%representing the
%boundaries between adjacent intervals (domains), evolves as follows.
Only domains whose length belongs to the
%epoch-dependent
finite range $[d^{(n)},d^{(n+1)})$ are
\emph{active}, that is, they can
%C merge with their left or right neighboring interval with quite
%general rates. Inactive intervals are frozen.
incorporate their left neighboring domain, their right neighboring
domain or both of them. Inactive domains cannot incorporate their
neighbors and can increase their length only if they are incorporated
by active neighbors. The rates of the merging events and the sequence
$\{d^{(n)}\}_{n\geq 1}$ are quite general, with the important feature
that the activity ranges $[d^{(n)},d^{(n+1)})$ should be such that
after each merging step the newly produced domain always becomes
inactive for that epoch but active for some future epoch.

We have introduced the concept of HCP in \cite{FMRT0}, considering only
left or right merging of domains, that is, a~domain cannot incorporate
simultaneously both its neighbors. There
%, here instead we also allow triple merging: We
we proved that if the initial distribution is a~renewal process, such
property is preserved at all times and epochs and the distribution of
certain rescaled variables---the domain length and the position of the
leftmost point (if any)---has a~well-defined limiting behavior
corresponding to large universality classes (most of the dynamical
details disappear in the scaling limit).
% The analysis in
%coalescence inside a~given epoch, which becomes extremely hard in
%the present setting. Hence, here we have followed a~different route.
Here we extend these results to the more general HCPs defined above
which also allow triple merging, and we determine the corresponding
limiting behavior and universality classes.

Besides the mathematical interest, our study has been motivated by the
fact that several HCPs have been implicitly introduced in the physics
literature to model the nonequilibrium evolution of one-dimensional
systems whose dynamics is dominated by the coalescence of proper
domains or droplets characterizing the experiments. We refer to
Section~\ref{physics} for a~review of some of these HCPs and the
\mbox{corresponding} physical systems. A~key common feature emerges
from the experiments on all these systems: an interesting coarsening
phenomena occurs which leads to a~scale-invariant morphology for large
times, namely the system is described by a~single (time-dependent)
length and the distribution approaches a~scaling form. Several models,
even very simple ones, have been proposed by physicists in order to
capture and explain such intriguing behavior and in many cases these
models turn out to be HCPs; see, for example,
\cite{P,DBG,DGY1,DGY2,SE,BDG}. Supported by computer simulations and
under the key assumption of a~well-defined limiting behavior under
suitable rescaling, physicists have derived for these HCPs in the mean
field approximation some nontrivial limiting distributions for the
relevant quantities and noticed that these distributions display
a~certain degree of universality. The results we obtained in
\cite{FMRT0} prove and generalize the findings of physicists. However
the analysis in \cite{FMRT0} does not cover some cases of interests for
physics which involve triple merging, for example, the HCP which has
introduced in \cite{BDG} to model Ising at zero temperature; see
Section~\ref{ising}. These models are instead covered by the present
study which explains why the limiting distributions of several models,
although different, have a~similar structure.

The analysis in \cite{FMRT0} is based on a~robust combinatorial study
of the coalescence inside a~given epoch, which becomes extremely hard
in the present setting. Hence, here we have followed a~different route
inspired by the approach of \cite{SE}. In particular, we start with the
infinitesimal generator of the one-epoch coalescence, giving a~complete
characterization of its form and domain (Theorem
\ref{amico_marco_zero}). It is well known that this allow us to
characterize the time evolution of the expectation of regular
observables in terms of an ordinary differential equation. Applied to
the domain length and the position of the leftmost point (if any), this
method leads to recursive equations between the Laplace transforms of
the involved quantities at the beginning and the end of each epoch, and
therefore at the beginning of two consecutive epochs (Theorems
\ref{differisco} and \ref{differiscobis}).

The study of the Markov generator for stochastic processes whose state
at a~given time is described by a~simple point process is not trivial.
Two fundamental contributions are given by \cite{Pr} and \cite{GK},
where spatial birth and death processes are obtained as solutions of
stochastic equations (in \cite{Pr} only finite populations are
considered, while in \cite{GK} the analysis is extended to locally
finite populations). Here, we have introduced a~lattice structure
(which is somehow artificial from a~geometric point of view) that
strongly simplifies the analysis of the Markov generator, and in
particular allows us to use the standard methods described in \cite{L}
(another possible route could have been to adapt the method developed
in \cite{GK}). Such a~discretization requires some very special care,
because of the use the vague topology on the space $\mathcal{N}$ of
locally finite subset of $\mathbb{R}$.

Once obtained the above mentioned system of recursive equations between
Laplace transforms, we have generalized the transformation introduced
in \cite{FMRT0}, Section~5, which in some sense linearizes the system
and allows us to analyze the recursive identities and obtain the limit
behavior (Theorems \ref{teo2} and \ref{teo3}). The resulting
transformation is now a~more abstract object and can therefore be
applied to a~larger class of models.

Finally we stress that the heuristic technique developed by physicists
(see~\cite{BDG}) to derive the limiting distribution (under the
assumption of the existence of a~limiting behavior) is restricted to
models with $d^{(n)}=n$, and it becomes meaningless also at heuristic
level if the ratio $d^{(n)}/ d^{(n+1)}$ does not converge to $1$ as $n$
goes to $\infty$.
% In this class a~particular interesting case is
%represented by the model in which (roughly) the smallest interval
%merges with its two neighbors.
Under the same hypothesis of \cite{BDG}, namely $d^{(n)}=n$ and via the
mean field approximation, in \cite{CP} the authors proposed a~time
evolution equation which should describe the domain size distribution
when the time variable $t$ is a~continuous approximation of the
discrete label $n$ of the epochs and one forgets how much time elapses
between and during the merging events. This equation has been
rigorously analyzed in \cite{CP} and \cite{GM}, and in the latter work
a limiting self-similar profile for this equation has been proved. In
this special case, a~transformation similar to the one presented in
more generality in \cite{FMRT0}, and here, has been used.
% (see also
%techniques.

%s2 #&#
\section{Model and results}\label{monello}

In this section we fix some notation and give our main results. We
first introduce the simple point processes we are interested in
(standard references are \cite{DV,FKAS}). Then we define the process
called one-epoch coalescence process (in short OCP) and the
hierarchical coalescence process (HCP). Finally we provide some
examples of HCPs coming from the physics literature.
%In this section, we present the hierarchical coalescence process
%(HCP) and our main results in the renewal case (generalizations are
%described in Appendix \ref{generale}). The
% HCP is given by a~sequence of one-epoch
%coalescence processes (OCP) suitably linked, each one corresponding
%to an annihilation dynamics on a~simple point process (SPP).

%s2.1 #&#
\subsection{Simple point processes (SPP)}\label{sec_SPP}

We denote by $\mathcal{N}$ the family of locally finite subsets $\xi
\subset \mathbb{R}$. $\mathcal{N}$ is a~measurable space endowed with
the \mbox{$\sigma$-}algebra of measurable subsets generated by
\[
\bigl\{ \xi\in\mathcal{N}\dvtx  |\xi\cap A_1|=n_1, \ldots, |\xi
\cap A_k |=n_k \bigr\},
\]
$A_1, \ldots, A_k$ being bounded Borel sets in $\mathbb{R}$ and $n_1,
\ldots, n_k \in\mathbb{N}$. We recall that any probability measure on
the measurable space $\mathcal{N}$ defines a~simple point process
(SPP).

We call \textit{domains} the intervals $[x,x']$ between
nearest--neighbor points $x,x'$ in $\xi\cup\{-\infty, +\infty\} $. Note
that the existence of the domain $[-\infty,x']$ corresponds to the fact
that $\xi$ is bounded from the left and its leftmost point is given by
$x'$. A similar consideration holds for $[x, \infty]$. Points of $\xi$
are also called \textit{domain separation points}. Given a~point
$x\in\mathbb{R}$, we define
\[
d_x^\ell:= \inf\{ t > 0\dvtx  x-t \in\xi\},\qquad
d_x ^r:= \inf\{ t > 0\dvtx  x+t \in\xi\}
\]
with the convention that the infimum of the empty set is $\infty$.
Note that if $x \in\xi$, then $d_x^\ell$ ($d_x^r$) is simply the
length of the domain to the left (right) of~$x$.

In what follows $\mathbb{N}$ ($\mathbb{N}_+$) will denote the set of
nonnegative (positive) integers.

%de2.1 #&#
\begin{definition}
\begin{longlist}[(iii)]
\item[(i)] We say that a~SPP $\xi$ is \textit{left-bounded} if it has
    a~leftmost point and has infinite cardinality.
%In this case we enumerate the points of $\xi$ as $\{x_k;  k \in\bbN

\item[(ii)] We say that a~SPP $\xi$ is $\mathbb{Z}$-\textit{stationary}
    if $\xi \subset\mathbb{Z}$ and its law $\mathcal{Q}$ is
    invariant by $\mathbb{Z}$-translations, that is, if for any $x
    \in\mathbb{Z}$ the random set $\xi-x$ has law
    $\mathcal{Q}$.\vadjust{\goodbreak}

\item[(iii)] We say that a~SPP $\xi$ is \textit{stationary} if its law
    $\mathcal{Q}$ is invariant under \mbox{$\mathbb{R}$-}translations,
    that is,  if for any $x \in\mathbb{R}$ the random set $\xi-x$
    has law $\mathcal{Q}$.
\end{longlist}
\end{definition}

If $\xi$ is $\mathbb{Z}$-stationary or stationary, then a.s. the
following dichotomy holds \cite{FKAS}: $\xi $ is unbounded from the
left and from the right or $\xi$ is empty. In the sequel we will always
assume the first alternative to hold a.s. and we will write $
\xi=\{x_k\dvtx  k \in \mathbb{Z}\}$ with the rules: $x_0 \leq0 < x_1$
and $x_k < x_{k+1}$ for all $k \in\mathbb{Z}$. In the case of
a~left-bounded SPP, we enumerate the points of $\xi$ as $\{x_k\dvtx
k \in\mathbb{N}\}$ in increasing order.

We now describe the main classes of SPP's we are interested in.
%
%de2.2 #&#
\begin{definition}
Let $\nu$ and $\mu$ be probability measures on $\mathbb{R}$ and
$(0,\infty)$, respectively. Let $\xi$ be a~SPP with law $\mathcal{Q}$.
\begin{itemize}
\item We say that $\xi$ is a~\emph{renewal SPP containing the
    origin and with interval law $\mu$}, and write $\mathcal{Q}=
    \operatorname{Ren}(\mu | 0)$, if:
\begin{longlist}[(ii)]
\item[(i)] $0 \in\xi$;

\item[(ii)] $\xi$ is unbounded from the left and from the right
    and, labeling the points in increasing order with $x_0=0$, the
    random variables $d_k= x_k-x_{k-1}$, $k \in\mathbb{Z}$, are
    i.i.d. with common law $\mu$.
\end{longlist}

\item We say that $\xi$ is a~\emph{right renewal} SPP with first
    point law $\nu$ and interval law $\mu$, and write $\mathcal{Q}=
    \operatorname{Ren}(\nu, \mu)$, if:
\begin{longlist}[(iii)]
\item[(i)] $\xi= \{x_k,  k \in\mathbb{N}\} $ is a~left-bounded SPP;

\item[(ii)] the first point $x_0$ has law $\nu$;

\item[(iii)] $ d_k = x_k -x_{k-1}$ ($k \in\mathbb{N}_+$) has
    law $\mu$;

\item[(iv)] the random variables $x_0, \{d_k\}_{k \in
    \mathbb{N}_+}$ are independent.
\end{longlist}

\item If $\mu$ has finite mean, we say that $\xi$ is a
    \textit{stationary renewal SPP with interval law~$\mu$}, and
    write $\mathcal{Q}=\operatorname{Ren} (\mu)$, if:
\begin{longlist}[(ii)]
\item[(i)] $\xi$ is a~stationary SPP with finite intensity and
    $\xi$ is nonempty a.s.;

\item[(ii)] the random variables $d_k= x_k-x_{k-1}$, $k \in
    \mathbb{Z}$, are i.i.d. with common law $\mu$ w.r.t. the
    Palm distribution associated to $\mathcal{Q}$.
\end{longlist}

\item If $\mu$ has support on $\mathbb{N}_+$ and has finite mean,
    we say that $\xi$ is a~$\mathbb{Z}$-\textit{stationary renewal
    SPP with interval law $\mu$}, and write
    $\mathcal{Q}=\operatorname{Ren} _{\mathbb{Z}} (\mu)$, if:
\begin{longlist}[(ii)]
\item[(i)] $\xi$ is $\mathbb{Z}$-stationary and a.s. nonempty;

\item[(ii)] w.r.t. the conditional probability $\mathcal{Q}(\cdot
    |0 \in\xi)$ the random variables $d_k= x_k-x_{k-1}$, $k
    \in\mathbb{Z}$, are i.i.d. with common law $\mu$.
\end{longlist}
\end{itemize}
\end{definition}

We recall that the intensity $\lambda_{\mathcal{Q}}$ of a~stationary
SPP with law $\mathcal{Q}$ is defined as the expectation
$\lambda_{\mathcal{Q}}:= \mathbb{E}_{\mathcal{Q}}  ( |\xi \cap[0,1]|
)$. A ($\mathbb{Z}$-)stationary renewal SPP with interval law $\mu$
having infinite mean cannot exist (see Proposition 4.2.I in \cite{DV}
and Appendix C in \cite{FMRT0}). Also recall that, given a~stationary
SPP with law $\mathcal{Q}$, its Palm distribution is formally defined
by the conditional probability measure $\mathcal{Q}( \cdot|0 \in\xi)$.
Since for a~stationary process the event $\{0 \in\xi\}$ has zero
probability, a~precise definition requires more care \cite{FKAS,DV}.
For those readers who are not familiar with Palm distributions, it is
sufficient to know that as discussed after Theorem 1.3.4 in
\cite{FKAS}, $\mathcal{Q}=\operatorname{Ren} (\mu)$ if and only if the
following holds: the random variables $d_k=x_k-x_{k-1}$, $k\neq1$ are
i.i.d. with law $\mu$ and are independent from the random vector
$(x_0,x_1)$, which satisfies
%
%e1 #&#
\begin{eqnarray}
\mathcal{Q}( -x_0> u, x_1> v) &=& \lambda_\mathcal{Q}
\int_{u+v}^\infty \bigl( 1- F(t) \bigr) \,dt,\qquad F(t):= \mu
\bigl((0,t] \bigr),
\nonumber
\\[-8pt]
\eqntext{u,v>0.}
\end{eqnarray}

%s2.2 #&#
\subsection{The one-epoch coalescence process (OCP)} \label{sec_OCP} This process depends on two constants
$0<d_{\min}<d_{\max}$ and on nonnegative bounded continuous functions
$\lambda_\ell$, $\lambda_r$, $\lambda_a $ defined on $[d_{\min},
\infty]$ which, with $ \lambda(d):= \lambda_\ell(d)+ \lambda_r
(d)+\lambda_a (d)$, satisfy the following assumptions:
\begin{longlist}[(A1)]
\item[(A1)] $\lambda(d) >0$ if and only if $d \in[d_{\min},
    d_{\max})$;

\item[(A2)] if $d,d' \geq d_{\min}$, then $d+d'\geq d_{\max}$.
\end{longlist}
Trivially, \textup{(A2)} is equivalent to the bound $2d_{\min}\geq
d_{\max}$.

The admissible starting configurations for the OCP belong to the subset
$\mathcal{N}(d_{\min} )$ given by the configurations $\xi
\in\mathcal{N}$ having only domains of length not smaller than
$d_{\min}$, that is,
%
%e2 #&#
\begin{equation}
\label{scudetto} \mathcal{N}(d_{\min})= \bigl\{\xi\in \mathcal{N}\dvtx
d_x^{ \ell} \geq d_{\min}, d_x^{  r}
\geq d_{\min}\ \forall x \in\xi\bigr\}.
\end{equation}

Then the stochastic evolution of the OCP is given by a~jump dynamics
with c\`{a}dl\`{a}g paths $\{ \xi(t) \}_{t \geq0}$ in the Skorohod
space $D ([0,\infty), \mathcal{N}(d_{\min})  )$; cf. \cite{B}. Roughly
speaking, the dynamics is the following. Each domain $\Delta$ of length
$d$ waits an exponential time with parameter $\lambda(d)$, and
afterward exactly one of the following annihilations takes place: the
left extreme of $\Delta$ is erased with probability
$\lambda_\ell(d)/\lambda(d)$; the right extreme of $\Delta$ is erased
with probability $\lambda_r(d)/\lambda(d)$; both the extremes of
$\Delta$ are erased with probability $\lambda_a(d)/\lambda(d)$. We say
that the domain $\Delta$ incorporates its left domain, its right
domain, both its neighboring domains, respectively. In
Section~\ref{universal} we present a~full construction of all OCPs,
varying the initial configuration, on the same probability space
(\emph{universal coupling}).

Note that assumptions \textup{(A1)} and \textup{(A2)} on the
coalescence rates imply that any domain which has been generated by
a~coalescence event is not active; that is, it cannot incorporate other
domains. This assumption comes from several models of physical interest
(see Section~\ref{physics}) and plays a~fundamental role in our
analysis.
% will be used when we establish the
%recursive identities between the Laplace transforms of the interval
%law at the beginning and at the end of the process (see Theorem
%origin of the simple formula \eqref{davide} which governs the
%evolution of the distribution of the active domains.

%re2.3 #&#
\begin{remark}
Note that $\lambda_\ell$ and $\lambda_r$ correspond to $\lambda_r $ and
$\lambda_\ell$ in \cite{FMRT0}, respectively. The dynamics here is
thought as coalescence of domains while in \cite{FMRT0} is thought as
annihilation of points. In particular, in \cite{FMRT0}
$\lambda_\ell(d)$ corresponds to the rate at which a~point $x$ is
erased by the effect of its \emph{left} domain $[x-d,x]$, while here
$\lambda_\ell(d)$ corresponds to the rate at which a~domain of length
$d$ merges with its left domain and similarly for $\lambda_r(d)$.

The case $\lambda_a \equiv0$ has been treated in \cite{FMRT0} without
the additional assumption that $\lambda_\ell$, $\lambda_r$ are
continuous functions. This assumption will be crucial in the
investigation of the Markov generator; see Section~\ref{infinitesimi}.
\end{remark}

Formally, the Markov generator of
the OCP is given by
%
%e3 #&#
\begin{eqnarray}
\label{santi} \mathcal{L}f(\xi) &=& \mathop{\sum_{[x,x+d]}}_{\mathrm{domain\ in\ }\xi}
\bigl\{ \lambda_\ell(d) \bigl[ f \bigl(\xi\setminus\{x\} \bigr)- f(
\xi) \bigr]\nonumber
\\[-4pt]
&&\hspace*{44pt}{} + \lambda_r (d) \bigl[ f \bigl(\xi\setminus\{x+d\} \bigr)-
f(\xi) \bigr]
\\
&&\hspace*{44pt}{}+ \lambda_a (d) \bigl[ f \bigl(\xi\setminus\{x, x+d\}
\bigr)- f(\xi) \bigr] \bigr\}.\nonumber
\end{eqnarray}
A precise description of the Markov generator $\mathcal{L}$ is
given %\club\club
below while its full rigorous analysis is postponed to Section~\ref{infinitesimi} for clarity of exposition.

We will write $\mathbb{P}_{\mathcal{Q}}$ for the law on $D ([0,\infty),
\mathcal{N}(d_{\min})  )$ of the OCP with initial law $\mathcal{Q}$ on
$\mathcal{N} (d_{\min})$ and $\mathcal{Q}_t$ for its marginal at time
$t$. With a~slight abuse of notation, for any configuration
$\zeta\in\mathcal{N}(d_{\min})$ we also let $\mathbb{P}_{\zeta}$ be
$\mathbb{P}_{\mathcal{Q}}$ when $\mathcal{Q}$ is concentrated in the
configuration $\zeta$.

Since the OCP is an annihilation process, points can only disappear.
Furthermore, assumptions \textup{(A1)}~and~\textup{(A2)} guarantee that
the process converges to a~limiting configuration. One can easily prove
the following lemma already stated in \cite{FMRT0} in a~less general
setting (details are left to the reader).

% Lemma 2.12 and Theorem
%2.13 of \cite{FMRT0} remain valid here by slightly changing their
%proofs (note that in Section \ref{universal} we give a~detailed
%description of the universal coupling which is at the base of the
%proof of Lemma 2.12 and Theorem 2.13 in \cite{FMRT0}). Hence, we
%recall these results omitting the proofs:
%
%le2.4 #&#
\begin{Lemma} \label{volarelontano}
For any given initial condition $\xi\in\mathcal{N}(d_{\min})$ the
following holds:
\begin{longlist}[(ii)]
\item[(i)] $\xi(t) \subset\xi(s)$ if $s \leq t$;
% the configuration $\xi(t) $ is constant on
%bounded intervals eventually in $t$,

\item[(ii)] there exists a~unique element $\xi(\infty)$ in
    $\mathcal{N}(d_{\max})$ such that $\xi(t)\cap I= \xi(\infty)\cap
    I$ for all large enough $t$ (depending on $I$) and all bounded
    intervals~$I$.
\end{longlist}
\end{Lemma}

The next result is a~simple generalization of \cite{FMRT0},
Theorem~2.13. It states that if the process starts with some right
renewal (resp., stationary, $\mathbb{Z}$-stationary, etc.) simple point
process $\xi$, then at any later time $t$ the process $\xi(t)$ is still
of the same type. This preservation of the renewal is a~consequence of
the following key property of the dynamics (called ``separation
effect'' in \cite{FMRT0}, Lemma 3.1): if at time~$t$ a~point $x$
survives this allows to decouple events that occur for $s\leq t$
spatially on its left and on its right. More precisely if $\mathcal{A}$
is in the \mbox{$\sigma$-}algebra generated by\vadjust{\goodbreak}
$\{\xi(s)\cap(-\infty,x)\}_{s\leq t}$, and $\mathcal{B}$ is in the
\mbox{$\sigma$-}algebra generated by $\{\xi(s)\cap(x,\infty)\}_{s\leq
t}$, then
\[
\mathbb{P}_{\zeta}\bigl(\mathcal{A}\cap\mathcal{B}\cap\bigl\{x\in\xi(t)
\bigr\} \bigr)=\mathbb{P}_{\zeta\cap(-\infty,x]}\bigl(\mathcal{A}\cap\bigl\{x\in\xi (t)
\bigr\}\bigr)\mathbb{P}_{\zeta\cap[x,\infty)}\bigl(\mathcal{B}\cap\bigl\{x\in\xi (t)
\bigr\}\bigr).
\]
The proof of this equality goes as in \cite{FMRT0}, Lemma 3.1, using
now the universal coupling described in Section~\ref{universal}. By
using this result, the proof of Lemma \ref{vecchioteo1} follows exactly
the same lines \cite{FMRT0}, Theorem 2.13.

%le2.5 #&#
\begin{Lemma} \label{vecchioteo1}
Let $\nu, \mu$ be two probability measures on $\mathbb{R}$ and
$[d_{\min}, \infty)$, respectively. Then, for all $t\in[0,\infty]$
there exist probability measures $\nu_t,\mu_t$ on $\mathbb{R}$ and
$[d_{\min}, \infty)$, respectively, such that $\nu_0=\nu$, $\mu_0=\mu$
and:
\begin{longlist}[(iii)]
\item[(i)] if $\mathcal{Q}=\operatorname{Ren}(\nu,\mu)$, then
    $\mathcal {Q}_t=\operatorname{Ren} (\nu_t, \mu_t)$;

\item[(ii)] if $\mathcal{Q}= \operatorname{Ren} (\mu)$, then
    $\mathcal{Q}_t= \operatorname{Ren} (\mu_t)$;

\item[(iii)] if $\mathcal{Q}= \operatorname{Ren}_\mathbb{Z}(\mu)$, then
    $\mathcal{Q}_t=\operatorname{Ren}_\mathbb{Z} (\mu_t)$;

\item[(iv)] if $\mathcal{Q}= \operatorname{Ren} (\delta_0, \mu)$, then
    $\mathcal{Q}_t(\cdot|0\in\xi)=
    \operatorname{Ren}(\delta_0,\mu_t)$;

\item[(v)] $\lim_{t\to\infty}\nu_t=\nu_\infty$ and $\lim_{t\to
    \infty}\mu_t=\mu_\infty$ weakly.
\end{longlist}
\end{Lemma}

Thanks to the previous results $\xi(\infty)$, $\mu_\infty$ and
$\nu_\infty$ are well defined. In fact there exists a~recursive
identity between the Laplace transform of the interval law, and of the
first point law, at time $t=0$ and at time $t=\infty$. These
identities, stated in Theorems \ref{differisco}~and~\ref{differiscobis}
below, will be the keystones of the analysis of the asymptotic of the
hierarchical coalescence process.

Given a~probability measure $\mu$ on $[d_{\min}, \infty)$, let $\mu_t$
be as in Lemma \ref{vecchioteo1}. Then, for $s \in\mathbb{R}_+$, define
\[
G_t (s) = \int e^{-s x} \mu_t (dx),\qquad
H_t (s) = \int_{[d_{\min}, d_{\max} )} e^{-sx}
\mu_t (dx).
\]

%th2.6 #&#
\begin{Theorem}[(Recursive identities for the interval law)]\label{differisco}
For any $s \in\mathbb{R}_+$, the functions $[0,\infty)\ni t \mapsto G_t(s),
H_t (s) $ are differentiable and satisfy
%
%e4 #&#
%e5 #&#
\begin{eqnarray}
\partial_t H_t(s)&=& -\int\mu_t (dx)
\lambda(x) e^{-s x}, \label{davide}
\\
\partial_t \bigl[G_t(s)-H_t(s) \bigr] &=&
%& \qquad\qquad
G_t(s)\int \mu_t(dx) (
\lambda_\ell+\lambda_r) (x) e^{-sx}
\nonumber\\[-8pt]\label{orecchiettebis} \\[-8pt]
&&{}  + G_t(s)^2 \int \mu_t(dx)\lambda_a(x)
e^{-sx}.\nonumber
\end{eqnarray}
In particular, it holds:
\begin{longlist}[(ii)]
\item[(i)] If $\lambda_a\equiv0$, then $\partial_t G_t (s)= \partial_t
    H_t(s) (1- G_t(s))$.
% \partial_t\bigl[ G_t (s)-H_t(s)\bigr]= -\partial_t H_t(s) \cdot
%G_t(s),\label{speranza1}
Hence,
%
%e6 #&#
%e7 #&#
\begin{eqnarray}
1-G_t(s) &=& \bigl(1-G_0(s)\bigr) e ^{H_0(s) - H_t(s) },
\qquad t\in\mathbb {R}_+, \label{caviale}
\\
1-G_\infty(s)&=& \bigl(1-G_0(s)\bigr) e
^{H_0(s) }. \label{uovone}
\end{eqnarray}

\item[(ii)] If $\lambda_\ell+\lambda_r \equiv\gamma\lambda_a$ for some
    $\gamma\geq0$, then
\[
\partial_t G_t (s)= \partial_t H_t(s)  \biggl( 1 -
    \frac{G_t(s)(\gamma+ G_t(s))}{1+\gamma} \biggr).
\]
% -\frac{\g}{1+\g} \partial_t H_t(s)\cdot G_t(s)-\frac{1}{1+\g}
Hence, for $s>0$ it holds
%
%e8 #&#
%e9 #&#
\begin{eqnarray}
&&  e^{-((\gamma+2)/(\gamma+1)) H_t(s)} \frac{\gamma+1+G_t(s)
}{1-G_t(s)}
\nonumber\\[-8pt]\label{salsa} \\[-8pt]
&&\qquad = e^{-((\gamma+2)/(\gamma+1)) H_0(s)} \frac{\gamma+1+G_0(s) }{1-G_0(s)},\qquad
t\in\mathbb{R}_+,\nonumber
\\
&& \frac{\gamma+1+G_\infty(s) }{1-G_\infty(s)} = e^{-((\gamma
+2)/(\gamma+1)) H_0(s)} \frac{\gamma+1+G_0(s) }{1-G_0(s)}. \label{salsona}
\end{eqnarray}
\end{longlist}
\end{Theorem}

In the above theorem, as in the rest of the paper, differentiability at
$t=0$ for a~function on $[0,\infty)$ means differentiability from the
right.
%
%re2.7 #&#
\begin{remark}\label{pasquetta}
The restriction to the above cases \textup{(i)} and \textup{(ii)} is
technical and motivated by the following. Set
\[
a_t(s):= \int\mu_t(dx) (\lambda_\ell+
\lambda_r) (x) e^{-sx} \quad \mbox{and} \quad
b_t(s):= \int\mu_t(dx) \lambda_a(x)
e^{-sx}.
\]
Thanks to \eqref{davide}, equation \eqref{orecchiettebis} can be
rewritten as
%
%e10 #&#
\begin{equation}
\label{serioso} \partial_t G_t (s)= -a_t(s) -
b_t(s) + a_t (s) G_t (s)+ b_t(s)
G_t (s)^2.
\end{equation}
Fixing functions $A_t (s)$ and $B_t(s)$ such that $\partial_t
A_t(s)=a_t(s)$ and $\partial_t B_t(s)=b_t(s)$, \eqref{serioso} leads to
%
%e11 #&#
\begin{equation}
\label{fulmini} \frac{ e^{A_t(s)+2 B_t(s)} }{1-G_t(s)}= \frac{ e^{A_0(s)+2B_0(s) }
}{1-G_0(s)} +\int_0
^t b_u(s) e^{A_u(s)+2 B_u(s) }\,du.
\end{equation}
In order to have a~recursive identity between $(G_0,H_0)$ and
$G_\infty$, one needs to find an explicit expression of the integral on
the right-hand side of \eqref{fulmini}. This can be achieved in cases
\textup{(i)} and \textup{(ii)} of Theorem \ref{differisco} by taking
$B_t(s)=b_t(s)=0$ and $A_t(s)=-H_t(s)$ in case \textup{(i)}, and by
taking $A_t(s)=\gamma B_t(s)$ and $A_t(s)+B_t(s)=-H_t (s)$ in
case~\textup{(ii)}.
%Trivially, one recovers \eqref{uovo} and \eqref{salsa}

Finally we point out that, since $\operatorname{arctanh}(x)=
\frac{1}{2} \ln \frac{1+x}{1-x}$ for $x \in(-1,1)$, \eqref{salsa} with
$\gamma=0$ can be rewritten in the more compact form
\[
-H_t + \operatorname{arctanh}G_t(s)= -H_0 +
\operatorname{arctanh} G_0(s).
\]
\end{remark}

The next result is concerned with the evolution of the first point law
$\nu_t$ when starting with a~SPP having law $\operatorname{Ren}(\nu,
\mu)$ (recall Lemma \ref{vecchioteo1}). First, we observe that if $\xi$
is a~SPP with law $\operatorname{Ren}( \delta_0, \mu)$, and $V$ is
a~random variable with law $\nu$ independent\vadjust{\goodbreak} from $\xi$, then the
translated random subset $\{x +V\dvtx  x \in\xi\}\subset\mathbb{R}$ is
a~SPP with law $\operatorname{Ren} (\nu, \mu)$. This simple observation
and the definition of the OCP, whose dynamics depends only on the
sequence of the domain lengths and not on the specific location of the
domains, allow us to conclude that $\nu_t$ is the convolution
%
%e12 #&#
\begin{equation}
\label{convoluto} \nu_t= \bar\nu_t \ast\nu,
\end{equation}
where $ \bar \nu_t$ denotes the evolution at time $t$ of the first
point law when starting from a~SPP having law
$\operatorname{Ren}(\delta_0, \mu)$. Hence, without loss we can
restrict our analysis to this case.

%th2.8 #&#
\begin{Theorem}[(Recursive identities for the first point law)]\label{differiscobis}
Assume that $\nu= \delta_0$. Then, for any $s \in\mathbb{R}_+$ the Laplace
transform
\[
[0,\infty)\ni t \mapsto L_t(s):= \int e^{-s x}
\nu_t (x) \in(0,1]
\]
is differentiable and satisfies
%
%e13 #&#
\begin{eqnarray}\label{bucatini}
\frac{\partial_t L_t(s)}{L_t(s)} &=& - \int \mu_t(dy) \bigl(
\lambda_\ell(y) + \lambda_a(y) \bigr) + \int
\mu_t(dy) \lambda_\ell(y)e^{-sy} %\right. \\
\nonumber\\[-8pt]\\[-8pt]
&&{} +G_t(s) \int \mu_t(dy) \lambda_a(y)e^{-sy}.\nonumber
\end{eqnarray}
In particular, it holds:
\begin{longlist}[(ii)]
\item[(i)] If $\lambda_a\equiv0$ and $\lambda_r \equiv\gamma\lambda
    _\ell$ for some\vspace*{1pt} constant $\gamma\geq0$, then it holds $\partial_t
    L_t (s)= \frac{L_t(s)}{1+\gamma}  ( \partial_t H_t (0) -
\partial_t H_t(s)  )$.
% \partial_t L_t (s)= \frac{L_t(s)}{1+\gamma} \left( \partial_t H_t
%(0) - \partial_t H_t (s) \right).
Hence
%
%e14 #&#
%e15 #&#
\begin{eqnarray}
\qquad L_t (s)&=& L_0(s)\exp \biggl\{\frac{-H_t(s) + H_t(0) + H_0(s) -
H_0(0)}{1+\gamma}
\biggr\},\qquad t\in\mathbb{R}_+, \label
{monte1}
\\
L_\infty(s)&=& L_0(s) \exp \biggl\{\frac{H_0(s)-H_0(0)}{1+\gamma
}
\biggr\}. \label{monte2}
\end{eqnarray}
If $\lambda_a \equiv0$ and $ \lambda_\ell\equiv0$, then trivially
$L_t(s)=L_0(s)$ for any $t \geq0$. %and any $s \geq0$.

\item[(ii)] If $\lambda_\ell\equiv0$ and $\lambda_r \equiv0$, then
    $\partial_t L_t (s)=L_t(s) (\partial_t H_t(0) - G_t(s) \partial_t
    H_t(s))$.
% \partial_t L_t (s)=L_t(s) G_t(s) \partial_t H_t(s).
% \]
Hence for $s>0$ it holds
%
%e16 #&#
%e17 #&#
\begin{eqnarray}
L_t(s)&=& L_0(s) \sqrt{ \frac{1-G_t^2(s)}{1-G_0^2(s)} }
e^{H_t(0)-H_0(0)},\qquad t\in\mathbb{R}_+,\label{gransasso1}
\\
L_\infty(s)&=& L_0(s)\sqrt{\frac{1-G_\infty^2(s)}{1-G_0^2(s)}}
e^{-H_0(0)}. \label{gransasso2}
\end{eqnarray}
\end{longlist}
\end{Theorem}
We point out that cases \textup{(i)} and \textup{(ii)} of Theorem
\ref{differiscobis} are included into (but not equal to) cases
\textup{(i)} and \textup{(ii)} of Theorem \ref{differisco}.\vadjust{\goodbreak}

%In the above theorem, if for $s \in[0,\infty)$ it holds
%$L_0(s)= \infty$ then we understand that $L_t (s) = \infty$ for all
%$t \geq0$. If $L_0(s) < \infty$, then we understand that $L_t (s) <
%and satisfies the differential identities above.
%%\eqref{agognata1} and \eqref{agognata2} hold.
%Differentiability at $t=0$ means from the
%right.

The previous results are based on our analysis of the Markov generator
$\mathcal{L}$ of the OCP. In general, the expected value at time $t$ of
a regular observable (i.e., a~test function in the domain of
$\mathcal{L}$) evolves according to an ordinary differential equation
that we describe below. We first fix some notation. Given $k
\in\mathbb{Z}$ we set
%
%e18 #&#
\begin{equation}
I_k:= \cases{ \bigl[k d_{\min}, (k+1) d_{\min} \bigr),
&\quad if $k \geq1$,
\vspace*{3pt}\cr
\bigl(k d_{\min}, (k+1) d_{\min}
\bigr), &\quad if $k=0$,
\vspace*{3pt}\cr
\bigl(k d_{\min}, (k+1) d_{\min} \bigr],
&\quad if $k\leq-1$.}
\end{equation}
%
% By definition of $\cN( d_{\min} )$ (see \eqref{scudetto}), each
%configuration $\xi
% \in\cN( d_{\min} )$ has at most one point in the interval $I_k$.
%Then,
Given $\xi\in\mathcal{N}( d_{\min} )$, we set for $k \in \mathbb{Z}$
and $k<k'$ in $\mathbb{Z}$,
\begin{eqnarray*}
 \xi^k&:=& \xi\setminus I_k,\qquad \nabla_k f (\xi):= f\bigl( \xi^k\bigr)-f(\xi),
\\
\xi^{k, k'}&:=& \xi\setminus(I_k \cup I_{k'}),
\qquad \nabla_{k,k'} f(\xi):= f\bigl(\xi^{k,k'} \bigr) - f (\xi).
\end{eqnarray*}
We define
\[
\mathcal{R}:= \mathbb{Z}\cup\bigl\{ \bigl(k,k'\bigr)\dvtx
k' \in \bigl\{k+1,\ldots,k+\lceil d_{\max}/d_{\min}
\rceil\bigr\}, k,k'\mbox{ in } \mathbb{Z}\bigr\},
\]
$\lceil a~\rceil$ being the smallest integer $n\geq a$. We consider the
space $\mathcal{N}( d_{\min})$ endowed of the vague topology (see
Section~\ref{dugundi}), making it a~compact space. We write
$\mathbb{B}$ for the Banach space of all bounded continuous functions
$f\dvtx\break  \mathcal {N}(d_{\min}) \mapsto\mathbb{R}$ endowed with the
uniform norm that we denote by $\| \cdot\|$. Also, and for later
purpose, we let $\mathbb{B}_{\mathrm{loc}}$ be the set of functions $f
\in\mathbb{B}$ that are local, that is, such that there exists
a~bounded interval $I \subset\mathbb{R}$ with $f(\xi)=f(\xi\cap I)$ for
all $\xi\in\mathcal{N}(d_{\min} )$. Then, similar to the analysis of
interacting particle systems \cite{L}, we define
\[
\Delta_f (r):= \sup_{ \xi\in\mathcal{N}(d_{\min}) } \bigl|
\nabla_r f (\xi) \bigr|, \qquad f \in\mathbb{B}, r \in\mathcal{R}
\]
%
%Note that $\D_f (r)$ is finite (since $\cN(d_{\min})$ is
%compact). Finally,
and we introduce the subset $\mathbb{D}$ of $ \mathbb{B}$ as
%
%e19 #&#
\begin{equation}
\label{ddd} \mathbb{D}:= \biggl\{ f \in\mathbb{B}\dvtx  |\!|\!|f|\!|\!|:= \sum
_{r \in\mathcal{R}} \Delta_f (r) < \infty \biggr\}.
\end{equation}
Observe that $\mathbb{B}_{\mathrm{loc}} \subset\mathbb{D}$. The
following result characterizes completely the Markov generator of the
OCP:

%th2.9 #&#
\begin{Theorem}\label{amico_marco_zero}
The subspaces $\mathbb{B}_{\mathrm{loc}}$ and $\mathbb{D}$ are a~core
of the Markov generator~$\mathcal{L}$, that is, $\mathcal{L}$ is the
closure of the operator obtained by restriction to
$\mathbb{B}_{\mathrm{loc}}$ or to $\mathbb{D}$. Moreover, if $f \in
\mathbb{D}$, $\mathcal{L}f(\xi)$ equals the absolutely convergent
series on the right-hand side of \eqref{santi}.
\end{Theorem}
The proof is given in Section~\ref{infinitesimi}. %\club\club
Although this analysis represents our starting point, we prefer to
postpone it to the end since rather technical. As a~consequence of the
above theorem and standard theory of Markov generators, we get the
following characterization of the time evolution of expected
observables:
%
%co2.10 #&#
\begin{Corollary}\label{corri}
Given $f \in\mathbb{D}$, the map $f(t,\xi):=\mathbb{E}_\xi [
f(\xi_t)  ]$
(the expectation of $f$ for the OCP at time $t$ starting from $\xi$) is
differentiable in $t$ as function in~$\mathbb{B}$ and moreover $ \frac{d}{dt}
f(t, \cdot) = \mathcal{L}f$.
\end{Corollary}

%s2.3 #&#
\subsection{The hierarchical coalescence process}\label{sec_HCP}

We can now introduce the \emph{hierarchical coalescence process} (in
short HCP). The dynamics depends on a~strictly increasing sequence of
positive numbers $\{d^{(n)}\}_{n\ge1}$ and a~family of bounded
continuous functions $\lambda_\ell^{(n)}$, $\lambda_r^{(n)}$,
$\lambda_a^{(n)}\dvtx [d^{(n)}, \infty]\rightarrow [0,A_n]$, $n \geq1$.
Without loss of generality, at cost of a~length rescaling, we may
assume
%
%e20 #&#
\begin{equation}
d^{(1)}=1.
\end{equation}
We set
$\lambda^{(n)}:= \lambda^{(n)}_\ell+ \lambda^{(n)}_r+ \lambda
^{(n)}_a$ and we assume:
\begin{longlist}[A3]
\item[(A1)] for any $n \in\mathbb{N}_+$, $\lambda^{(n)} (d) >0$ if and
    only if $d \in[d^{(n)}, d^{(n+1)})$;

\item[(A2)] for any $n \in\mathbb{N}_+$, if $d,d'\geq d^{(n)}$,
    then $d+d'\geq d^{(n+1)}$ (i.e., $2d^{(n)} \geq d^{(n+1)}$);

\item[(A3)] $\lim_{n\to\infty}\,d^{(n)}=\infty$.
\end{longlist}
For example, one could take $d^{(n)}=n$ or $d^{(n)} = a^{n-1}$ with
$a \in(1,2]$.

The HCP is then given by a~sequence of one-epoch coalescence processes,
suitably linked. More precisely, at the beginning of the first epoch
one starts with a~SPP with support on $\mathcal{N}( d^{(1)})=
\mathcal{N}(1) $. Then the stochastic evolution of the HCP is described
by the sequence of random paths $\{ \xi^{(n) } (\cdot)\}_{n\ge1}$,
where each $\xi^{(n)}$ is the random trajectory of the OCP with rates
$\lambda^{(n)}_\ell$, $\lambda^{(n)}_r$, $\lambda^{(n)}_a$, active
domain lengths $d^{(n)}_{\min}=d^{(n)},  d^{(n)}_{\max}=d^{(n+1)}$ and
initial condition $\xi^{(n)}(0)=\xi^{(n-1)} (\infty), n\geq2$.
Informally we refer to $\xi^{(n)}$ as describing the evolution in the
$n$th-epoch. Note that, by Lemma \ref{volarelontano}, one can prove
recursively that at the end of the $n$th-epoch the random configuration
$\xi^{(n)}(\infty)$ belongs to $\mathcal{N}( d^{(n+1)} )$, and hence it
is an admissible starting configuration for the OCP associated to the
$(n+1)$th-epoch.

Lemma \ref{vecchioteo1} gives us information on the evolution and its
asymptotics inside each epoch when the initial condition is a~SPP of
the renewal type. If, for example, the initial distribution $\mathcal
{Q}$ for the first epoch is $\operatorname{Ren}(\nu,\mu)$, where $\mu$
has support on $[d^{(1)}, \infty)=[1, \infty)$, we can use Lemma
\ref{vecchioteo1} together with the link $\xi^{(n+1)}(0)=\xi^{(n)}
(\infty)$ between two consecutive epochs to recursively define the
measures $\mu^{(n)},\nu^{(n)}$ by
%
%e21 #&#
\begin{eqnarray}
\mu^{(n+1)}&:=&\mu_\infty^{(n)},\qquad
\mu^{(1)}:=\mu,
\nonumber\\[-8pt]\label{eq:1}  \\[-8pt]
\nu^{(n+1)}&:=&\nu_\infty^{(n)},\qquad\nu^{(1)}:=\nu.\nonumber
\end{eqnarray}
%
%Note that each $\mu^{(n)} $ (the interval law at the beginning of
%the epoch $n$) has support in $[d^{(n)}, \infty)$, thus allowing to
%apply recursively Theorem \ref{vecchioteo1}.
With this position it is then natural to ask if, in some suitable
sense, the measures $\mu^{(n)}$, $\nu^{(n)}$ have a~well-defined
limiting behavior as $n\to \infty$. The affirmative answer is contained
in the following theorem, which is the core of the paper, for some
specific choice of transition rates. Before stating it we recall
a~useful result on the Laplace transform of probability measures on
$[1,\infty)$.

%le2.11 #&#
\begin{Lemma}[(\cite{FMRT0})] \label{prelim:mainth}
Let $\mu$ be a~probability measure on $[1, \infty)$, and let $g(s)$ be
its Laplace transform, that is, $g(s)= \int e^{-s x} \mu(dx), s
\in\mathbb{R}_+$.
\begin{longlist}[(ii)]
\item[(i)] If
%
%e22 #&#
\begin{equation}
\label{lim_arrosto} \lim_{s \downarrow0} -\frac{s
g'(s)}{1-g(s)}=c_0,
\end{equation}
then necessarily $0\le c_0\le1$.

\item[(ii)] The existence of limit \eqref{lim_arrosto} holds if:
\begin{longlist}[(a)]
\item[(a)] $\mu$ has finite mean and then $c_0=1$ or

\item[(b)] for some
    $\alpha\in(0,1)$ $\mu$ belongs to the domain of attraction of an
    $\alpha$-stable law or, more generally, $\mu ((x,\infty)
    )=x^{-\alpha} L(x)$ where $L(x)$ is a~slowly varying\footnote{A
    function $L$ is said to be slowly varying at infinity if for all
    $c>0$, $\lim_{x \to\infty} L(cx)/L(x)=1$.} function at
    $+\infty$, $\alpha\in[0,1]$ and in this case $c_0=\alpha$.
\end{longlist}
\end{longlist}
\end{Lemma}

The reader may find the proof in \cite{FMRT0}, Appendix A, together
with an example for which the limit \eqref{lim_arrosto} does not exist.

%th2.12 #&#
\begin{Theorem}\label{teo2}
Let $\nu, \mu$ be probability measures on $\mathbb{R}$ and $[1,
\infty)$, respectively.
% and let $g(s)$ be the Laplace transform of $\mu$.
Suppose that:
\begin{itemize}
\item the law $\mathcal{Q}$ of $\xi^{(1)}(0)$ is either
    $\mathcal{Q}=\operatorname{Ren} (\nu, \mu)$ or $\mathcal{Q}=
    \operatorname{Ren} (\mu)$ or
    $\mathcal{Q}=\operatorname{Ren}_\mathbb{Z}(\mu)$;

\item it holds \textup{(i)} $\lambda_a^{(n)} \equiv0$ for all
    $n\geq1$, or \textup{(ii)} $\lambda_\ell^{(n)} +
    \lambda_r^{(n)} \equiv\gamma\lambda_a^{(n)}$ for all $n\geq1 $
    and for some $\gamma\geq0$ independent from $n$;

\item the Laplace transform $g(s)$ of $\mu$
    satisfies~\eqref{lim_arrosto}.
\end{itemize}

For any $n\geq1$ let $X^{(n)}$ be a~random variable with law $\mu^{(n)
}$ defined in \eqref{eq:1}
so that $g(s):=\mathbb{E} [ e^{-s X^{(1)}} ]$. %and
%$Z^{(n)}:=X^{(n)}/d^{(n)}$.

Then the following holds:
\begin{itemize}
\item If $c_0=0$, then the rescaled variable
$Z^{(n)}=X^{(n)}/d^{(n)}$ weakly converges to the random
variable $Z^{(\infty)}_{0}= \infty$.

\item If $c_0 \in(0,1]$, then the rescaled variable
    $Z^{(n)}=X^{(n)}/d^{(n)}$ weakly converges to the random
    variable $ Z^{(\infty)}_{\kappa}$ with values in $[1, \infty)$,
    whose Laplace transform is given by
%
%e23 #&#
\begin{equation}
\label{macedonia} g^{(\infty)}_{\kappa}(s)= \mathcal{R} \biggl(\kappa\int
_1^\infty \frac{e^{-sx}}{x} \,dx \biggr), \qquad s
>0,
\end{equation}
where
%
%e24 #&#
\begin{equation}
\label{galattico} %
\cases{ \kappa:=c_0\quad\mbox{and}\quad
\mathcal{R}(x):= 1 - e^{-x},
\cr
\qquad \mbox{in case \textup{(i)}},
\vspace*{6pt}\cr
\kappa:=
\displaystyle\frac{\gamma+1}{\gamma+2} c_0\quad\mbox{and}\quad \mathcal{R}(x):=
\frac{\exp \{((\gamma+2)/(\gamma+1))x \} -1} {
\exp \{((\gamma+2)/(\gamma+1))x \}+ 1/(\gamma+1)},
\vspace*{3pt}\cr
\qquad \mbox{in case \textup{(ii)}.}}\hspace*{-20pt}
\end{equation}
\end{itemize}
%
%,\qquad\mbox{and}\qquad
%1 & \mbox{in case } (i),\\
\end{Theorem}
The proof of Theorem \ref{teo2} is given in Section~\ref{fegatograsso}.
Case \textup{(i)} has already been proved in \cite{FMRT0} with a~more
combinatorial method, not suited for extensions.
%
%re2.13 #&#
\begin{remark}
In the above result the only reminiscence of the initial distribution
is through the constant $c_0$ which is ``universal'' for a~large class
of initial interval laws $\mu$; see Lemma \ref{prelim:mainth}.
%Hence the term \emph{universality} in the title.
%C spostato qui
In particular, starting with a~stationary or $\mathbb{Z}$-stationary
renewal SPP (which necessarily corresponds to a~law $\mu$ with finite
mean), the weak limit of $Z^{(n)}$ always exists and is universal
($c_0=1$), depending on the rates only through the fulfilment of case
\textup{(i)} or case \textup{(ii)}, and not depending on
%not depending on the rates $\l^{(n)}_\ell$, $\l^{(n)}_r$, $
the sequence $\{d^{(n)}\}_{n\ge1}$ which defines the active
intervals.

We also underline that our results cover a~slightly more general class
of HCP. Indeed, following exactly the same lines of our proofs, we can
also treat more general triple merging allowing, for example, an active
domain to incorporate either its two neighbors to the left and/or its
two neighbors to the right and/or its left and right neighbors (this
last case is the only one considered in this paper). For this more
general class of HCP, both the above asymptotic result as well as the
one in Theorem \ref {teo3} are unchanged [instead the single epoch
evolution expressed by the differential equation \eqref{basilea} has to
be properly changed by adding to $\lambda_a$ the rates of these new
triple mergence events].
\end{remark}
%
%introduced by mistake
%We point out that the asymptotic Laplace distribution
%$g^{(\infty)}_{c_0}$ can be written also as $$g^{(\infty)}_{c_0}(s)= 1-
%$$
%where $Ei(\cdot)$ denotes the exponential integral function
%(following the notation of \cite{SE} and our paper \cite{FMRT0}) is
%instead more frequently denoted in mathematics literature by $
%is indeed the form appearing in \cite{DGY2} and \cite{SE} with $c_0
%=1$ (see previous Remark).
%%The remarkable fact of the above result is that the only reminiscence
%%of the initial distribution in the limiting law is through the
%%constant $c_0$ which, as
%%proved in Lemma \ref{prelim:mainth}, is ``universal'' for a~large
%class
%%of initial laws $\mu$. Hence the term \emph{universality} in the
%title.
%%%C spostato qui
%%We also stress that, starting with a~stationary or $\bbZ$--stationary
%%renewal SPP, the weak limit of $Z^{(n)}$ always exists
%%and it is universal ($c_0=1$), not depending on the rates
%%$\l^{(n)}_\ell$, $\l^{(n)}_r$.
%
%re2.14 #&#
\begin{remark} \label{referee}
The asymptotic Laplace distribution $g^{(\infty)}_{\kappa}$ can be
written also~as
\[
g^{(\infty)}_{\kappa}(s)= \mathcal{R} \biggl( \kappa\int
_s^\infty \frac{e^{-x}}{x} \,dx \biggr) =
\mathcal{R} \bigl( \kappa \mathrm{Ei}(s) \bigr),
\]
where $\mathrm{Ei}(\cdot)$ denotes the exponential integral
function.\footnote{Note that the function that we denote by
$\mathrm{Ei}(s)$ (following the notation of \cite{SE} and our paper
\cite{FMRT0}) is instead more frequently denoted in mathematics
literature by $\mathrm{E}_1(s)$.} This is indeed the form appearing in
\cite{DBG} and \cite{SE}. Moreover, in case \textup{(ii)} with
$\gamma=0$ in the above theorem (as in \cite{DBG}), one simply has
$\kappa= c_0/2$
and $ g^{(\infty)}_{c_0/2}= \tanh ( \frac{c_0}{2} \mathrm{Ei}(s) )$.

In case \textup{(i)} an expression as formal series of the probability
density with Laplace transform $g^{(\infty)}_{\kappa}$ can be found in
\cite {SE,FMRT0}.
\end{remark}

Next we concentrate on the asymptotic behavior of the first point law
when starting with a~right renewal SPP. %\marginpar{teorema da scrivere}
%
%th2.15 #&#
\begin{Theorem}\label{teo3}
%Let $\mu$ be probability measures on
%$\bbR$ and consider the hierarchical coalescence
%process such that the
%initial law $\cQ$ of $\xi^{(1)}$ is $\operatorname{Ren} (\delta_0,
Let $\nu, \mu$ be probability measures on $\mathbb{R}$ and $[1,
\infty)$, respectively.
% and let $g(s)$ be the Laplace transform of $\mu$.
Suppose that:
\begin{itemize}
\item the law $\mathcal{Q}$ of $\xi^{(1)}(0)$ is
    $\operatorname{Ren} (\nu, \mu)$;

\item it holds \textup{(i)} $\lambda_a^{(n)} \equiv0$ and
    $\lambda^{(n)}_r \equiv\gamma\lambda_{\ell}^{(n)}$ for all
    $n\geq1$ and for some $\gamma\geq0$ independent from $n$, or
    \textup{(ii)} $\lambda_\ell ^{(n)}\equiv 0$ and $
    \lambda_r^{(n)} \equiv0$ for all $n\geq1$;

\item the Laplace transform $g(s)$ of $\mu$ satisfies
\eqref{lim_arrosto}.
\end{itemize}
For any $n\geq1$ let $X_0^{(n)}$ be
the position of the first
point of the HCP at the beginning of the $n$th epoch, and let
$Y^{(n)}$ be the rescaled random variable $Y^{(n)}:=X_0^{(n)} /
d^{(n)}$.

Then the following holds:
\begin{itemize}
\item\cite{FMRT0} In case \textup{(i)}, as $n\to\infty$, $Y^{(n)}$
    weakly converges to the positive random variable
    $Y^{(\infty)}_{c_0}$ with Laplace transform given by
%
%e25 #&#
\begin{equation}
\label{thenero} \mathbb{E} \bigl( e^{-s Y^{(\infty)}_{c_0}} \bigr)=\exp \biggl\{-
\frac{c_0}{1+\gamma} \int_{(0,1)} \frac{1-e^{-sy} }{y} \,dy \biggr\},
\qquad s\in\mathbb{R}_+.
\end{equation}
%
%%%
%%%
%
\item In case \textup{(ii)}, supposing that $\int z \mu(dz)<
    \infty$ and that
%
%e26 #&#
\begin{equation}
\label{pato} \lim_{z \to\infty} \frac{1}{z} \int
_{[1,z]} x^2 \mu( dx) = 0
\end{equation}
as $n\to\infty$ the variable $Y^{(n)}$ weakly
converges
% $\infty$ if $\g>0$ and, for $\g=0$,
to the random variable $Y^{(\infty)}$ with values in $(0,\infty)$
and Laplace transform given by
%
%e27 #&#
\begin{equation}
\label{duro} \mathbb{E} \bigl( e^{-s Y^{(\infty)}} \bigr)= \frac{e^{-\bar\gamma/2}}{2}
\sqrt{ \frac{1-\tanh^2 ( \mathrm{Ei}(s)/2)}{s}}, \qquad s \in\mathbb{R}_+, %\exp\left\{ - \frac{1}{ 2}\left( e^{-s} -1 +\int_{(0,1)} s e^{-s y}
%+ \int_{(0,1)} \frac{e^{-sy}-1}{y} \,dy \right)\right\}
\end{equation}
where we let $\bar\gamma= -\!\int_0^\infty e^{-t}(\log t) \,dt
\simeq 0,577$ be the Euler--Mascheroni constant. Condition
\eqref{pato} is satisfied if $\int x^{1+\varepsilon} \mu(dx)<\infty
$ for some $\varepsilon>0$.
\end{itemize}
\end{Theorem}
The proof is given in Section~\ref{provateo3}. Case \textup{(i)} in the
above theorem has been stated only for completeness. It has already
been obtained in \cite{FMRT0}; see Theorem 2.24 there. Finally, we
point out that, due to Lemma \ref{prelim:mainth}, under condition
\eqref{pato} it must be $c_0=1 $ in limit \eqref{lim_arrosto}.

%re2.16 #&#
\begin{remark}
Extensions of the results presented in this section to OCPs and HCPs
starting from an exchangeable SPP can be easily achieved following the
arguments reported in \cite{FMRT0}, Appendix D. The key tool is given
by De Finetti's theorem, which gives a~characterization of exchangeable
SPPs as convex combinations of renewal SPPs. As final result, one gets
limit theorems as the ones obtained here but with more general limit
points.
\end{remark}

%s2.4 #&#
\subsection{Examples of HCPs}\label{physics}
We conclude this section by discussing some HCPs coming from the
physics literature.

%s2.4.1 #&#
\subsubsection{\texorpdfstring{The HCP associated to the East model at low temperature \cite{SE,FMRT1}}
{The HCP associated to the East model at low temperature  [10, 19]}}
An interesting and highly nontrivial example of
HCP has been devised in the physics literature \cite{SE} to model the
high density (or low temperature) nonequilibrium dynamics of the
\textit{East model} when a~deep quench from a~normal density state is
performed.
% (see
%C
The East model \cite{SE,EJ} is a~well-known example of kinetically
constrained stochastic particle system with site exclusion which
evolves according to a~Glauber dynamics submitted to the following
constraint: the $0/1$ occupancy variable at a~given site
$x\in\mathbb{Z}$ can change only if the site $x+1$ is empty (i.e., the
corresponding occupation variable equals zero). The change of the
occupation variable, when allowed by this constraint, occurs at rate
$q$ (resp., $1-q$) if it corresponds to a~change toward an empty
(resp., occupied) site. Note that each configuration can also be
represented by a~sequence of domains on $\mathbb{Z}$, where a~domain
represents a~maximal sequence of consecutive occupied sites delimited
by two empty sites. If the equilibrium vacancy density is very low
(i.e., in the limit $q\to 0$) and the initial distribution has a~normal
density (e.g., $q=1/2$), most of the nonequilibrium evolution will try
to remove the excess of vacancies of the initial state and will thus be
dominated by the coalescence of domains. In this setting, under
a~proper rescaling~\cite{FMRT1}, the East process can be well described
by an HCP with the following parameters: $d^{(1)}=1$,
$d^{(n)}=2^{n-2}+1$ for $n\geq2$,
$\lambda^{(n)}_r(d)=\lambda^{(n)}_a(d)=0$ for any value of the domain
length $d$, thus $\lambda^{(n)}=\lambda^{(n)}_\ell$ where
$\lambda^{(n)}_\ell$ is a~function expressed via a~proper large
deviation probability; see \cite{FMRT1} for the precise form of this
function.
%and it holds $\l^{(n)}_\ell(d)>0$ for any $d$ which belongs to the
%active range $d\in[d^{(n)},^{(n+1)})$.
We provide here only a~very short explanation to justify the above
choices of the parameters and refer the reader to \cite{SE} for an
heuristic explanation of the connection of this HCP with East and to
Section~3 of \cite{FMRT1} for a~rigorous description. The choice
$\lambda^{(n)}_a=0$ is due to the fact that the relevant event for East
corresponds to the disappearance of one zero at a~time, namely to the
coalescence of two domains (triple domain merging is not allowed). The
asymmetry between the right and left coalescence is due to the oriented
character of the East constraints, which implies that only the left
domains can be incorporated. Finally, the apparently weird choice of
the active ranges $d^{(n)}$ is due to the fact that in order to remove
the vacancy sitting at the left border of a~domain of size
$\ell\in[2^{n-1}+1,2^n]$, one needs to create at least $n$ additional
vacancies inside the domain; again, see \cite{SE,FMRT1} for details of
the combinatorial argument leading to this result. Thus energy barrier
considerations imply that this event requires a~typical time of order
$1/q^n$ which in turn means that in the regime $q\to0$ domains of sizes
$\ell,\ell'$ with $\ell\in[2^{n-1}+1,2^n]$, $\ell'\in[2^{m-1}+1,2^m]$
and $n\neq m$ are active (namely their left border can disappear) on
very well separated time scales.

%s2.4.2 #&#
\subsubsection{\texorpdfstring{The paste-all model \cite{DGY2}}{The paste-all model [7]}}
Another interesting HCP has been ``introduced'' in \cite{DGY2} and
named \textit{Paste-all model}. The model was intended to describe
breath figures, namely the patterns formed by growing and coalescing
droplets when vapor condenses on a~nonwetting surface. A common feature
of breath figure experiments is the occurrence of a~scale-invariant
regime with a~stable distribution of the drop sizes. In \cite{DGY2}
several simplified one-dimensional models were proposed to understand
this phenomenon, including the HCP named Paste-all model. In this case
all the domains are subintervals of the integer lattice, a~single
length is active in each epoch and domains merge with their left/right
neighbor with rate one, namely $d^{(n)}=n$,
$\lambda^{(n)}_\ell(n)=\lambda^{(n)}_r(n)=1$ and $\lambda
^{(n)}_a(n)=0$ (drops can coalesce either with their right or left
neighbor and the smaller droplets are the first that disappear).

%s2.4.3 #&#
\subsubsection{\texorpdfstring{The HCP associated to the 1d Ising model \cite{BDG}}
{The HCP associated to the 1d Ising model [2]}}\label{ising}
Finally, we recall the HCP which has been ``introduced'' in \cite{BDG}
to model the zero temperature Glauber dynamics of the one-dimensional
Ising model evolved from a~random initial condition. In this case the
domains correspond to the ordered spin regions, namely the maximal
sequence of consecutive sites with the same value of the spin, either
up or down. At late stages of the dynamics a~scale-invariant morphology
develops: the structure at different times is statistically similar
apart from an overall change of scale, that is, the system is described
by a~single, time-dependent length scale. Instead of considering the
stochastic Glauber dynamics the authors of \cite{BDG} start from the
well-known simpler deterministic model which is expected to mimic this
dynamics, namely the time-dependent Ginzburg--Landau equation for
a~scalar field in $d=1$, $\partial_t\phi=\partial_x^2\phi-dV/d\phi$
with $V(\phi)$ a~symmetric double well potential with minima at
$\phi=\pm1$ corresponding to the up and down phases for the Ising
model. If the model starts with a~$\phi$ profile corresponding to
a~random initial condition for the Ising model, then it evolves rapidly
to a~phase of subsequent regions were $\phi$~is close to $\pm1$
(corresponding to the ordered domains), and the dynamics is dominated
by the events that bring together and annihilate the closest pair of
domain walls. This in turn corresponds to the fact that the smaller
domains merge with the two neighboring domains.
% (suppose the smaller domain is in the up phase, then by definition
%the two neighboring domains are in the down phase and the first event
%will typically correspond to the fact that the up phase regions
%becomes a~down phase thus leading to the formation of a~larger down
%phase domain corresponding to the region initially covered by three
%domains).
Consequently, the HCP which has been introduced in~\cite{BDG} to mimic
this domain dynamics has parameters: $d^{(n)}:=n$ (only the smallest
length is active at each epoch), $\lambda^{(n)}_\ell(n)=\lambda
^{(n)}_r(n)=0$ and $\lambda^{(n)}_a(n)=1$ (only triple merging occurs).
%and the sequence
% analyze the metric structure on the space $\cN(d_{\min})$. In
%Section \ref{universal} we present the universal coupling of the
%OCP's. Both Sections \ref{dugundi} and \ref{universal} are
%preliminary to the study of the Markov generator of the OCP,
%performed in Sections \ref{semigruppo_OCP} and \ref{infinitesimi},
%and leading the Theorem \ref{amico_marco}. In Section
%generator of the OCP to prove Theorem \ref{differisco}. In Section

%%%%%%%%%%%%%%%%%%%%%%%%%%%%%%%%%%%%%%%%%%%%%%%%%%%%%%%%%%%%%%%%%%%%%%%%%%%%%
%%%%%%%%%%%%%%%%%%%%%%%%%%%%%%%%%%%%%%%%%%%%%%%%%%%%%%%%%%%%%%%%%%%%%%%%%%%%%%%

%s3 #&#
\section{Metric structure of $\mathcal{N}(d_{\min})$}\label{dugundi}

%In this section we we collect some useful results on the vague
%topology and we explain how $\cN(d_{\min})$ can be metrized.

Let us write $\mathcal{M}$ for the space of Radon measures on $\mathbb
{R}$, that is, locally finite Borel nonnegative measures. We consider
this space endowed of the vague topology, such that $\mu_n \to\mu$ in
$\mathcal {M}$ if and only if $\mu_n(f) \to\mu(f)$ for all continuous
functions $f$ on $\mathbb{R}$ with compact support (shortly, $f \in
C_0$). Then $\mathcal{M}$ can be metrized by a~suitable metric $m$
making it a~Polish space; see \cite{DV}, Section~A2.6, and observe that
since the Euclidean space $\mathbb{R}$ is Polish and locally compact,
the vague topology coincides with the $\hat w$-topology as discussed
before \cite{DV}, Corollary A2.6.V. We recall the definition of $m$
since it will be useful below,
\[
m(\mu,\nu):=\int_0^\infty e^{-r}
\frac{ d_r ( \mu^{(r)},\nu^{(r)}
)}{1+ d_r ( \mu^{(r)},\nu^{(r)}  ) } \,dr,\qquad \mu,\nu\in\mathcal{M},
\]
where $\mu^{(r)}$, $\nu^{(r)}$ denote the restriction to $(-r,r)$ of
$\mu,\nu$, while $d_r$ stands for the Prohorov distance for measures on
$(-r,r)$; see \cite{DV}, Section A2.5.

The space $\mathcal{N}$ introduced in Section~\ref{sec_SPP} can be
thought of as a~subspace of $\mathcal{M}$, identifying the set
$\xi\in\mathcal{N}$ with the measure $\sum_{ x \in\xi} \delta _x$. Then
one gets that the \mbox{$\sigma$-}algebra of its Borel subsets
coincides with the \mbox{$\sigma$-}algebra of measurable subsets
introduced in Section~\ref{sec_SPP}; see \cite{DV}, Chapter~7, in
particular Proposition 7.1.III and Corollary 7.1.VI there. Therefore,
the same property holds for $\mathcal{N}(d_{\min})$ (i.e., Borel
subsets and measurable subsets coincide).

%le3.1 #&#
\begin{Lemma}\label{gnocchi}
The following holds:
\begin{longlist}[(iii)]
\item[(i)] %A sequence $(\xi_n)_{n \geq1}$ in $\cN( d_{\min})$
%vaguely converges to $\xi$
$\xi_n \to\xi$ in $\mathcal{N}( d_{\min})$ if and only if $|\xi _n
\cap [a,b]| \to|\xi \cap[a,b]|$ for each interval $[a,b]$ such that
$\xi\cap \{a,b\}=\varnothing$. The same criterion holds replacing
closed intervals $[a,b]$ by open intervals $(a,b)$.

\item[(ii)] Suppose that $\xi_n \to\xi$ in $\mathcal{N}( d_{\min})$.
    Fix
    $a<b$ with $\xi\cap\{a,b\} = \varnothing$. Then $\xi_n \cap(a,b)
    \to\xi\cap(a,b)$ in $\mathcal{N}(d_{\min} )$. Moreover, for $n$
    large enough $\xi_n \cap(a,b)$ has the same cardinality as
    $\xi\cap(a,b)$ and $\lim_{n\to\infty} x_i^{(n)}= x_i$ for $1\leq i
    \leq k$, where $\xi_n\cap(a,b)= \{x_1^{(n)}<x_2^{(n)} <
    \cdots<x_k^{(n)} \}$ and $\xi\cap(a,b)= \{x_1<x_2 < \cdots<x_k\}$.

\item[(iii)] The space $\mathcal{N}( d_{\min})$ is a~closed subset of
    $\mathcal{M}$. In particular, it is a~Polish space endowed of the
    metric $m$.

\item[(iv)] The space $\mathcal{N}(d_{\min})$ is compact.
\end{longlist}
\end{Lemma}

\begin{pf} Part \textup{(i)} with closed intervals follows from \cite{DV}, Proposition A2.6.II; see also \cite{FKAS}, Theorem 1.1.16, with
$P_n:=\delta_{\xi_n}$ and $P:= \delta_\xi$. The same criterion with
open interval is a~simple derivation from the one with closed interval.

Let us consider part~\textup{(ii)}. Applying the criterion in
part~\textup{(i)} it is trivial to check that $\xi_n \cap(a,b)
\to\xi\cap(a,b)$. Take now $\varepsilon>0$ small enough that all the
intervals $J_i=[x_i-\varepsilon,x_i+\varepsilon]$, $1\leq i \leq k$,
are disjoint and intersect $\xi$ only at $x_i$. Then, by item
\textup{(i)} for $n$ large $\xi_n$ has exactly one point in each $J_i$.
Similarly, $\xi_n$ has exactly $k$ points in $(a,b)$ for $n$ large. By
the arbitrariness of $\epsilon$ we can conclude.

To prove part \textup{(iii)} call $\bar\mathcal{N}$ the family of
counting measures in $\mathbb{R}$, that is, $\xi\in\bar\mathcal{N}$ if
and only if $\xi= \sum_i k_i \delta_{x_i}$ with $k_i \in\mathbb{N}_+$
and $\{x_i\}$ being a~locally finite countable subset of $\mathbb{R}$.
By \cite{DV}, Proposition 7.1.III, $\bar \mathcal{N}$ is a~closed
subset of $\mathcal{M}$. Hence, if $\xi _n\in\mathcal{N}(d_{\min}) $
and $\xi_n \to\xi$ with $\xi$ in $\mathcal{M}$, then $\xi\in
\bar\mathcal{N}$. We only need to show that $\xi\in\mathcal{N}(
d_{\min})$. Suppose by contradiction that $\xi(\{x\})\geq2$ for some $x
\in\mathbb{R}$. Take $I= [x-\epsilon, x+\epsilon]$ such that $\xi(
\{x-\epsilon, x+\epsilon\})=0$ and $2 \varepsilon< d_{\min}$ (the
existence of $\epsilon$~is guaranteed by the fact that
$\xi\in\bar\mathcal{N}$). By part \textup{(i)} it must be $\xi_n(I)
\geq2$ for $n$ large enough, in contradiction with the fact that
$\xi_n$ can have at most one point in $I$.

Due to part \textup{(iii)}, part \textup{(iv)} is a~simple consequence
of the compactness criterion given in \cite{DV}, Corollary A2.6.V.
\end{pf}

Recall that $\mathbb{B}$ denotes the Banach space of all bounded
continuous functions $f\dvtx \mathcal {N}(d_{\min}) \to\mathbb{R}$
endowed with the uniform norm $\| \cdot\|$ and that
$\mathbb{B}_{\mathrm{loc}}$ denotes the set of local functions $f \in
\mathbb{B}$.

%le3.2 #&#
\begin{Lemma}\label{lemino}
The set $\mathbb{B}_{\mathrm{loc}}$ is dense in $\mathbb{B}$. In
particular, given $f \in\mathbb{B}$ and defining $f_N(\xi):=
\int_{N}^{N+1} f(\xi \cap(-r,r))\,dr$, it holds $f_N
\in\mathbb{B}_{\mathrm{loc}}$ and $f_N \to f$ in~$\mathbb{B}$.
\end{Lemma}

Note that the map $\mathbb{R}_+ \ni r \mapsto f(\xi\cap(-r,r))\in
\mathbb{R}$ is stepwise, with a~finite number of jumps in any finite
interval. Hence, the above function $f_N$ is well defined.

\begin{pf*}{Proof of Lemma \ref{lemino}}
Take $ f\in\mathbb{B}$. Since $\mathcal{N}(d_{\min})$ is compact, $f$
is uniformly continuous. Hence, given $\epsilon>0$, there exists
$\delta_0>0$ such that $m(\xi,\xi')<\delta_0$ implies
$|f(\xi)-f(\xi')|<\epsilon$. Take $N_0\in\mathbb{N}$ large enough that
$e^{-N_0} \leq\delta_0$. By the definition of $m$ we have $m (\xi,
\xi\cap(-N,N))\leq \int_{N}^\infty e^{-a}\,da \leq\delta_0$ for any $N
\geq N_0$. This implies that $|f(\xi)-f(\xi\cap(-r,r)) |\leq\epsilon$
for all $r \geq N_0$ and therefore $\|f- f_N\|\leq\epsilon$. Trivially
$f_N$ is a~local function, and it remains to prove that $f_N$ is
continuous. To this aim, fix $\xi\in\mathcal{N}(d_{\min})$. Then the
set $R=\{r\in [N,N+1]\dvtx  \xi\cap\{-r,r\}\neq \varnothing\}$ is
finite. In particular, by Lemma \ref{gnocchi}\textup{(ii)}, if $\xi_n
\to\xi$, then $\xi_n \cap(-r,r)\to\xi\cap(-r,r)$ for all $r
\in[N,N+1]\setminus R$. Since $f$ is continuous, we get that
\[
f\bigl(\xi_n\cap(-r,r)\bigr) \to f\bigl(\xi\cap(-r,r)\bigr)\qquad
\forall r\in [N, N+1]\setminus R.
\]
We conclude applying now the dominated convergence theorem.
\end{pf*}

%%%%%%%%%%%%%%%%%%%%%%%%%%%%%%%%%%%%%%%%%%%%%%%%%%%%%%%%%%%%%%%%%%%%%%%%%%%%%%%
%%%%%%%%%%%%%%%%%%%%%%%%%%%%%%%%%%%%%%%%%%%%%%%%%%%%%%%%%%%%%%%%%%%%%%%%%%%%%%%
%%%%%%%%%%%%%%%%%%%%%%%%%%%%%%%%%%%%%%%%%%%%%%%%%%%%%%%%%%%%%%%%%%%%%%%%%%%%%%%
%s4 #&#
\section{\texorpdfstring{OCP process: Proof of Theorems \protect\ref{differisco} and~\protect\ref{differiscobis}}
{OCP process: Proof of Theorems 2.6 and 2.8}}\label{prova_differisco}

In this section we prove Theorems \ref{differisco}
and~\ref{differiscobis}, applying our analysis of the Markov generator
of the OCP; recall Corollary \ref{corri}.

%s4.1 #&#
\subsection{\texorpdfstring{Differential equation for $\mu_t$ and proof of Theorem~\protect\ref{differisco}}
{Differential equation for mu t and proof of Theorem 2.6}}
As application of Corollary \ref{corri}
we
can prove the following result:
%
%pr4.1 #&#
\begin{Proposition}\label{diuretico}
Let $f\dvtx [0,\infty)\to\mathbb{R}$ be a~continuous function such that
%
%e28 #&#
\begin{equation}
\label{neve} \sum_{k =0}^\infty\sup
_{x \geq k} \bigl|f(x)\bigr|<\infty.\vadjust{\goodbreak}
\end{equation}
Let $\mu$ be a~probability measure on $[d_{\min},\infty)$ and $\mu_t$
be as in Lemma \ref{vecchioteo1} with the choice $Q= \operatorname{Ren}
(\mu)$. Then the function $[0,\infty)\ni t\mapsto\mu_t(f)\in\mathbb{R}$
is differentiable and
%
%e29 #&#
\begin{eqnarray}\label{basilea}
\frac{d}{dt} \mu_t(f) & =& -\int\mu_t
(dx) \lambda(x) f(x)\nonumber
\\
&&{} +\int\mu_t (dx) \int\mu_t (dy)
\bigl( \lambda_r(x)+\lambda_\ell(y) \bigr)f(x+y)
\\
&&{}+ \int\mu_t (dx) \int\mu_t (dy) \int
\mu_t (dz) \lambda_a (y) f(x+y+z).\nonumber
\end{eqnarray}
\end{Proposition}

\begin{pf}
Set $\mathcal{Q}= \operatorname{Ren} (\delta_0, \mu)$. Note that
$\mathbb{P}_\mathcal{Q}$-a.s. $\xi(t)$ belongs to the set
$\mathcal{N}_*$ of configurations $\xi\in\mathcal{N}(d_{\min})$ such
that $\xi\subset[0, \infty)$, $\xi\cap(0,d_{\min}/2]=\varnothing$, and
$\xi$~is given by an increasing sequence of points diverging to
$\infty$. %Note that $\cN_*$ is left invariant by the annihilation of
%a finite set of points.
Points in $\xi\in\mathcal{N}_*$ are labeled as $x_0(\xi),x_1(\xi ),x_2
(\xi), \ldots$ in increasing order. Then, by Lemma \ref{vecchioteo1},
$\mu_t$~equals the law of $x_1(\xi(t))$ under $\mathbb{P}_\mathcal {Q}(
\cdot| 0\in\xi(t))$. Hence we can write $ \mu_t(f) = N_t/D_t $ where
\begin{eqnarray*}
N_t&=& \mathbb{E}_\mathcal{Q} \bigl[ f \bigl(x_1
\bigl(\xi(t) \bigr) \bigr); 0\in\xi(t) \bigr], \qquad D_t =
\mathbb{P}_\mathcal{Q}\bigl( 0\in\xi(t) \bigr).
\end{eqnarray*}
Let $\rho\dvtx  \mathbb{R}\to[0,1]$ be a~continuous function such that
$\rho(x)=0$ for $x \notin(-\frac{d_{\min}}{2},\break  \frac{d_{\min}}{2})$ and
$\rho(0)=1$. By definition of vague convergence (see
Section~\ref{dugundi}), the function $\Phi\dvtx
\mathcal{N}(d_{\min})\mapsto \mathbb{R}$ defined as $\Phi(\xi):=
\sum_{x \in\xi} \rho(x)$ is a~continuous map. Since $\Phi$ is local, it
belongs to $\mathbb{B}_{\mathrm{loc}}$, and moreover it satisfies
$\Phi(\xi)= \mathbh{1}_{ 0 \in\xi}$ for all $\xi\in\mathcal{N}_*$. In
Lemma \ref{tosse} below we exhibit a~function $\Psi\in\mathbb{D}$ that
satisfies $\Psi(\xi)=f (x_1 (\xi ) )\mathbh{1}_{ 0 \in\xi}$ for all
$\xi\in\mathcal{N}_*$. Hence we can write
\[
N_t =\mathbb{E}_\mathcal{Q} \bigl[\Psi\bigl(\xi(t) \bigr)
\bigr], \qquad D_t=\mathbb{E}_\mathcal{Q} \bigl[ \Phi\bigl(
\xi(t) \bigr) \bigr].
\]
By standard properties of Markov generators, we conclude that the maps
$N_t, D_t$ are differentiable and that
\[
N_t'=\mathbb{E}_\mathcal{Q} \bigl[ \mathcal{L}
\Psi\bigl(\xi(t) \bigr) \bigr], \qquad D'_t =
\mathbb{E}_\mathcal{Q} \bigl[\mathcal{L}\Phi\bigl(\xi(t) \bigr) \bigr].
\]
Since $\Psi, \Phi\in\mathbb{D}$, we can use equation \eqref{santi}
%%apply Theorem \ref{amico_marco}
to compute $\mathcal{L}\Psi$ and $\mathcal{L}\Phi$. We need their
value only on
$\mathcal{N}_*$.
%Since at the end we only
%need their value on $\cN_*$ and since the space $\cN_*$ is left
%invariant by annihilation of a~finite number of points in $\xi\in
%of $\Phi, \Psi$ outside $\cN_*$. Since we are dealing with right
%renewal SPP's points are numbered in increasing order starting from
%$x_0$.
Suppose that $\zeta, \xi\in\mathcal{N}_*$ are such that $\zeta
\subset\xi$ and $0\in\xi$. %Recall that $d_i:= x_{i}- x_{i-1}$ and
%that
%points are enumerated such that $x_0 \leq0 \leq x_1$.
Writing $x_i$ and $d_i$ instead of $x_i(\zeta)$ and $d_i (\zeta)= x_i
(\zeta)-x_{i-1}(\zeta)$, we get
\[
\cases{ \mathcal{L}\Psi(\zeta) = \mathbh{1}(0\in\zeta) G (\zeta),
\cr
\mathcal{L}\Phi(\zeta)= \mathbh{1}(0 \in \zeta) H(\zeta),}
\]
where $ H(\zeta):= - \lambda_\ell(d_1 ) -\lambda_a( d_1 )$ and
\begin{eqnarray*}
G(\zeta)&:=& - \bigl[\lambda_\ell(d_1)+
\lambda_a (d_1) \bigr] f(x_1) + \bigl[
\lambda_r (d_1)+ \lambda_\ell(d_2)
\bigr] \bigl[f(x_2)-f(x_1) \bigr]
\\
&&{}+ \lambda_a (d_2) \bigl[ f(x_3)-f(x_1)
\bigr].
\end{eqnarray*}
Since $N_t$ and $D_t$ are differentiable, we get that $N_t/D_t$ is
differentiable and that
\[
\frac{d}{dt} \mu_t (f) = \frac{d}{dt} \frac{N_t}{D_t}
= \frac
{N_t'}{D_t}- \frac{N_t}{D_t} \frac{D_t '}{D_t}.
\]
Writing $F(\xi)= f( x_1 (\xi) )$, the above identities imply that
%
%e30 #&#
\begin{eqnarray}\label{malattia}
\frac{d}{dt} \mu_t (f)&=& %\frac{d}{dt} \frac{N_t}{D_t}=
\mathbb{E}_\mathcal{Q} \bigl( G\bigl(\xi(t) \bigr) | 0\in \xi(t) \bigr)
\nonumber\\[-8pt]\\[-8pt]
&&{}-\mathbb{E}_\mathcal{Q} \bigl( F \bigl(\xi(t) \bigr) | 0\in\xi (t) \bigr)
\mathbb{E}_\mathcal{Q} \bigl( H\bigl(\xi(t) \bigr) | 0\in\xi (t) \bigr).\nonumber
\end{eqnarray}
By Lemma \ref{vecchioteo1}\textup{(iv)}, we can write (brackets should
help to follow the computations)
%
%e31 #&#
\begin{eqnarray}
\label{infettiva}
&& \mathbb{E}_\mathcal{Q} \bigl( G\bigl(\xi(t) \bigr) | 0\in
\xi(t) \bigr)\nonumber
\\
&&\qquad = - \bigl\{ \mu_t ( \lambda_\ell f) +
\mu_t (\lambda_a f ) \bigr\}
\nonumber
\\
&&\quad\qquad{} + \biggl\{ \int\mu_t(dx)\int\mu_t (dy) \bigl[
\lambda_r(x)+ \lambda_\ell(y) \bigr]f(x+y)
\nonumber\\[-12pt]\\[-8pt]
&&\hspace*{130pt}{} -\mu_t (\lambda_r f) -\mu_t (
\lambda_\ell) \mu_t (f) \biggr\}\nonumber
\\[-2pt]
&&\quad\qquad{} + \biggl\{ \int\mu_t(dx) \int\mu_t (dy)\int
\mu_t(dz) \lambda_a(y) f(x+y+z) - \mu_t (
\lambda_a) \mu_t (f) \biggr\}\hspace*{-14pt}\nonumber
\end{eqnarray}
and
%
%e32 #&#
\begin{eqnarray}\label{preghiamo}
&& \mathbb{E}_\mathcal{Q} \bigl( F \bigl(\xi(t) \bigr) | 0\in
\xi(t) \bigr) \mathbb{E}_\mathcal{Q} \bigl( H\bigl(\xi(t) \bigr) | 0\in
\xi(t) \bigr)
\nonumber\\[-8pt]\\[-8pt]
&&\qquad  = - \mu_t (f)\mu_t (\lambda_\ell)
-\mu_t (f) \mu_t (\lambda_a).\nonumber
\end{eqnarray}
Combining the above identities \eqref{malattia}, \eqref{infettiva} and
\eqref{preghiamo} we get the thesis.
\end{pf}

In the proof of Proposition \ref{diuretico} above, we used the
following technical lemma.

%le4.2 #&#
\begin{Lemma}\label{tosse}
Let $f$ be a~real continuous function on $[0,\infty)$ satisfying
\eqref{neve} and extend it to a~continuous function on $\mathbb{R}$
constant on $(-\infty,0]$. Given $s \in\mathbb{R}$ define
\[
f_s (\xi) = \cases{ 0, &\quad if $ \bigl|\xi\cap(s, \infty) \bigr|\leq1$,
\cr
f \bigl( z \bigl(\xi\cap(s,\infty) \bigr) \bigr), &\quad otherwise,}
\]
where $ z (\xi\cap(s,\infty) )$ denotes the second point from
the left of $\xi\cap(s, \infty)$. Then the function
\[
F\dvtx \mathcal{N}(d_{\min} ) \ni\xi\mapsto\frac{1}{d_{\min}
}\int
_{-d_{\min}}^{0} f_s(\xi) \,ds \in\mathbb{R}
\]
belongs to $\mathbb{B}$. Moreover, the function $\Psi(\xi)= \Phi(\xi)
F(\xi)$ belongs to $\mathbb{D}$ and $\Psi(\xi)=f (x_1 (\xi)
)\mathbh{1}_{ 0 \in\xi}$ for all $\xi\in\mathcal{N}_*$; for the
definition of $\Phi$ and $\mathcal{N}_*$, see the proof of Proposition
\ref{diuretico}.
\end{Lemma}
The integrand in the definition of $F$ is a~stepwise function with
a~finite number of jumps; hence it is integrable.

\begin{pf*}{Proof of Lemma \ref{tosse}}
Let us prove the continuity of $F$. Take $\xi_n \to\xi$ in
$\mathcal{N}(d_{\min})$ and set $R:=\{s \in(-d_{\min},0)\dvtx   s
\notin \xi\}$. We claim that, fixed $s \in R$, it holds $f_s (\xi_n)
\to f(\xi)$. Let us first suppose that $|\xi\cap (s,\infty)|\geq2$. Let
$a<b$ be the first two points of $\xi \cap(s,\infty)$, and take $c$
larger than $b$ such that $\xi$ has no point in $(b,c]$. Then by Lemma
\ref{gnocchi}\textup{(ii)} $\xi_n \cap(s,c)$ has exactly two points
$a^{(n)}<b^{(n)}$ eventually in~$n$, moreover $a^{(n)} \to a$ and
$b^{(n)} \to b$. By the continuity of $f$, we have
\[
f_s(\xi_n)= f\bigl(b^{(n)}\bigr) \to f(b) =
f_s (\xi).
\]
Let us now suppose that $|\xi\cap(s,\infty)|\leq1$. Suppose first that
$\xi\cap(s, \infty)$ has only one point, denoted by $x_*$. Given
$\varepsilon>0$ fix $L>x_*$ such that $L \notin\xi$ and $|f(x)|
\leq\varepsilon$ for $x \geq L$ ($L$ exists due to \eqref{neve}). By
Lemma \ref{gnocchi}\textup{(i)} for $n$ large $\xi_n$ has exactly one
point in $(s, L)$. This assures that $|f_s(\xi_n)|\leq\varepsilon$ for
$n$ large and therefore that $\lim_{n\to\infty} f_s(\xi_n)=0=f_s
(\xi)$. A similar argument can be applied when $\xi$ has no point in
$(s, \infty)$. This completes the proof of our claim.

Combining our claim with the dominated convergence theorem and with the
fact that $R$ is a~finite set, we get that $F(\xi_n) \to F(\xi)$, thus
proving the continuity of~$F$.

If $\xi\in\mathcal{N}_*$ it is simple to check that $\Psi(\xi )=f (x_1
(\xi) )\mathbh{1}_{ 0 \in\xi}$. It remains to prove that
$|\!|\!|\Psi|\!|\!|<\infty$. Suppose that $k \in\mathbb{Z}$ and
$\nabla_k \Psi (\xi)\neq 0$. Then $k \geq-1$ and $\xi$ has at least two
points in $(-d_{\min}, \infty)$, the first or the second one (from the
left) must lie in $I_k$. In particular, it must be $|\nabla_k \Psi(\xi)
|\leq2 \sup_{x \geq k d_{\min} } |f(x) |$.
Take now %\club\club
$k<k'\leq k+\lceil d_{\max}/d_{\min}\rceil$ in $\mathbb{Z}$. Suppose
that $\nabla_{(k,k')} \Psi(\xi)\neq0$. If $k<-1$, then $\nabla_{(k,k')}
\Psi(\xi)=\nabla_{k'} \Psi(\xi)$ which can be bounded as above. If
$k,k'\geq-1$, then we conclude that $\xi$ has at least two points in
$(-d_{\min}, \infty)$; the first or the second one (from the left) must
lie in $I_k\cup I_{k'}$. Hence $|\nabla_{(k,k')} \Psi(\xi) |\leq2
\sup_{x \geq k d_{\min} } |f(x) |$. The above bounds and condition
\eqref{neve} allow us to conclude.
\end{pf*}

%Suppose that $\xi\subset\bbZ$ for $\cQ$--a.a. $\xi$. Then the
%function $f(x)= \mathbh{1}(x=d)$ satisfies the hypothesis of
%Proposition \ref{diuretico}, while $\mu_t(f)= \mu_t
%(\{d\})=:\mu_t(d)$. Hence, \eqref{basilea} reads

We have now all the tools to prove Theorem \ref{differisco}.

\begin{pf*}{Proof of Theorem \ref{differisco}}
To prove \eqref{davide} and \eqref{orecchiettebis} we can restrict
to $s>0$. Indeed, writing these differential equations as integral
identities one can take the limit $s\downarrow0$ and recover the case
$s=0$.

We can apply Proposition \ref{diuretico} to the function $f(x)= e^{-s
x}$, $x \geq0$, getting that $G_t(s) = \mu_t (f)$ is
$t$-differentiable, with derivative given by \eqref{basilea}.

We can write $H_t(s) = \mu_t (\tilde f)$ where $\tilde f (x):=e^{-s
x}\mathbh{1}(x<d_{\max} )$. Obviously $\tilde f$~is not suited to
Proposition \ref{diuretico} since it is not continuous. If $\mu$ had
support on a~lattice, for example, $\mathbb{N}$, trivially $\tilde f $
could be replaced by a~nice function. In the general case we need more
care. For $\varepsilon>0 $ small enough, we fix a~continuous function
$f_\varepsilon$ on $[0, \infty)$ with values in
$[0,1]$ such that $f_\varepsilon(x)=\tilde f (x) $ if $x \notin
(d_{\max} -\varepsilon,d_{\max} )$. Applying Proposition
\ref{diuretico} we get that the function $[0,\infty) \ni t \mapsto\mu_t
(f_\varepsilon)$ is differentiable with derivative given by
\eqref{basilea} (with $f$ replaced by $f_\varepsilon$). Since $\mu _t$
has support in $[d_{\min}, \infty)$ and since $f_\varepsilon(x+y)=0$,
$f_\varepsilon(x+y+z)=0$ if $x,y, z \geq d_{\min}$ [recall assumption
\textup{(A2)} in Section~\ref{sec_OCP}], from \eqref{basilea} we
conclude that $ \mu_t' (f_\varepsilon)= - \mu_t ( \lambda
f_\varepsilon)$.

Let $\alpha_t(\varepsilon):= \mu_t  ( (d_{\max}-\varepsilon, d_{\max})
)$. Trivially, $\lim_{\varepsilon\downarrow0} \alpha_t (\varepsilon
)=0$ and $0\leq \alpha_t(\varepsilon)\leq1$. Since $ |\mu_t (\tilde f)
- \mu_t (f_\varepsilon)| \leq\alpha_t( \varepsilon) $ and
\[
\bigl| \mu_t ' (f_\varepsilon)+ \mu_t (
\lambda\tilde{f}) \bigr| = \bigl| \mu_t (\lambda f_\varepsilon)-
\mu_t (\lambda\tilde{f}) \bigr| \leq \alpha_t( \varepsilon)\|
\lambda\|_\infty,
\]
applying the dominated convergence theorem we
get
\begin{eqnarray*}
H_t (s)&=&\mu_t (\tilde f)= \lim
_{\varepsilon\downarrow0} \mu_t (f_\varepsilon) = \lim
_{\varepsilon\downarrow0} \biggl[ \mu_0(f_\varepsilon) - \int
_0 ^t\mu_u( \lambda
f_\varepsilon) \,du \biggr]
\\
&=&\mu_0(\tilde f) - \int
_0 ^t\mu_u( \lambda \tilde f) \,du.
\end{eqnarray*}
Since $\lambda$ is a~continuous function [extendable on $[0,\infty)$]
and is zero on $[ d_{\max}, \infty)$, we can apply Proposition
\ref{diuretico} to the function $\lambda\tilde f$ concluding that the
map $[0,\infty) \ni t \mapsto\mu_t (\lambda\tilde f)$ is differentiable
and therefore continuous. This observation together with the above
identity allows us to conclude that $H_t(s)$ is $t$-differentiable, and
its derivative satisfies \eqref{davide} [note that $\lambda\tilde{f}=
\lambda f $ by assumption \textup{(A1)} in Section~\ref{sec_OCP}].
Knowing that $\partial_t G_t(s)$ is given by \eqref{basilea} and using
\eqref{davide} we get \eqref{orecchiettebis}.
%we conclude that
% \partial_t \bigl[G_t(s)-H_t(s)\bigr]
% & =
% \int\mu_t (dx) \int\mu_t (dy) \bigl( \l_r(x)+\l_\ell(y)
% &&{}+ \int\mu_t (dx) \int\mu_t (dy) \int\mu_t (dz) \l_a
%(y)e^{-s(x+y+z)}.. \nonumber
%Since $\mu_t$ has support on $[d_{\min},\infty)$, by assumption
%(A2) we can remove the characteristic functions in

We observe that in case \textup{(i)} it holds $ \lambda_\ell+ \lambda
_r=\lambda$, while in case \textup{(ii)} it holds
$\lambda_\ell+\lambda_r= \lambda\gamma/ (1+\gamma)$ and $\lambda_a =
\lambda /(1+\gamma)$. These identities allow us to derive from
\eqref{davide} and \eqref{orecchiettebis} that $\partial_t G_t (s)=
\partial_t H_t(s) (1- G_t(s))$ in\vspace*{1pt} case \textup{(i)} and that
$\partial_t G_t (s)=
\partial_t H_t(s)  ( 1 - \frac{G_t(s)(\gamma+
G_t(s))}{1+\gamma} )$ in case \textup{(ii)}. The rest of the proof
follows by the computations outlined in Remark \ref{pasquetta} using
\eqref{davide}, \eqref{orecchiettebis} and the fact that $\lim_{t\to
\infty} G_t(s)= G_\infty(s)$, $\lim_{t \to\infty} H_t(s) =H_\infty
(s)=0$ [which is due to Lemmas~\ref{volarelontano}\textup{(iii)}
and~\ref{vecchioteo1}\textup{(v)}]. We only point out that with the
definition of $A_t(s), B_t(s)$ given in Remark \ref{pasquetta} one gets
$ b_u(s) e^{A_u(s) + 2 B_u(s) } = \frac{1}{\gamma+2}\partial_u e^{- H_u
((\gamma+2)/(\gamma+1)) } $ in case \textup{(ii)}.
\end{pf*}

%%%%%%%%%%%%%%%%%%%%%%%%%%%%%%%%%%%%%%%%%%%%%%%%%%%%%%%%%%%%%%%%%%%%%%%%%%%%%
%%%%%%%%%%%%%%%%%%%%%%%%%%%%%%%%%%%%%%%%%%%%%%%%%%%%%%%%%%%%%%%%%%%%%%%%%%%%%
%%%%%%%%%%%%%%%%%%%%%%%%%%%%%%%%%%%%%%%%%%%%%%%%%%%%%%%%%%%%%%%%%%%%%%%%%%%%%

%s4.2 #&#
\subsection{\texorpdfstring{Differential equation for $\nu_t$ and proof of Theorem \protect\ref{differiscobis}}
{Differential equation for nu t and proof of Theorem 2.8}}
As in the case of the interval law
$\mu_t$, in order to prove Theorem \ref{differiscobis}, we need first
to establish a~differential equation for the expectation $\nu_t(f)$ for
nice functions~$f$.
%Recall that $\lfloor\cdot\rfloor$ denotes the entire part.

%pr4.3 #&#
\begin{Proposition}\label{vongole}
%Consider a~hierarchical coalescence process starting from a~right
%renewal SPP $\mathcal{Q} = Ren(\nu,\mu)$ with.
%Then, for any continuous function $f\dvtx  [\inf\mbox{supp}(\nu),\infty)
%[k, \infty)} |f(x)|<\infty,
Let $f\dvtx  [0,\infty) \to\mathbb{R}$ be as in Proposition \ref
{diuretico}. Let $\mu$ be a~probability measure on $[d_{\min},\infty)$
and $\nu_t$ be as in Lemma \ref{vecchioteo1} with the choice $Q=
\operatorname{Ren} (\delta_0,\mu)$. Then the map $[0,\infty)\ni
t\mapsto\nu_t(f) \in\mathbb{R}$ is differentiable and
%
%e33 #&#
\begin{eqnarray}
\label{carbonara}
\frac{d}{dt} \nu_t(f) & =& - \int
\nu_t(dx) \int \mu_t(dy) \bigl(\lambda_\ell(y)
+ \lambda_a(y) \bigr) f(x)\nonumber
\\
&&{} + \int \nu_t (dx) \int
\mu_t (dy) \lambda_\ell(y) f(x+y)
\\
&&{} + \int\nu_t (dx) \int\mu_t (dy) \int
\mu_t (dz) \lambda_a (y) f(x+y+z).\nonumber
\end{eqnarray}
\end{Proposition}
%
%In the above statement, if $\inf\mbox{supp}(\nu)=-\infty$, then we
%understand that $[\inf\mbox{supp}(\nu),\infty)=\mathbb{R}$.

\begin{pf}
We extend $f$ as continuous function to all $\mathbb{R}$, constant on
$(-\infty,0]$. Given $s \in\mathbb{R}$ define
\[
f_s (\xi) = \cases{ 0, &\quad if $\xi\cap(s, \infty)=0$,
\cr
f
\bigl( z \bigl(\xi\cap(s,\infty) \bigr) \bigr), &\quad otherwise,}
\]
where $ z (\xi\cap(s,\infty) )$ denotes the first point from
the left of $\xi\cap(s, \infty)$. Then the function
\[
\Theta\dvtx \mathcal{N}(d_{\min} ) \ni\xi\mapsto\frac{1}{d_{\min} }\int
_{-d_{\min}}^{0} f_s(\xi) \,ds \in\mathbb{R}
\]
belongs to $\mathbb{D}$ and
$\Theta(\xi) = f(x_0(\xi)) $ if $\xi\in\mathcal{N}_*$. The proof
is similar
to the one of Lemma \ref{tosse}, and we omit the details.

Set $\mathcal{Q}= \operatorname{Ren} (\delta_0, \mu)$. Note that
$\mathbb{P}_\mathcal{Q}$-a.s. $\xi(t)$ belongs to the set
$\mathcal{N}_*$ of configurations $\xi\in\mathcal{N}(d_{\min})$ such
that $\xi\subset[0, \infty)$, $\xi\cap(0, d_{\min}/2]=\varnothing$, and
$\xi$ is given by an increasing sequence of points diverging to
$\infty$. Points in $\xi\in\mathcal{N}_*$ are labeled as
$x_0(\xi),x_1(\xi),x_2 (\xi), \ldots$ in increasing order. Hence, we
can write
\[
\nu_t(f) = \mathbb{E}_{\mathcal{Q}} \bigl[ f \bigl(x_0
\bigl(\xi (t)\bigr) \bigr) \bigr] = \mathbb{E}_{\mathcal{Q}} \bigl[ \Theta
\bigl(\xi(t)\bigr) \bigr].
\]
Using that $\Theta\in\mathbb{D}$ and therefore \eqref{santi},
one concludes that the map $t \mapsto\nu_t(f)$
is differentiable and that
\begin{eqnarray*}
\frac{d}{dt} \nu_t(f) & =& \mathbb{E}_\mathcal{Q} \bigl[
\mathcal{L} \Theta\bigl(\xi(t)\bigr) \bigr]
\\
& =& \mathbb{E}_{\mathcal{Q}} \bigl[ \lambda_\ell
(x_1-x_0)\bigl[f(x_1)-f(x_0)
\bigr] +\lambda_a(x_1-x_0)
\bigl[f(x_2)-f(x_0)\bigr] \bigr]
\\
& =& \int\nu_t(dx) \int\mu_t(dy)\lambda_\ell(y)
\bigl[f(x+y)-f(x)\bigr]
\\
&&{} + \int\nu_t(dx) \int
\mu_t(dy) \int\mu_t(dz) \lambda_a(y)
\bigl[f(x+y+z)-f (x)\bigr],
\end{eqnarray*}
where we used, for simplicity of notation, $x_0=x_0(\xi(t))$,
$x_1=x_1(\xi(t))$ and $x_2=x_2(\xi(t))$ and the fact that $x_0$ has law
$ \nu_t$, while $x_1-x_0$ and $x_2-x_1$ have law $\mu_t$.
\end{pf}

\begin{pf*}{Proof of Theorem \ref{differiscobis}}
As in the proof of Theorem \ref{differisco} we can take $s>0$. Using
Proposition \ref{vongole} with the function $f\dvtx  x \mapsto
e^{-sx}$, we get that $t \mapsto L_t(s)$ is differentiable and that
\begin{eqnarray*}
\frac{d}{dt} L_t(s) & =& - \int \nu_t(dx) \int
\mu_t(dy) \bigl(\lambda_\ell(y) + \lambda_a(y)
\bigr) e^{-sx}
\\
&&{} +\int \nu_t (dx) \int \mu_t (dy)
\lambda_\ell(y) e^{-sx-sy}
\nonumber
\\
&&{} + \int\nu_t (dx) \int\mu_t (dy) \int
\mu_t (dz) \lambda_a (y) e^{-sx-sy-sz}.
\end{eqnarray*}
The above equation corresponds to \eqref{bucatini}.

Consider the case $\lambda_a\equiv0$. Then if $\lambda_\ell\equiv0$,
trivially $L_t=L_0$ for any $t \geq0$. While for
$\lambda_r\equiv\gamma\lambda_\ell$, one has
$\lambda_\ell\equiv\frac{1}{1+\gamma} \lambda$ so that the differential
equation \eqref{bucatini} satisfied by $L_t$ reads
\begin{eqnarray*}
\partial_t L_t(s) &=& \frac{L_t(s)}{1+\gamma} \biggl[ -\int
\mu_t(dy) \lambda(y) + \int \mu_t(dy)
\lambda(y)e^{-sy} \biggr]
\\
&=& \frac{L_t(s)}{1+\gamma} \bigl(
\partial_t H_t(0) - \partial_t
H_t(s) \bigr),
\end{eqnarray*}
where we used \eqref{davide}. Integrating and using that $\lim_{t \to
\infty} H_t (s) = 0$ leads to \eqref{monte1} and~\eqref{monte2}.

Now consider the case $\lambda_\ell\equiv0$$,  \lambda_r \equiv0$.
Noticing that $\lambda_a \equiv\lambda$ and using \eqref{davide},
from~\eqref{bucatini} we obtain that $
\partial_t \ln L_t(s) = \partial_t H_t(0)- G_t (s) \partial_t
H_t(s) $. At this point we apply point~\textup{(ii)} in Theorem
\ref{differisco} with $\gamma=0$ getting for $s>0$
\[
\partial_t \ln L_t(s)= \partial_t
H_t(0) - \frac{G_t(s)\partial_t
G_t(s) }{1-
G_t(s)^2} = \partial_t \biggl\{
H_t (0) +\frac{1}{2} \ln \bigl(1-G_t(s)^2
\bigr) \biggr\}.
\]
This leads to \eqref{gransasso1}, which implies \eqref{gransasso2}
after taking the limit $t\to\infty$.
\end{pf*}

%%%%%%%%%%%%%%%%%%%%%%%%%%%%%%%%%%%%%%%%%%%%%%%%%%%%%%%%%%%%%%%%%%%%%%%%%%%%%
%%%%%%%%%%%%%%%%%%%%%%%%%%%%%%%%%%%%%%%%%%%%%%%%%%%%%%%%%%%%%%%%%%%%%%%%%%%%%
%%%%%%%%%%%%%%%%%%% HCP
%%%%%%%%%%%%%%%%%%%%%%%%%%%%%%%%%%%%%%%%%%%%%%%%%%%%%%%%%%%%%%%%%%%%%%%%%%%%%
%%%%%%%%%%%%%%%%%%%%%%%%%%%%%%%%%%%%%%%%%%%%%%%%%%%%%%%%%%%%%%%%%%%%%%%%%%%%%
%%%%%%%%%%%%%%%%%%%%%%%%%%%%%%%%%%%%%%%%%%%%%%%%%%%%%%%%%%%%%%%%%%%%%%%%%%%%%

%s5 #&#
\section{\texorpdfstring{Abstract generalization of the transformation introduced in \cite{FMRT0}}
{Abstract generalization of the transformation introduced in [9]}}\label{trasformo}

% proof of Theorem \ref{teo2} and Theorem \ref{teo3}}
%In this section we prove Theorem \ref{teo2} and Theorem \ref{teo3}.
%This will be achieved using an abstract result that transforms the
%nonlinear identities on the Laplace transforms appearing in Theorem
%identities involving Radon measures. \club{\mathbf Ale\dvtx } \textit{nonsono
%completamente sicura che l'aggettivo ''linear'' sia ragionevole}

We extend here a~transformation developed in \cite{FMRT0}, Section~5,
allowing us to re\-phrase the nonlinear identities on the Laplace
transforms appearing in Theorems \ref{differisco}
and~\ref{differiscobis} into linear identities involving Radon
measures. This transformation will be crucial in our analysis of the
limiting behavior of the HCP process; see Sections~\ref{fegatograsso}
and~\ref{provateo3}.
%l'aggettivo ''linear'' sia ragionevole}
%revealing why the limit states for the interval law and the first
%point law of the HCP have a~similar structure when comparing case

Consider the OCP starting from a~renewal SPP with interval law $\mu$
having support on $[d_{\min}, \infty)$ [i.e., $\xi(0)$ has law
$\operatorname{Ren}(\nu, \mu) $ or $\operatorname{Ren} (\mu)$ or
$\operatorname{Ren}_\mathbb{Z}(\mu )$]. We recall that $\mu_\infty$
denotes the interval law at the end of the epoch (see Lemma~\ref{vecchioteo1}),
and we call $X_0, X_\infty$ some generic random
variables with law $\mu, \mu_\infty$, respectively. Then we define the
rescaled random variables
\[
Z_0=X_0/d_{\min}\quad\mbox{and}\quad
Z_\infty=X_\infty/d_{\max}
\]
and we set, for $s \geq0$,
\begin{eqnarray*}
g_0(s) &=& \mathbb{E} \bigl( e^{-s Z_0} \bigr),\qquad g_\infty
(s)=\mathbb{E} \bigl(e^{- s Z_\infty} \bigr),
\\
h_0 (s) &=& \mathbb{E} \bigl( e^{-s Z_0 }; Z_0< a \bigr),\qquad a:=
\frac{d_{\max}}{d_{\min}}.
\end{eqnarray*}
By definition and because of assumption \textup{(A2)} and Lemma
\ref{volarelontano}\textup{(iii)}, we have that $Z_0\ge1$, $
Z_\infty\ge1$ and $ a \in[1,2]$. In particular, $g_0(s),g_\infty(s)
\in(0,1)$ for $s>0$.

We observe that equations \eqref{uovone} and \eqref{salsona} have the
following common structure:
\[
\mathcal{F} \bigl( g_\infty(as ) \bigr) = \mathcal{F} \bigl(
g_0(s) \bigr) -h_0(s), \qquad s>0,
\]
where
%
%e34 #&#
\begin{equation}
\label{richiamo2} \mathcal{F}(x):= \cases{\displaystyle -\ln(1-x), &\quad for equation
\eqref{uovone},
\vspace*{3pt}\cr
\displaystyle\frac{\gamma+1}{\gamma+2}\ln\frac{1+x/(\gamma+1)}{1-x}, &\quad for equation
\eqref{salsona}.}
\end{equation}
With these examples in mind, we introduce the following
definition.

%de5.1 #&#
\begin{definition}[{[Hypothesis (H)]}] \label{cavolo}
We say that a~real function $\mathcal{F}$ satisfies hypothesis
\textup{(H)} if there exists $\varepsilon>0$ such that $\mathcal{F}$ is
defined on $(-\varepsilon,1)$ and:
\begin{longlist}[(H1)]
\item[(H1)] $\mathcal{F}$ is $C^1$;
\item[(H2)] the derivative $\mathcal{F}'$ admits an analytic expansion
    on $(0,1)$ of the form $\mathcal{F}'(x)=\sum_{n=0}^\infty c_n x^n $
    with $c_n \geq0$ for all $n \geq0$;
\item[(H3)] $\mathcal{F}(0)=0$, $\mathcal{F}$ is bijective from
    $(-\varepsilon, \varepsilon)$ to an open interval $U$
    containing $0$ such that $\mathcal{R}:= \mathcal {F}^{-1}\dvtx
    U \to(-\varepsilon, \varepsilon)$ is an analytic function and
    $\mathcal{R}'(0)=1$ [i.e., $\mathcal{F}'(0)=1$].
\end{longlist}
\end{definition}
By analytic expansion in \textup{(H3)} we mean that $\mathcal{R}(x)=
\sum_{k=1}^\infty r_k x^k$ for all $x \in U$, where the series on the
RHS is absolutely convergent.

One easily verifies that both functions $\mathcal{F}$ defined in
\eqref{richiamo2}
%two examples above
satisfy hypothesis \textup{(H)} since for $|x|<1$ we have the analytic
expansions $ -\ln(1-x)=x+x^2/2+ x^3/3+\cdots,$ while
%
%e35 #&#
\begin{equation}
\label{zittizitti} \ln\frac{ 1+x/(\gamma+1)}{1-x} = %\sum_{n=1}(-1)^{n+1} \frac{x^n}{n(1+\g)^n} + \sum_{n=1}^\infty
\sum
_{n=1}^\infty\frac{x^{n}}{n} \bigl[1+(-1)^{n+1}
(\gamma+1)^{-n} \bigr].
\end{equation}
%
%and $ a_n:=1+(-1)^{n+1} (\g+1)^{-n}> 0$, $a_1 =(\g+2)/(\g+1)$.
Moreover, for $|x| <1$, it holds
%
%e36 #&#
\begin{equation}
\label{spaziale} \mathcal{R}(x)= \cases{\displaystyle 1- e^{-x},
\cr
\qquad \mbox{if }\mathcal{F}(x)= -\ln(1-x),
\vspace*{6pt}\cr
\displaystyle\frac{\exp \{((\gamma+2)/(\gamma+1))x \} -1}{ \exp
\{((\gamma+2)/(\gamma+1))x \}+ 1/(\gamma+1)},
\vspace*{3pt}\cr
\qquad \mbox{if } \displaystyle\mathcal{F}(x)=\frac{\gamma+1}{\gamma+2}\ln\frac{ 1+x/(\gamma+1)}{1-x}.}
\end{equation}

Finally, we introduce the following notation. Given a~strictly
increasing function $\phi\dvtx [0,\infty) \rightarrow[0,\infty)$ and
a~Radon measure $\mathfrak{m}$ on $[0,\infty)$, we denote by
$\mathfrak{m}\circ\phi$ the new Radon measure on $[0,\infty)$ defined
by
\[
\mathfrak{m} \circ\phi(A)= \mathfrak{m}\bigl( \phi(A)\bigr), \qquad
A\subset\mathbb{R}\mbox{ Borel}.
\]
Note that $\mathfrak{m}\circ\phi$ is indeed a~measure, due to the
injectivity of $\phi$. Moreover, it holds
%
%e37 #&#
\begin{equation}
\label{regoletta} \int_0^\infty f(x) \mathfrak{m}
\circ\phi(dx)= \int_{[\phi(0),\phi(\infty)]} f\bigl( \phi^{-1} (x)
\bigr) \mathfrak{m}(dx).
\end{equation}
Above, and in what follows, we use the short notation $\int_0^\infty$
for $\int_{[0, \infty)}$.

%th5.2 #&#
\begin{Theorem}\label{astratto}
Let $\mathcal{F}$ be a~function satisfying hypothesis \textup{(H)}.
Then there exist unique Radon nonnegative measures $t_0 (dx)$ and
$t_\infty(dx)$ on $[0,\infty)$ such that for all $s >0$ it holds
%
%e38 #&#
%e39 #&#
%e40 #&#
\begin{eqnarray}
 \mathcal{F} \bigl( g_0(s) \bigr)&=& \int_0^\infty
\frac{e^{-s(1+x) }}{1+x} t_0 (dx),\label{sorriso1}
\\
\mathcal{F} \bigl(g_\infty(s) \bigr)&=& \int_0^\infty
\frac{e^{-s(1+x)
}}{1+x} t_\infty (dx),\label{sorriso2}
\\
 h_0(s)&=&\int_{[0,a-1) } \frac{e^{-s(1+x)}}{1+x}
t_0(dx).\label{sorriso3}
\end{eqnarray}
Moreover, the equation
%
%e41 #&#
\begin{equation}
\label{fidanzato} \mathcal{F} \bigl( g_\infty(a s ) \bigr) = \mathcal{F}
\bigl( g_0(s) \bigr) -h_0(s), \qquad s>0
\end{equation}
is equivalent to the
relation
%
%e42 #&#
\begin{equation}
\label{bianchini} t_\infty= (1/a) t_0 \circ\phi,
\end{equation}
where the linear function $\phi\dvtx [0,\infty) \rightarrow[0,\infty)$
is defined as $\phi(x)= a(1+x)-1$.
\end{Theorem}

%re5.3 #&#
\begin{remark}\label{codroipo}
Combining \textup{(H2)} and \textup{(H3)} in Definition \ref{cavolo},
it follows that the map $\mathcal{F}$ is strictly increasing on
$[0,1)$. In particular, equation \eqref{fidanzato} univocally
determines $g_\infty$ knowing $g_0$ and $h_0$ on $(0,\infty)$, and
similarly equations \eqref{sorriso1} and~\eqref{sorriso2} univocally
determine $g_\infty$ and $g_0$ knowing $t_\infty$ and $t_0$,
respectively.
\end{remark}

We divide the proof of the above theorem in different steps.
%
%le5.4 #&#
\begin{Lemma}\label{gola}
Let $Z$ be a~random variable such that $Z \geq1$ and define
$g(s)=\mathbb{E} [ e^{-sZ} ]$, $s \geq0$. Let $w\dvtx (0, \infty)
\rightarrow\mathbb{R}$ be the unique function such that
%
%e43 #&#
\begin{equation}
\label{pompiere0} \mathcal{F} \bigl( g(s) \bigr) = \int_s
^\infty\,du e^{-u} w(u),\qquad s>0,
\end{equation}
that is,
%
%e44 #&#
\begin{equation}
\label{pompiere1} w (s):= -e^s \mathcal{F}' \bigl(g(s)
\bigr) g'(s), \qquad s >0.
\end{equation}
Then the function $w$ is completely monotone.\footnote{Recall that
a~function $f\dvtx  (0,\infty ) \to\mathbb{R}$ is said to be completely
monotone if it is $\mathcal {C}^\infty$ and if for any integer $k$,
$(-1)^k f^{(k)} \geq0$.} In particular, there exists a~unique Radon
measure $t(dx)$ on $[0, \infty)$ (not necessarily of finite total mass)
such that
%
%e45 #&#
\begin{equation}
\label{cotoletta} w(s)= \int_0 ^\infty
e^{-sx} t(dx), \qquad s>0
\end{equation}
and therefore
%
%e46 #&#
\begin{equation}
\label{animainpena} \mathcal{F} \bigl(g(s) \bigr)= \int_0^\infty
\frac{e^{-s(1+x)
}}{1+x} t (dx), \qquad s>0.
\end{equation}
%
% Moreover,
% \begin{equation}\label{fuoco}
% \limsup_{s \downarrow0}-\frac{s g'(s)}{1-g (s)}\in[0,1].
% \end{equation}
Moreover, the above identity \eqref{animainpena} univocally
determines $t(dx)$.
\end{Lemma}

\begin{pf} %\club\club
The last statement follows from the inversion formula of the Laplace
transform. For the rest, the proof is similar to the proof of Lemma 5.1
in \cite{FMRT0}. The only slight difference is in the following
argument. By condition \textup{(H2)} and since $g(s) \in(0,1)$ for
$s>0$, we can write $w=f \sum_{k=0}^\infty c_k g^k$, $f =-e^s g'(s)$.
Since $c_k \geq0$ for all $k \geq0$ and since the product and the sum
of completely monotone functions is again completely monotone (cf.
\cite{F}) we get that $ \sum_{k=0}^\infty c_k g^k$ is completely
monotone. The rest of the proof is as in \cite{FMRT0}.
% On the other hand, $f$ is completely monotone. Indeed, by
% the Leibniz rule one has $D^n f (s)= e^s \bbE\left( e^{-s Z}
% Z(1-Z)^ n \right)$ (see the proof of Lemma 5.1 in \cite{FMRT0}). From
%this formula and the assumption $Z\geq1$, we conclude that
% $(-1)^n D^n f(s)
% \geq0$ for $s >0$.
% We can apply Theorem 1a in Section XIII.4 of
%determined due to the \emph{inversion formula} given in Theorem 2,
%Section XIII.4 \cite{F}. Finally, we derive \eqref{animainpena} for
%$s
%>0$ from \eqref{pompiere0}, \eqref{cotoletta} and the identity
%e^{-u} \int_0 ^\infty e^{-u x} t (dx), \qquad s>0.
\end{pf}

%le5.5 #&#
\begin{Lemma}\label{limonata}
Let $Z$ be a~random variable such that $Z \geq1$, and let $g(s)$ be its
Laplace transform. Let $t$ be the unique Radon measure on $[0,\infty)$
satisfying \eqref{animainpena} and call $m(dx)$ the Radon measure with
support in $[1,\infty)$ such that
%
%e47 #&#
\begin{equation}
\label{bilancia} m (A)= \int_0 ^\infty
\frac{\mathbh{1}_{1+x \in A} }{1+x} t(dx).
\end{equation}
For each $k \geq1$, consider the convolution measure $ m ^{(k)}$
with support in $[k,\infty)$ defined as
%
%e48 #&#
\begin{equation}\qquad
m^{(k)} (A)= \int_1^\infty
m(dx_1)\int_1^\infty
m(dx_2)\cdots\int_1^\infty
m(dx_k) \mathbh{1}_{x_1 +x_2+ \cdots+ x_k\in A}.
\end{equation}
Then the law of $Z$ is given by the measure
$m_*:=\sum_{k=1}^\infty r_k m^{(k)}$,
where the coefficients $r_k$ are determined by the series expansion
$\mathcal{R}(x)= \sum_{k=1}^\infty r_k x^k$ of the function $\mathcal
{R}$ around $0$ [recall condition \textup{(H3)} in Definition
\ref{cavolo}]. In particular,
%
%e49 #&#
\begin{equation}
\label{brucogiallo} \mathbb{E} \bigl[ e^{-sZ}; Z< a\bigr]= \int
_{[0,a-1) } \frac
{e^{-s(1+x)}}{1+x} t(dx),\qquad s \geq0.
\end{equation}
\end{Lemma}

We point out that, given a~bounded Borel set $A$, since $m^{(k)}$ has
support in $[k,\infty)$, the series $m_*(A)=\sum_{k=1}^\infty r_k
m^{(k)}(A) $ is a~finite sum.
%The thesis includes that this
%sum is a~nonnegative number and that the set--function $A\mapsto
%m_*(A)$, defined on bounded Borel sets, extend univocally to a~Radon
%measure on all Borel sets.

\begin{pf*}{Proof of Lemma \ref{limonata}}
The proof is a~generalization of the proof of Lemma~5.2 in
\cite{FMRT0}; hence we give only a~sketch.

By definition of $m(dx)$ and
by \eqref{animainpena} we can write
%
%e50 #&#
\begin{equation}
\label{succhetto} \mathcal{F} \bigl( g(s) \bigr)= \int_0^\infty
\frac{e^{-s(1+x)
}}{1+x} t (dx) = \int_0^\infty
e^{-sx} m(dx), \qquad s>0.
\end{equation}
Since $\lim_{s \to\infty} g(s)=0$, by \textup{(H3)} we conclude that
$\mathcal{F} (g(s) )$ goes to zero as $s$ goes to $\infty$. In
particular, by \textup{(H3)}, for $s$ large enough we can invert
\eqref{succhetto} and use the analytic expansion of $\mathcal{R}$
getting
%
%e51 #&#
\begin{equation}
\label{asiagoDOP} \quad g(s) = \mathcal{R} \biggl(\int_0^\infty
e^{-sx} m(dx) \biggr)= \sum_{k=1}^\infty
r_k \int_0 ^k e^{-s x}
m^{(k)} (dx), \qquad s\mbox{ large},
\end{equation}
where, in the last equality, we used that $ ( \int_0^\infty e^{-sx}
m(dx)  )^k = \int_0^\infty e^{-sx } m^{(k)} (dx)$. Reasoning as in the
proof of Lemma 5.2 in \cite{FMRT0} and using the analytic expansion of
$\mathcal{R}(x)$ around $0$, we get that $g(s) = \int_{[0,\infty)}
e^{-sx }m_*(dx) $. As in \cite{FMRT0} we can conclude that $m_*$ is
a~nonnegative measure.

To complete the proof it remains to check \eqref{brucogiallo}. It is
enough to prove the thesis for $s >0$, since the case $s=0$ follows by
monotonicity. To this aim we observe that, since $m^{(k)}$ has support
contained in $[k, \infty)$ and since $r_1= \mathcal{R}'(0)=1$ by
\textup{(H3)}, the measure $m_*$ equals $m$ on $[1, 2)$. Since $a\leq2$
and using the definition of the measure $m$ given by \eqref{bilancia},
we obtain that
\begin{eqnarray*}
\mathbb{E} \bigl[ e^{-sZ}; Z< a\bigr]&=& \int_{ [1,a) } e^{-sx }
p_Z(dx)= \int_{ [1,a) } e^{-sx } m(dx)
\\
&=& \int_{[0,a-1)}
\frac{e^{-s(1+x)}}{1+x} t(dx).
\end{eqnarray*}
This completes the proof of \eqref{brucogiallo}.
\end{pf*}

We are now in position to prove Theorem \ref{astratto}.

\begin{pf*}{Proof of Theorem \ref{astratto}}
Observe that equations \eqref{sorriso1} and \eqref{sorriso2} follow
from Lemma \ref{gola}, and that equation \eqref{sorriso3} follows from
\eqref{brucogiallo} in Lemma \ref{limonata}.

To prove the last statement we write $\rho(dx) $ for the measure on the
RHS of \eqref{bianchini}. Using that $a[\phi^{-1}(x)+1]= 1+ x $, we
obtain for $s\geq0$ that
\begin{eqnarray*}
\int_0^\infty\frac{e^{-a s(1+x) }}{1+x}\rho(dx) &=& a^{-1}
\int_{[\phi(0), \infty)}\frac{e^{-s (1+x)}}{a^{-1}(1+x)}t_0(dx)
\\
&=& \int_{[a-1,\infty)}\frac{e^{-s (1+x)}}{1+x} t_0 (dx).
\end{eqnarray*}
Using also \eqref{sorriso1}, \eqref{sorriso2}, \eqref{sorriso3} we
conclude that
equation \eqref{fidanzato} is equivalent to
%
%e52 #&#
\begin{equation}
\label{petrolio} \int_0^\infty
\frac{e^{-as(1+x) }}{1+x}t_{\infty}(dx) = \int_0^\infty
\frac{e^{-a s(1+x)}}{1+x}\rho(dx) \qquad\forall s > 0.
\end{equation}
Thinking the above integrals as Laplace transforms of suitable
nonnegative measures in the variables $as$, by Theorem 1a in Section
XIII.1 in \cite{F} we conclude that \eqref{petrolio} is equivalent to
the identity $t_\infty= \rho$.
\end{pf*}

%%%%%%%%%%%%%%%%%%%%%%%%%%%%%%%%%%%%%%%%%%%%%%%%%%%%%%%%%%%%%%%%%%%%%%%%%%%%%%%%%%%%%%%%%%%%%%%%%%%%%%%%
%%%%%%%%%%%%%%%%%%%%%%%%%%%%%%%%%%%%%%%%%%%%%%%%%%%%%%%%%%%%%%%%%%%%%%%%%%%%%%%%%%%%%%%%%%%%%%%%%%%%%%%%

%s6 #&#
\section{\texorpdfstring{Asymptotic of the interval law for HCP: Proof of Theorem \protect\ref{teo2}}
{Asymptotic of the interval law for HCP: Proof of Theorem 2.12}} \label{fegatograsso}
%We use notation and definitions of Theorem \ref{teo2}.
% \club{\mathbf Ale:} \textit{prima venivano richiamate alcune defizioni
% del Teorema \ref{teo2} ma vi erano alcune imprecisioni. penso che
% la cosa naturale e' che il lettore si rilegga il Th. \ref{teo2} il
% cui statement e' preciso e breve}
% and Theorem \ref{teo3}. In particular $\mu$ and $\nu$ are probability
%measures
%on $[1,\infty)$ and $\mathbb{R}$ respectively. We define here
%$X^{(n)}$, $n \in\bbN_+$, as the length of the leftmost domain
%inside $(0,\infty)$ at the beginning of the $n$--th epoch, i.e.
%$X^{(n)}= x^{(n)} _2 (0)- x^{(n)}_1 (0)$. Moreover we set
%$Z^{(n)}=X^{(n)}/d^{(n)}$. Note that $X^{(n)}$ has law $\mu^{(n)}$.
% is a~random variable with law $\mu^{(n)}$, $n\in\bbN_+$, and
% Also,
%starting indifferently from $\cQ=\operatorname{Ren}(\nu,\mu)$,
%$\cQ=\operatorname{Ren}(\mu)$ or $\cQ=\operatorname{Ren}_{\mathbb{Z}}(
%In addition,

The key result of this section is Theorem \ref{astrattobis} which in
turn allows us to prove easily Theorem \ref{teo2}. Theorem
\ref{astrattobis} is proved by using the recursive identities for the
OCP process established in Section~\ref{prova_differisco} and our
extension of the transformation of~\cite{FMRT0} derived in the previous
section.

Let us start by recalling the notation of Theorem \ref{teo2} which will
be used throughout this section and by giving a~few more definitions.
We let $\mu$ be a~probability measure on $[d^{(1)}, \infty)=[1,\infty)$
and consider the HCP such that $\xi^{(1)}(0)$ has law of the form
$\operatorname{Ren} (\nu, \mu)$, $\operatorname{Ren}(\mu)$ or
$\operatorname{Ren}_\mathbb{Z}(\mu)$ ($\nu$ being a~probability measure
on~$\mathbb{R}$). Call $\mu^{(n)}$ the interval law of $\xi^{(n)}(0) $,
that is, at the beginning of epoch $n$, and let $X^{(n)}$ be a~generic
random variable with law $\mu^{(n)}$ and $Z^{(n)}$ be the rescaled
variable $Z^{(n)} = X^{(n)} /d^{(n)}$. Finally, for any $n \geq1 $ and
$s \geq 0$ set
%
%e53 #&#
\begin{equation}
\label{sciopero} g^{(n)}(s):=\mathbb{E} \bigl(e^{-sZ^{(n)}} \bigr),
\qquad h^{(n)}(s):=\mathbb{E} \bigl(e^{-sZ^{(n)}}
\mathbh{1}_{1 \leq Z^{(n)}
<a_n} \bigr)
\end{equation}
with
%
%e54 #&#
\begin{equation}
\label{sciopero2}a_n:=d^{(n+1)}/d^{(n)}.
\end{equation}
Note that $\mu= \mu_1$ and $g^{(1)}(s)= g(s): = \int e^{-sx} \mu(dx)$.
The following holds:

%th6.1 #&#
\begin{Theorem}\label{astrattobis}
Let $\mathcal{F}$ be a~function satisfying hypothesis \textup{(H)} (see
Definition~\ref{cavolo}), and assume that for some number $\kappa$ it
holds
%
%e55 #&#
\begin{equation}
\label{baseperarrosto} \lim_{s\downarrow0} -s \mathcal{F}' \bigl(
g(s) \bigr) g'(s) = \kappa
\end{equation}
%
%} $\lim_{x \to\infty} \cF^{-1}(x) = 1$, that $\lim_{x \to1}
and that
%
%e56 #&#
\begin{equation}
\label{mennea} \mathcal{F}\bigl(g^{(n+1)}(a_n s)\bigr) =
\mathcal{F}\bigl(g^{(n)}( s)\bigr) - h^{(n)}(s), \qquad n \geq1, s >0.
\end{equation}
Then it must be $\kappa\geq0$. Moreover, the rescaled variable $Z^{(n)}
$ weakly
converges to the random variable $Z^{(\infty)}\equiv Z^{(\infty
)}_{\kappa}$
whose Laplace transform $g^{(\infty)}_\kappa$ satisfies
%
%e57 #&#
\begin{equation}
\label{radicchio} \mathcal{F} \bigl( g^{(\infty)}_{\kappa}(s) \bigr) =
\kappa\int_1^\infty \frac{e^{-sx}}{x} \,dx, \qquad
s >0.
\end{equation}
If $\kappa=0$, then $Z^{(\infty)}_{\kappa}= \infty$, while if
$\kappa>0$, then
$Z^{(\infty)}_\kappa$ takes value in $[1, \infty)$.
%where $\beta= \lim_{x \to1} \cF'(x)(1-x)$.
\end{Theorem}

\begin{pf}
We first apply Theorem~\ref{astratto} getting that, for each $n \geq1$,
there exists a~unique measure $t^{(n)}$ on $[0,\infty)$ such that
%
%e58 #&#
\begin{equation}
\label{asilobello} \mathcal{F} \bigl( g^{(n)}(s) \bigr) = \int
_0^\infty \frac{e^{-s(1+x)}}{1+x} t^{(n)}(dx), \qquad s>0.
\end{equation}
Due\vspace*{-1pt} to \eqref{mennea} and Theorem~\ref{astratto} again,
for $n \geq 2$ it holds $t^{(n)} = \frac{1}{a_{n-1}}t^{(n-1)}
\circ\phi_{n-1}$, with $\phi_{n-1}(x)=a_{n-1}(1+x)-1$. The recursive
identities relating the $t^{(n)}$'s can be explicitly solved, leading
to
%
%e59 #&#
\begin{equation}
\label{brumsporcobis} t^{(n)} = \frac{1}{d^{(n)}} t^{(1)} \circ
\psi_{n-1}, \qquad n \geq2
\end{equation}
with $\psi_{n-1}(x)=d^{(n)} (1+x)-1$. %\club
Defining $ U^{(n)}(x)= t^{(n)} ([0,x]) \mathbh{1}(x \geq0)$,
we get that %\club
$dU^{(n)} = t^{(n)}$ and $U^{(n)}(x)=0$ for $x<0$. By
\eqref{brumsporcobis} it holds that
%
%e60 #&#
\begin{equation}
\label{pescefritto} % \begin{split}
\qquad U^{(n)}(x) %& = \frac{1}{d^{(n)}}\Big[ U^{(1)}(\psi_{n-1}
%(x))-
%U^{(1)}(\psi_{n-1} (0)-)\Big],\\
=\frac{1}{d^{(n)}} \bigl[ U^{(1)} \bigl(
d^{(n)} (1+x)-1 \bigr)- U^{(1)} \bigl(\bigl( d^{(n)}-1
\bigr)- \bigr) \bigr], \qquad n\geq1. %\end{split}
\end{equation}
Moreover, for each $n \geq1$, integrating by parts and using that
$U^{(n)}(0-)=0$, we can rewrite the integral on the RHS of
\eqref{asilobello} as
%
%e61 #&#
\begin{eqnarray}\label{ananas} %\mbox{\club}
&& \int_0^\infty
\frac{e^{-s(1+x) }}{1+x} t^{(n)}(dx) %& =
\nonumber\\[-8pt]\\[-8pt]
&&\qquad = \lim
_{y \uparrow\infty} \frac{e^{-s(1+y) }}{1+y} U^{(n)}(y)- \int
_0^\infty \biggl( \frac{d}{dx} \biggl(
\frac{e^{-s(1+x) }}{1+x} \biggr) \biggr) U^{(n)}(x) \,dx.\nonumber
\end{eqnarray}
We now use the key additional hypothesis \eqref{baseperarrosto}. Since
$g^{(1)}(s)= g(s)$ because $d^{(1)}=1$, if $w^{(1)}$ denotes the
Laplace transform of $t^{(1)}$ [i.e., $w^{(1)}(s) = \int_0^\infty e^{-s
x} t^{(1)}(dx) $], then \eqref{baseperarrosto}
together with \eqref{pompiere1} implies %\club
that $ \lim_{s\downarrow0} s  w^{(1)}(s)= \kappa$. The above limit and
the Tauberian Theorem 2 in Section XIII.5 of \cite{F} allow us to
conclude that
%
%e62 #&#
\begin{equation}
\label{marina} \lim_{y\uparrow\infty} \frac{ U^{(1)}(y)}{y}= \kappa.
\end{equation}
The above limit together with \eqref{pescefritto} implies that there
exists a~suitable constant $C>0$ such that
%
%e63 #&#
\begin{equation}
\label{bound} U^{(n)}(x) \leq C (1+x), \qquad n\geq1, x \geq0.
\end{equation}
In
particular, the limit on the RHS of \eqref{ananas} is zero and
%
%e64 #&#
\begin{equation}
\label{miserve} \int_0^\infty\frac{e^{-s(1+x) }}{1+x}
t^{(n)}(dx)=- \int_0^\infty \biggl(
\frac{d}{dx} \biggl( \frac{e^{-s(1+x) }}{1+x} \biggr) \biggr) U^{(n)}(x)
\,dx,\qquad n\geq1.\hspace*{-35pt}
\end{equation}
By \eqref{pescefritto}, \eqref{marina} and the fact that $d^{(n)} \to
\infty$, we conclude that $\lim_{n\to\infty}U^{(n)}(x) = \kappa x$
for all
$x \geq0$. This limit together with \eqref{bound} allows us to apply
the dominated convergence theorem, getting that
%
%e65 #&#
\begin{eqnarray}\label{monnalisa}
&& \lim_{n\to\infty} \int_0^\infty
\frac{e^{-s(1+x) }}{1+x} t^{(n)}(dx)
\nonumber\\[-8pt]\\[-8pt]
&&\qquad = -\kappa\int_0
^\infty \biggl( \frac{d}{dx} \biggl( \frac{e^{-s(1+x) }}{1+x} \biggr)
\biggr) x \,dx = \kappa\int_0 ^\infty
\frac{e^{-s (1+x) }}{1+x} \,dx\nonumber
\end{eqnarray}
(in the last identity we have simply integrated by parts).

Let us come back to \eqref{asilobello}. We know the limit of the
RHS as $n\to\infty$ by \eqref{monnalisa}. Let us analyze the LHS.
% Since $Z_n \geq1$ for any $n$, given $s >0$ it holds
%$e^{g^{(n)}(s)\leq e^{-s} \in(0,1)$ for any $n \geq1$.
We claim that, given $s>0$, the sequence $\{g^{(n)}(s)\}_{n\geq1}$
converges to some number in $[0,e^{-s}]$. Indeed, since $Z_n \geq1$, it
holds $ g^{(n)}(s) \in(0, e^{-s}]$. If the sequence was not convergent,
by compactness we could find two subsequence $\{n_k\}_{k\geq1}$ and
$\{n_r\}_{r\geq1}$ such that $ \lim_{k \to \infty}g^{(n_k)}(s)< \lim_{r
\to\infty}g^{(n_r)}(s)$ and both limits exist and belong to $[0,
e^{-s}]$. On the other hand, by hypothesis \textup{(H1)} and Remark
\ref{codroipo}, the function $\mathcal{F}$ is continuous and strictly
increasing on $[0,1)$. Hence
\[
\lim_{k \to\infty}\mathcal{F} \bigl(g^{(n_k)}(s) \bigr)=
\mathcal{F} \Bigl(\lim_{k \to\infty}g^{(n_k)}(s) \Bigr)<
\mathcal {F} \Bigl(\lim_{r \to\infty}g^{(n_r)}(s) \Bigr)=\lim
_{r \to
\infty}\mathcal{F} \bigl(g^{(n_r)}(s) \bigr)
\]
in contradiction with the fact
that the first member and the last member equal the RHS of
\eqref{monnalisa}, by \eqref{asilobello} and \eqref{monnalisa}.

Since we have proved that for all $s >0$ the sequence
$\{g^{(n)}(s)\}_{n \geq1}$ converges to some number
$g^{(\infty)}_{\kappa}(s)\in[0, e^{-s}]$, using the continuity of
$\mathcal{F}$ on $[0,1)$, \eqref{asilobello} and~\eqref{monnalisa}, we
conclude that $g^{(\infty)}_\kappa$ satisfies \eqref{radicchio}.

%Equation \eqref{asilobello} and the limit \eqref{monnalisa} imply
%that, on $(0,\infty)$, $g^{(n)} $ converges pointwise to the a~well
%defined function $g^{(\infty)}_{\k}(s)$, which satisfies

Since by hypothesis \textup{(H)} the function $\mathcal{F}'$ is
positive on $[0,1)$, the limit $\kappa$ in~\eqref{baseperarrosto} must
be nonnegative. Let us first consider the case $\kappa=0$. Then, by
\eqref{radicchio}, the fact that $\mathcal{F}$ is strictly increasing
on $[0,1)$ and $\mathcal {F}(0)=0$, we conclude that
$g^{(\infty)}(s)=0$ for all $s>0$. This implies that the law of the
random variable $Z^{(n)}$ weakly converges to $\delta_\infty$.

We now consider the case $\kappa>0$. As pointwise limit of decreasing
functions, also $g^{(\infty)}_\kappa$ is decreasing on $(0,\infty)$. In
particular the limit $\lim_{s \downarrow0} g^{(\infty)}_\kappa(s)$
exists and belongs to $[0,1]$. Let us call $z$ this limit and prove
that $z=1$. Suppose by absurd that $z \in[0,1)$. Then by the continuity
of $\mathcal{F}$ on $[0,1)$ and equation~\eqref{radicchio}, we would
have
\[
\mathcal{F}(z)= \lim_{s \downarrow0} \mathcal{F} \bigl(g^{(\infty
)}_\kappa(s)
\bigr) = \lim_{s \downarrow0} \kappa\int_1
^\infty \frac{e^{-sx}}{1+x} \,dx = \infty.
\]
Since $\mathcal{F}$ takes finite value on $[0,1)$ it cannot be
$\mathcal{F}(z)=\infty$, thus implying that $z =1$. In conclusion we
have proved that $\lim_{s\downarrow0 } g^{(\infty)}_{\kappa}(s)=1$.
Then, by Theorem~2 in Section XIII.1 of \cite{F}, we conclude that
$g^{(\infty)}_{\kappa}$ is the Laplace transform of some nonnegative
(finite) random variable $Z _{\kappa}^{(\infty)}$ and that $Z^{(n)}$
weakly converges to $Z _{\kappa}^{(\infty)}$. The fact that $Z
_{\kappa}^{(\infty)}\geq1$ a.s. follows from the fact $Z ^{(n)}\geq1$
for all $n \geq1$.
\end{pf}

\begin{pf*}{Proof of Theorem~\ref{teo2}}
Thanks to Theorem~\ref{differisco} and the discussion before Definition
\ref{cavolo}, the Laplace transforms of the rescaled variables
$Z^{(n)}$ satisfy
\[
\mathcal{F}\bigl(g^{(n+1)}(a_n s)\bigr) = \mathcal{F}
\bigl(g^{(n)}( s)\bigr) - h^{(n)}(s) \qquad\forall n \geq1,\ \forall s > 0,
\]
where
\[
\mathcal{F}(x) = \cases{\displaystyle - \ln(1-x), &\quad in case \textup{(i)},
\vspace*{5pt}\cr
\displaystyle\frac{\gamma+1}{\gamma+2} \ln \frac{1 +x/(\gamma+1)}{1-x}, &\quad in case \textup{(ii)},}
\]
respectively. We have already observed that in both cases $\mathcal{F}$
satisfies the hypothesis~\textup{(H)}. Computing $\mathcal{F}'$ we get
% Since, for $x \in[0,1)$,
%$$ \cF'(x)=
% \frac{1}{1-x} & \mbox{ in case \textup{(i)}},\\
% \frac{\gamma+1}{\gamma+2} \left( \frac{1}{\g+1+x}+\frac{1}{1-x}
% \right)&
% \mbox{ in case \textup{(ii)}},
%$$
%using also that $\lim_{s \downarrow0} g(s)=1$, we have
\begin{eqnarray*} %\label{mammina}
&& \lim_{s \downarrow0 } -s \mathcal{F}'\bigl(g(s)
\bigr) g'(s)
\\
&&\qquad = \cases{\displaystyle -\lim_{s \downarrow0}
\frac{s g'(s)}{1-g(s)}, &\quad in case \textup{(i)},
\vspace*{5pt}\cr
\displaystyle -\lim_{s \downarrow0}
\frac{\gamma+1}{\gamma+2} \biggl( \frac{s g'(s)}{\gamma+1+g(s)}+\frac{s g'(s)}{1-g(s)} \biggr), &\quad
in case \textup{(ii)}.}
\end{eqnarray*}
Since we have assumed the limit~\eqref{lim_arrosto} and since
$1-g(s)=o(1)$ for $s$ small, it must be $\lim_{s \downarrow0}
sg'(s)=0$. This last observation allows us to conclude that
%
%e66 #&#
\begin{equation}
\label{mammina} \lim_{s \downarrow0 } -s \mathcal{F}'
\bigl(g(s) \bigr) g'(s)= %
\cases{ c_0, &
\quad in case \textup{(i)},
\vspace*{3pt}\cr
\displaystyle\frac{\gamma+1}{\gamma+2} c_0, &\quad in case
\textup{(ii)}.}
\end{equation}
At this point Theorem~\ref{teo2} is an immediate consequence of Theorem
\ref{astrattobis} and the computation of $\mathcal{R}= \mathcal
{F}^{-1}$ given in~\eqref{spaziale}.
\end{pf*}

%% We point out that in the general case where $d^{(1)}$ is not
%necessarly equal to $1$,
%It is useful to observe that if the initial scale $d^{(1)}$ was
%different from one than necessarily $g(s)\neq g^{(1)}(s)$. However,
%and that is the reason why we could fix $d^{(1)}=1$, the limit
%a constant \ie\eqref{lim_arrosto} for $g$ implies the same limit
%for $g^{(1)}$. {\mathbf Ale}: \textit{questo remark va migliorato}

%s7 #&#
\section{\texorpdfstring{Asymptotic of the first point law: Proof of Theorem \protect\ref{teo3}}
{Asymptotic of the first point law: Proof of Theorem 2.15}}\label{provateo3}

In this section we prove Theorem~\ref{teo3}. While in the derivation of
Theorem~\ref{teo2} we have tried to keep the discussion at a~general
and abstract level in order to catch the fundamental structure of the
transformation introduced in \cite{FMRT0} and therefore explain the
similar asymptotics of very different HCPs, we restrict here to the
special cases mentioned in Theorem~\ref{teo3}. Indeed, as the reader
will see, the proof goes through estimates which are very
model-dependent.

\begin{pf*}{Proof of Theorem~\ref{teo3}}
Case \textup{(i)} has been solved in \cite{FMRT0}, Theorem 2.24. Hence
we focus on case \textup{(ii)}. Without loss of generality we can
restrict to the case $\nu=\delta_0$, that is, when the HCP starts with
$\xi^{(1)} (0)$ having law $\operatorname{Ren} (\delta_0, \mu)$, $\mu$
being a~probability measure on $[d^{(1)}, \infty)=[1, \infty)$. Indeed,
in the general case $X^{(n)}_0$ can be expresses as $V+ \bar
X^{(n)}_0$, where $\bar X^{(n)}_0$ is the first point in $\xi^{(n)}_0$
for the above HCP starting with distribution
$\operatorname{Ren}(\delta_0, \mu)$, while $V$ is a~random variable
with law $\nu$ independent of $\bar X^{(n)}_0$. Since $d^{(n)}
\to\infty$, when taking the rescaled random variable $Y^{(n)}=
X^{(n)}_0/ d^{(n)}$, the effect of the random translation $V$
disappears as $n \to\infty$. From Lemma~\ref{vecchioteo1} we know that
the configuration $\xi^{(n)}(0)$ at the beginning of epoch $n$ has law
$\operatorname{Ren} ( \nu^{(n)}, \mu^{(n)} )$. As in the previous
section $X^{(n)}$ will be a~random variable with law $\mu^{(n)}$ and
$Z^{(n)}$ the rescaled random variable $Z^{(n)}= X^{(n)}/ d^{(n)}$.
Moreover, we write $X^{(n)}_0$ for a~generic random variable with law
$\nu^{(n)}$ and set
\[
\ell^{(n)}(s)= \mathbb{E} \bigl( e^{ -s Y^{(n)}} \bigr), \qquad s \in
\mathbb{R}_+, \qquad Y^{(n)}= X^{(n)}_0
/d^{(n)}.
\]
Recalling the definitions of $g^{(n)}(s)$, $h^{(n)}(s)$ and $a_n$ in
equations~\eqref{sciopero} and~\eqref{sciopero2}, we use formula
\eqref{convoluto} and Theorem~\ref{differiscobis}\textup{(ii)} to
obtain the recursive equations
\[
\ell^{(n+1)}(a_n s) = \ell^{(n)}(s) \sqrt{
\frac{1-g^{(n+1)}(a_n
s)^2}{1-g^{(n)}(s)^2}} e^{-h^{(n)}(0)}, \qquad n\geq1.
\]
By iteration, we get
%
%e67 #&#
\begin{eqnarray}
\label{pescedaprile}
\ell^{(n)}(s) & =& \ell^{(1)}
\bigl(s/d^{(n)}\bigr)\sqrt{ \frac{1-g^{(n)}(s)^2}{1-g^{(1)}(s/d^{(n)})^2} } \exp \Biggl\{ - \sum
_{j=1}^{n-1} h^{(j)}(0) \Biggr\}\nonumber
\\
& =& \ell^{(1)}\bigl(s/d^{(n)}\bigr)\sqrt{ \frac{1-g^{(n)}(s)^2}{s}}
\sqrt{ \frac
{s/d^{(n)}}{1-g^{(1)}(s/d^{(n)})^2}}
\\
&&{}\times \exp \Biggl\{ \frac{1}{2} \log d^{(n)} -
\sum_{j=1}^{n-1} h^{(j)}(0) \Biggr\}.\nonumber
\end{eqnarray}
%
%where
%$$
%F^{(n)}(s):= \sum_{j=1}^{n-1} h^{(j)}(sd^{(j)}/d^{(n)}).
%$$
Since $d^{(n)} \to\infty$, we have $\lim_{n \to\infty}
\ell^{(1)}(s/d^{(n)}) =1$. By assumption, $\mu$ has finite mean $\bar
\mu= -g'(0)$. Hence$,  \frac{s/d^{(n)}}{1-g^{(1)}(s/d^{(n)})^2}$
converges to $1/ 2\bar\mu$ as $n \to\infty$. Finally, invoking
Theorem~\ref{teo2} (see also Remark~\ref{referee}), from
\eqref{pescedaprile} we get
%
%e68 #&#
\begin{eqnarray}\label{balena}
\qquad && \lim_{n \to\infty} \ell^{(n)}(s)
\nonumber\\[-8pt]\\[-12pt]
&&\qquad =\frac{1}{\sqrt{2 \bar\mu}} \sqrt{ \frac{1-\tanh^2 ( \mathrm{Ei}(s)/2  ) }{s}} \lim_{n\to\infty}\exp
\Biggl\{ \frac{1}{2} \log d^{(n)} - \sum
_{j=1}^{n-1} h^{(j)}(0) \Biggr\}.\nonumber
\end{eqnarray}
It remains to study the last limit in~\eqref{balena}. To this aim we
come back to the measures~$t^{(n)}$. As already observed in the proof
of Theorems~\ref{astrattobis} and~\ref{teo2}, applying Theorem
\ref{astratto} one gets that for each $n\geq1$ there exists a~unique
measure $t^{(n)}$ on $[0,\infty)$ satisfying~\eqref{asilobello} with
$\mathcal{F}(x)=\frac{1}{2} \ln\frac{1 +x}{1-x}$
%%%%%%%%%%%%%\frac{\gamma+1}{\gamma+2} \ln\frac{1 +\frac{x}{\g+1}}{1-x}
for $x \in[0,1)$. Moreover, by equation~\eqref{sorriso3} it holds
\[
h^{(n)}(s) = \int_{[0, a_n-1)} \frac{e^{-s (1+x)}}{1+x}
t^{(n)}(dx)
\]
and by formula~\eqref{brumsporcobis} it holds
%for each $n \geq2$ it holds
% $t^{(n)} = \frac{1}{a_{n-1}}t^{(n-1)} \circ\Phi_{n-1}$, with $
$t^{(n)} =  (1/d^{(n)} ) t^{(1)} \circ\psi_{n-1}$ with
$\psi_{n-1}(x)=d^{(n)} (1+x)-1$, for all $n\geq2$. Combining the last
identities, from~\eqref{regoletta} one gets
\[
\label{nuvolettabis} h^{(n)}(s) = \int_{[d^{(n)}-1,
d^{(n+1)}-1)}
\frac{ e^{-s(1+x)/d^{(n)} }}{1+x} t^{(1)} (dx), \qquad n \geq1.
\]
The above integral representation implies
\[
\sum_{j=1}^{n-1} h^{(j)}(0) = \int
_{[0, d^{(n)}-1)} \frac{1}{1+x}t^{(1)} (dx).
\]
%
%h^{(j)}( s  d^{(j)}/d^{(n)}) = \int_{[d^{(j)}-1, d^{(j+1)}-1)}
%1.
%This allows to write
%F^{(n)}(s)= \sum_{j=1}^{n-1}
% h^{(j)}(s d^{(j)}/d^{(n)}) =
%(dx).
%integration by parts and the change of variable $y=(1+x)/d^{(n)}$ to
%conclude that \begin{equation}\label{brodino} F^{(n)}(s)=
%As already observed, the above function $\cF$ satisfies all the
%hypothesis of Theorem~\ref{astrattobis}. In particular,
%Theorem~\ref{astrattobis},~\eqref{baseperarrosto} implies the
%Tauberian limit~\eqref{marina}, i.e. $U^{(1)} (y)/y \to1/2$ as $y
%we conclude that
%F^{(n)}(0)\Bigr)=\frac{1}{2}\left( e^{-s} -1 +\int_{(0,1)} s e^{-s
%y} + \int_{(0,1)} \frac{e^{-sy}-1}{y} \,dy \right).
%We point out that the above algebra follows closely the derivation
%of (6.19) in \cite{FMRT0}: this has been possible due to the
%abstractization of the transformation introduced in \cite{FMRT0}.
%
%
%Having~\eqref{balena} and~\eqref{sognando}, it remains to analyze
%the limit
% -\int_{[0, d^{(n)}-1)}
%.
%For the identity we have used~\eqref{secchio}. Here the arguments
%become very model--dependent and there is nothing similar in
%structure analyzed in \cite{FMRT0}.
Equation~\eqref{duro} then follows from Claim~\ref{30gradi}. From this
formula one can check that $\lim_{s \to0}
\mathbb{E}(e^{-sY^{(\infty)}})=1$ and $\lim_{s \to\infty}
\mathbb{E}(e^{-sY^{(\infty)}})=0$, thus implying
$Y^{(\infty)}\in(0,\infty)$. Indeed, it is known that $\mathrm{Ei}(x) =
- \bar \gamma- \log(x) - \sum_{n=1}^\infty\frac{(-x)^n}{n \cdot n!}$
for $x
>0$, which, after some computation leads to the limit when $s \to0$,
while for $s \to\infty$, it is enough to observe that $\mathrm{Ei}(s)
\to0$ and thus $\tanh(\mathrm{Ei}(s)/2) \to0$.

Finally, we remark that condition~\eqref{pato} is satisfied if $\mu$
has finite\break  $(1+\varepsilon)$-moment. Indeed, under this hypothesis it
holds $\int_{[1,z]} x^2 \mu(dx) \leq\break  z^{1-\varepsilon} \int
x^{1+\varepsilon} \mu(dx) \leq C z^{1-\varepsilon}$ for
$\varepsilon\in(0,1)$ and $\int_{[1,\infty )} x^2 \mu(dx)<\infty$ if
$\varepsilon\geq1$.
\end{pf*}

%cl7.1 #&#
\begin{claim}\label{30gradi}
%
%e69 #&#
\begin{equation}
\lim_{z \to
\infty} \biggl( \int_{[0, z-1)}
\frac{1}{1+x}t^{(1)} (dx)- \frac{1}{2} \ln z \biggr)=
\frac{1}{2}\log2 + \frac{\bar\gamma}{2} -\frac{1}{2}\ln(\bar\mu),
\end{equation}
where $\bar\gamma\simeq0,577$ is the Euler--Mascheroni constant.
\end{claim}
%
%for a~suitable positive constant $C_*$. Assuming the above claim
% we conclude that
%$\sqrt{d^{(n)}} \exp\left\{ - F^{(n)}(s) \right\}$ goes to
%$e^{-C_*} \sqrt{\bar\mu}$. Combining this limit with~\eqref{balena}
%and~\eqref{sognando} one gets immediately the thesis.

\begin{pf*}{Proof of Claim~\ref{30gradi}}
First of all we give an explicit formula for the measure $m(dx)$ with
support in $[1, \infty)$ such that
%
%e70 #&#
\begin{equation}
\label{falchetto} \int_A m (dx)= \int_{A-1}
\frac{1}{1+x} t^{(1)}(dx), \qquad A \subset[1,\infty)\mbox{ Borel}.
\end{equation}

%le7.2 #&#
\begin{Lemma}\label{lemma7.2}
Let $m(dx)$ be the measure defined by~\eqref{falchetto}. Let $
\otimes_k\mu$ be the convolution of $k$ copies of the interval law
$\mu$. Then
\[
m(A)= \sum_{k=1} ^\infty
\alpha_k [\otimes_k \mu] (A), \qquad A \subset[1,\infty)\mbox{ Borel},
\]
where $\alpha_k:= (1+(-1)^{k+1}  )/(2k)$.
\end{Lemma}
Note that, since $\mu$ has support in $[1, \infty)$, the probability
measure $\otimes_k \mu$ has support in $[k,\infty)$.

\begin{pf*}{Proof of Lemma \ref{lemma7.2}}
We know that $t^{(1)}$
satisfies~\eqref{asilobello} with
$\mathcal{F}(x)=\operatorname{arctanh}(x)$. Since $g^{(1)}(s)= g(s):=
\int e^{-sx} \mu(dx)$, by~\eqref{falchetto} the identity
\eqref{asilobello} can be rewritten as
$\mathcal{F}( g(s) ) = \int_1 ^\infty e^{-sx} m(dx)$.
For $s$ large $g(s)$ goes to zero, and hence we can use the analytic
expansion of $\mathcal{F}(x)$ around zero [recall that
$\operatorname{arctanh}(x)=1/2\ln\frac{1+x}{1-x}$ and use
\eqref{zittizitti} with $\gamma=0$] getting
\[
\sum_{k=1}^\infty
\alpha_k g(s)^k = \int_1
^\infty e^{-sx} m(dx), \qquad s\mbox{ large}.
\]
Since $g(s)^k = \int e^{-sx} [\otimes_k \mu](dx)$, the above equation
can be written as
\[
\sum_{k=1}^\infty\alpha_k \int
e^{-sx} [\otimes_k \mu](dx)= \int_1
^\infty e^{-sx} m(dx), \qquad s\mbox{ large}.
\]
The thesis then follows from Theorem 1a in \cite{F}, Section XIII.1.
\end{pf*}
Let $W_1, W_2, \ldots, W_k$ be i.i.d. random variables with common law
$\mu$. Then, $\otimes_k \mu$ is the law of $W_1+W_2+ \cdots+ W_k$. Due
to the above lemma and since $W_i \geq1$ a.s., we can write
%
%e71 #&#
\begin{equation}
\label{tamburo} \int_{[0,z-1)} \frac{t^{(1)}(dx)}{1+x}= \int
_{[1,z)} m(dx)= \sum_{k=1}
^{ \lfloor
z \rfloor}\alpha_k \mathbb{P}( W_1+
\cdots+W_k \leq z),
\end{equation}
where $\lfloor z\rfloor$ denotes the integer part of $z$.
%Since the
%r.h.s. is a~monotone function in $z$, for the proof of our claim it
%is enough to restrict to the case of integer $z$.

Recall that $\bar\mu:= \int x \mu(dx) = \mathbb{E}( W_i)\geq1$. If
$\bar \mu=1$, then $\mu= \delta_1$ and, as the reader can check, the
arguments below become trivial. Hence we assume that
\mbox{$\bar\mu>1$}.

Given $z >1$ we define \mbox{$ \tilde W_i:= W_i \mathbh{1} (W_i \leq z ) $}
and $\bar\mu(z):= \mathbb{E}( \tilde W_i) =\mathbb{E}(W_i;\break  W_i
\leq z)$. We can estimate the variance of $\tilde W_i$ as
%
%e72 #&#
\begin{equation}
\label{barbalalla}
\operatorname{Var}( \tilde W_i) \leq\mathbb
{E}\bigl( \tilde W_i^ 2\bigr)= \mathbb{E}\bigl(
W_i^2; W_i \leq z\bigr).
\end{equation}
Fix $\varepsilon>0$. We deal separately with the case \textup{(i)}
$k\leq\frac {z}{\bar \mu} (1-\varepsilon)$ and \textup{(ii)}
$k\geq\frac{z}{\bar\mu} (1+\varepsilon)$.
\begin{itemize}
\item Case \textup{(i)}. Since $\lim_{ z \to\infty} \bar\mu(z) =
    \bar \mu$, this implies that there exists $z(\varepsilon)$
    large enough and independent from $k$ such that for $z\geq
    z(\varepsilon)$ it holds
%
%e73 #&#
\begin{equation}
k\label{equationpalla} \bar\mu(z) <z \quad\mbox{and}\quad %$$\frac{\bar\mu(z) - \bar\mu}{ z/k -\bar\mu(z) }\leq1$$
%which in turn implies
\biggl(\frac{z- k \bar\mu}{ z-k \bar\mu(z) } \biggr)^2\leq4.
\end{equation}
\end{itemize}
Therefore for $z\geq z(\varepsilon)$ thanks to~\eqref{equationpalla} we
can use the Markov inequality to obtain
%
%e74 #&#
\begin{eqnarray}\label{ben10}
&& \mathbb{P}( W_1 + \cdots+W_k >z)\nonumber
\\
&&\qquad  \leq \mathbb{P}( \exists i \leq k\dvtx  W_i \neq \tilde W_i) +
\mathbb{P}(\tilde W_1 + \cdots+\tilde W_k >z)
\nonumber\\[-2pt]\\[-18pt]
&&\qquad \leq k \mathbb{P}( W_1 >z) + \mathbb{P} \biggl(\frac{\tilde W_1 +
\cdots+\tilde W_k}{k} -\bar\mu(z) > \frac{z- k \bar\mu(z)}{k} \biggr)\nonumber
\\
&&\qquad  \leq k
\mathbb{P}(W_1>z) + \frac{k \operatorname{Var}(\tilde
W_1)}{ (z-k \bar\mu(z) )^2}.\nonumber %%\\ & k \bbP(W_1>z) + \frac{k)}{ (z-k \mu(z) )^2}.
\end{eqnarray}
Then, combining~\eqref{barbalalla},~\eqref{equationpalla} and
\eqref{ben10}, we get for $z\geq z(\varepsilon)$
%
%e75 #&#
\begin{eqnarray}\label{ben110}
\quad && \alpha_k \mathbb{P}( W_1 +
\cdots+W_k >z)\nonumber
\\
&&\qquad  \leq \mathbb{P}( W_1>z) +
\frac{ \mathbb{E}( W_1^2; W_1 \leq z) }{ (z-k
\bar\mu(z) )^2}
\leq \mathbb{P}( W_1>z) + 4\frac{ \mathbb{E}( W_1^2; W_1 \leq z) }{ (z-k
\bar{\mu})^2}.
\end{eqnarray}
\begin{itemize}
\item Case \textup{(ii)}. By similar arguments one can prove
    that
    there exists $\bar z(\varepsilon)$ such that for $z\geq\bar
    z(\varepsilon )$ it holds
%
%e76 #&#
\begin{equation}
\label{ben111} \alpha_k \mathbb{P}( W_1 +
\cdots+W_k \leq z) \leq \mathbb{P}( W_1>z) + 4
\frac{\mathbb{E}( W_1^2; W_1 \leq z)}{(z-k
\bar{\mu})^2}.
\end{equation}
\end{itemize}

% Using
%the Markov inequality, if $k \bar\mu(z) <z$ we can bound
% \bbP( W_1 + \cdots+W_k >z) & \leq\bbP( \exists i \leq k\dvtx   W_i
%+\bbP(\tilde W_1 + \cdots+\tilde W_k >z)\\
%& \leq k \bbP( W_1 >z) + \bbP\Big(\frac{\tilde W_1 + \cdots+\tilde
%W_k}{k} -\bar\mu(z) > \frac{z- k \bar\mu(z)}{k}\Big)\\ & \leq k
%%\\ & k \bbP(W_1>z) + \frac{k)}{ (z-k \mu(z) )^2}.
%Let us suppose that $k \leq\frac{z}{\bar\mu} (1-\e)$. Since $
%$k$. In particular, we are allowed to use~\eqref{ben10}. In
%addition, for $z \geq z(\e)$ and $k \leq\frac{z}{\bar\mu}
%(1-\e)$, it holds
%(z) }= 1+ \frac{\bar\mu(z) - \bar\mu}{ z/k -\bar\mu(z) }\leq2
%
%Combining~\eqref{barbalalla},~\eqref{ben10} and~\eqref{palline}, for
%any $\e>0$ and for $z\geq z(\e)$ we can bound
%$$ \a_k \bbP( W_1 + \cdots+W_k >z) \leq\bbP( W_1>z) + 4\frac{ \bbE(
%W_1^2; W_1 \leq z) }{ (z-k \bar\mu
%)^2}, \qquad\forall k \leq\frac{z}{\bar\mu} (1-\e).
%$$
%By similarly arguments, the same estimate holds also for any $\e>0$,
%$z \geq z(\e)$ and $k \geq\frac{z}{\bar\mu} (1+\e)$.

At this point we get
%
%e77 #&#
\begin{eqnarray}\label{perpetuo}
\sum_{k=1}^{\lfloor (z/\bar\mu) (1-\varepsilon)\rfloor}
\alpha_k+\mathcal{E}_1 &\leq& \sum
_{k=1} ^{ \lfloor z \rfloor} \alpha_k \mathbb{P}(
W_1+ \cdots +W_k \leq z)
\nonumber\\[-8pt]\\[-8pt]
&\leq& \sum
_{k=1}^{\lfloor (z/\bar\mu) (1+\varepsilon)\rfloor} \alpha_k+
\mathcal{E}_2,\nonumber %+ \sum_{k=\frac{z}{\bar\mu} (1-\e)}^{\frac{z}{\bar\mu} (1+\e)}\a_k
\end{eqnarray}
where
%, using that $0\leq a_k k \leq2$,
the error $\mathcal{E}_1$ can be bounded via~\eqref{ben110} as
\begin{eqnarray*}
|\mathcal{E}_1| &\leq& \sum_{k=1}^{\lfloor (z/\bar\mu)
(1-\varepsilon)\rfloor}
\biggl( \mathbb{P}( W_1>z) + 4\frac{ \mathbb
{E}( W_1^2; W_1
\leq z) }{ (z-k \bar\mu)^2} \biggr)
\\
&\leq& z \mathbb{P}( W_1>z) +C \mathbb{E}\bigl( W_1^2;
W_1 \leq z\bigr) \int_{\varepsilon z} ^z
\frac{1}{x^2} \,dx
\\
&\leq& z \mathbb{P}( W_1>z) + \frac{C'}{\varepsilon
z}\mathbb{E}\bigl( W_1^2; W_1 \leq z\bigr)
\end{eqnarray*}
and similarly the error $\mathcal{E}_2$ can be bounded
via~\eqref{ben111} as
\[
|\mathcal{E}_2| %\sum_{k=\lfloor
%4\frac{ \bbE( W_1^2; W_1 \leq z) }{ (z-k \bar\mu)^2}\Big)
\leq z \mathbb{P}( W_1>z) + \frac{C'}{\varepsilon z}
\mathbb{E}\bigl( W_1^2; W_1 \leq z\bigr).
\]
The bound $\bar\mu=\mathbb{E}(W_1)<\infty$ trivially implies that
$\lim_{z \to\infty} z \mathbb{P}( W_1 >z)=0$. This observation,
together with hypothesis~\eqref{pato}, assures that for any fixed
$\varepsilon>0$ it holds
%
%e78 #&#
\begin{equation}
\label{richiamo} \lim_{z\to\infty}\mathcal{E}_1=\lim
_{z\to\infty}\mathcal{E}_2=0.
\end{equation}
We point out that the above estimates follow closely the arguments used
to prove the weak LLN. If $\mu$ has finite variance, exactly as in the
proof of the LLN, the truncation $\tilde W_i$ would be unnecessary, and
a~direct application of the Markov inequality would allow us to
estimate $\mathcal{E}_1, \mathcal{E}_2$.

It remains to study the behavior of the series $ \sum_{k=1} ^n \alpha_k$
for $n$ integer. It is known that $ \sum_{k=1}^n \frac{1}{k} = \log n +
\bar\gamma+ o(1)$, where $\bar\gamma$ is Euler--Mascheroni constant.
Assume that $n$ is even and $n=2p$. Then
\[
\phi(n) = \mathop{\sum_{k=1}}_{k\ \mathrm{odd}}^n
\frac{1}{k} = \sum_{k=1}^n
\frac{1}{k} - \frac{1}{2}\sum_{k=1}^p
\frac{1}{k} %= \log n + \alpha- \frac{1}{2}(\log p + \alpha) +o(1)
= \frac{1}{2} \log2 + \frac{1}{2}
\log n + \frac{\bar\gamma}{2} + o(1).
\]
%
%We first take $\g=0$.
% Let us restrict to $z$ integer and odd. Since
% $\frac{1}{k}= \frac{1}{2} \int_k^{k+2} \frac{1}{k} \,dx $, we can write
%$$ \sum_{k=1}^z \a_k = \sum_{k \mbox{ odd}, 1\leq k \leq z}
%1\leq k \leq z}\frac{1}{2} \int_k ^{k+2} \frac{1}{x} \,dx+
% \sum_{k \mbox{ odd},
%1\leq k \leq z}\frac{1}{2} \int_k ^{k+2} \left(\frac{1}{k} -
%$$
%Let us analyze the r.h.s. The first sum in the r.h.s. equals
%$\frac{1}{2} \ln(z+2)$.
% We note that for $x \in[k, k+2]$ it
%holds $|x^{-1} - k^{-1}| \leq(x-1)^{-2}$. Hence the last sum is
%absolutely convergent as $z $ goes to $\infty$. Since $| \ln z -
%$ goes to zero as $k $ goes to $\infty$ we conclude that there
%exists a~finite constant $C_*$ such that
For $n$ odd, one obtains a~similar expression. Hence, we conclude that
%
%e79 #&#
\begin{equation}
\label{russia} \lim_{z \to\infty} \Biggl(\sum
_{k=1}^{ \lfloor z \rfloor} \alpha_k - \frac{1}{2}
\ln z \Biggr) = \frac{1}{2} \log2 + \frac{\bar\gamma}{2}.
\end{equation}
%
%without any restriction on $z$.
% \centerline{ {\mathbf Ale:} la suddetta formula si dovrebbe estendere
% a~tutti i $\g$. Nel seguito assumo che sia provata nel caso generale}
Collecting~\eqref{tamburo},~\eqref{perpetuo},~\eqref{richiamo} and
\eqref{russia} we get %for $\g=0$
that
\[
C_* +\frac{1}{2} \ln \biggl(\frac{z(1-\varepsilon)}{\bar\mu} \biggr) - o(1) \leq\int
_{[0,z-1) } \frac{t^{(1)}(dx)}{1+x}\leq C_* +\frac{1}{2} \ln
\biggl(\frac{z(1+\varepsilon)}{\bar\mu} \biggr) + o(1),
\]
where $o(1)$ goes to zero as $z \to\infty$ (for any fixed
$\varepsilon>0$) and $C_* =
\frac{1}{2} \log2 + \frac{\bar\gamma}{2} $. Hence
\[
\biggl\llvert \int_{[0, z-1)} \frac{t^{(1)}(dx)}{1+x}-
\frac{1}{2} \ln (z/\bar\mu)-C_* \biggr\rrvert \leq C \bigl(\varepsilon+ o(1)
\bigr).
\]
At this point take first the limit $z \to\infty$ and then the limit
$\varepsilon \downarrow0$, thus concluding the proof of our claim.
\end{pf*}

\section{Universal coupling: Graphical construction of the dynamics}\label{universal}

In this section we describe the universal coupling for the OCPs. The
construction is standard and very similar to the one presented in
Section~3.1 of \cite{FMRT0}. On the other hand, it will be used in
Section~\ref{semigruppo_OCP} and is fundamental in order to recover
results as Lemma~\ref{hollanda} and the first part of Proposition
\ref{note}.

Given $ \xi\in\mathcal{N}(d_{\min})$, we enumerate its points in
increasing order with the rule that the smallest positive one (if it
exists) gets the label $1$, while the largest nonpositive one (if it
exists) gets the label $0$. We write $N(x, \xi)$ for the integer number
labeling the point $x \in\xi$. This allows us to enumerate the domains
of $\xi$ as follows: a~domain $[x,x']$ is said to be the $k$th domain
if \textup{(i)} $x$ is finite and $N(x, \xi)=k$, or \textup{(ii)}
$x=-\infty$ and $N(x', \xi)=k+1$. Recall that if $x=-\infty$, then
$\xi$ is unbounded from the left and $x'$ is the smallest number in
$\xi$.

We set $\|\lambda\| _\infty= \sup_{d \in[d_{\min},d_{\max} )}
\lambda(d)$ where we recall that
$\lambda=\lambda_r+\lambda_\ell+\lambda_a$. We consider a~probability
space $ (\Omega, \mathcal{F},P )$ on which the following random objects
are defined and are all independent: the Poisson processes
$\mathcal{T}^{(k)} = \{ T_m^{(k)}\dvtx   m\in\mathbb{N}\}$, $\bar
\mathcal{T}^{(k)} = \{\bar T_m^{(k)}\dvtx   m\in\mathbb{N}\}$ and
$\tilde{\mathcal{T}}^{(k)} = \{\tilde{T}_m^{(k)}\dvtx m\in\mathbb{N}\}$
of parameter $\| \lambda \|_\infty$, indexed by $k \in\mathbb{Z}$, and
the random variables $U^{(k)}_m$, $\bar U^{(k)}_m$ and
$\tilde{U}_m^{(k)}$, uniformly distributed in $[0,1]$, indexed by $k
\in\mathbb{Z}$ and $m \in\mathbb{N}$. Above, the Poisson processes are
described in terms of the jump times $T_m ^{(k)}$, $\bar T_m^{(k)}$,
$\tilde{T}_m^{(k)}$. By discarding a~set of $P$-probability $0$, we may
assume that
%
%e80 #&#
\begin{eqnarray}
\label{disgarding} &\mbox{As $k_1,k_2,k_3$ vary
in $\mathbb{Z}$, the sets } \mathcal {T}^{(k_1)}, \bar
\mathcal{T}^{(k_2)}\mbox{ and }\tilde{\mathcal{T}}^{(k_3)}&
\nonumber\\[-10pt]\\[-10pt]
&\mbox{are locally finite and disjoint.}&\nonumber
\end{eqnarray}

Next, given $\zeta\in\mathcal{N}(d_{\min})$ and $\omega\in \Omega$, to
each domain $\Delta$ that belongs to $\zeta$, we associate the Poisson
process $\mathcal{T}^{(k)}$ if $\Delta$ is the $k$th domain in $\zeta$.
In this case, we write $\mathcal{T}^{(\Delta)}$ instead of
$\mathcal{T}^{(k)}$. Similarly we define $\bar\mathcal {T}^{(\Delta)}$,
$\tilde{\mathcal{T}}^{(\Delta)}$, $U^{(\Delta )}_m$, $\bar
U^{(\Delta)}_m$ and $\tilde{U}^{(\Delta)}_m$. The idea behind the
construction of the universal coupling is the following: if, for
example, $s=T^{(\Delta)}_m$ for some $m \in\mathbb{N}$ and if the
domain $\Delta$ is present at time~$s-$, then the left extreme of
$\Delta$ has to be erased at time $s$ if and only if $U_m
^{(\Delta)}\leq\lambda _\ell (d)/\|\lambda\|_\infty$, $d$ being the
length of the domain $\Delta $. Similarly, the right extreme (or both
the extremes) of $\Delta$ can be erased at time $s=\bar T^{(\Delta)}_m$
[resp., $\tilde{T}^{(\Delta)}_m$]. Working with infinite domains, to
formalize the above construction one needs some percolation argument as
presented below.

We define $\mathcal{W}_t[\omega,\zeta]$ as
the set of domains $\Delta$ in $\zeta$ such that
%
%e81 #&#
\begin{eqnarray}
\label{caldo-bestiale}\quad 
&& \bigl\{ s\in[0,t]\dvtx  s \in\mathcal{T}^{(\Delta)}\cup\bar
\mathcal {T}^{(\Delta)} \cup\tilde {\mathcal{T}}^{(\Delta)}
\mbox{ or }s \in \mathcal{T}^{(\Delta')} \cup\bar\mathcal{T}^{(\Delta')}
\cup \tilde{\mathcal{T}}^{(\Delta')}
\nonumber\\[-8pt]\\[-8pt]
&&\hspace*{127pt}\mbox{for some domain } \Delta' \mbox{ neighboring } \Delta
\bigr\} \neq\varnothing.\nonumber
\end{eqnarray}
On $\mathcal{W}_t[\omega,\zeta]$ we define a~graph structure putting an
edge between domains $\Delta$ and $\Delta'$ if and only if they are
neighboring in $\zeta$. Since the function $\lambda$ is bounded from
above, we deduce that the set
\[
\hspace*{-3pt}\mathcal{B}(\zeta):= \bigl\{ \omega\dvtx  \mathcal{W}_t[\omega,\zeta ]
\mbox{ has all connected components of finite cardinality }\forall t\!\geq\!0 \bigr\}
\]
has $P$-probability equal to $1$. Note that the event $ \mathcal
{B}(\zeta)$ depends on $\zeta$ only through the infimum and the
supremum of the set $ \{ N(x,\zeta)\in\mathbb{Z}\dvtx   x \in\zeta\}$.
By a~simple argument based on countability, we conclude that $
P(\mathcal{B})=1$, where $\mathcal{B}$ is defined as the family of
elements $\omega\in\Omega$ satisfying~\eqref{disgarding} and belonging
to $\bigcap_{\zeta\in\mathcal{N}(d_{\min})} \mathcal{B}(\zeta)$,
%. and such that all the sets $\cT^{(k)}[\o]$, $\bar\cT^{(k)}[\o]$ and
%$\tilde{T}^{(k)}[\o]$, $k \in\bbZ$, are
%disjoint:
%e82 #&#
\begin{equation}
\label{spagna} \mathcal{B}= \bigcap_{\zeta\in\mathcal{N}(d_{\min})} \mathcal
{B}(\zeta) \cap \bigl\{\omega\in\Omega\dvtx  \omega\mbox{ satisfies }\eqref{disgarding} \bigr\}.
\end{equation}

In order to define the path $\{\xi(s)\}_{s\ge0}:=
\{\xi^\zeta(s,\omega)\}_{s\ge0}$ associated to $\zeta\in\mathcal
{N}(d_{\min} )$ and $\omega\in\Omega$, we first fix a~time $t>0$ and
define the path up to time $t$. If $\omega\notin\mathcal{B}$, then we
set
\[
\xi(s)=\zeta\qquad \forall s \in[0,t].
\]
If $\omega\in\mathcal{B}$, recall the definition of the graph
$\mathcal{W}_t[\omega,\zeta] $. Given a~set of domains $ V$ we write
$\bar V$ for the set of the associated extremes, that is, $x \in\bar V$
if and only if there exists a~domain in $V$ having $x$ as left or right
extreme. Moreover, we write $\mathcal{V}_t[\omega,\zeta] $ for the set
of all domains in $\zeta$ that do not belong to
$\mathcal{W}_t[\omega,\zeta]$. We require that
%
%e83 #&#
\begin{equation}
\label{heidi} \xi(s)\cap\overline{ \mathcal{V}_t[\omega,\zeta]}:=
\overline{ \mathcal{V}_t[\omega,\zeta]} \qquad\forall s \in[0,t],
\end{equation}
that is, up to time $t$ all points in $\overline{ \mathcal
{V}_t[\omega,\zeta] }$
survive.
% Since $\o\not\in\cB$, for all time $t \geq0$
% this graphs has all clusters of finite cardinality. If the domain
% $[x,x']$ does not belong to $\cW_t[\o,\z] $, then we impose that
% $x,x' \in\xi(t )[\z, \o]$.
Let us now fix a~cluster $ \mathcal{C}$ in the graph
$\mathcal{W}_t[\omega,\zeta]$.
% and write $C$ for the set of points that
% are extremes of some domain in $\cC$.
The path $ ( \xi(s )\cap\bar\mathcal{C}\dvtx   s \in[0,t]  )$ is
implicitly defined by the following rules (the definition is well posed
since $\omega\in\mathcal{B}$). If $s \in[0,t]$ equals $T ^{(\Delta)}_m$
with $\Delta=[x,x'] \in \mathcal{C}$ and $x,x'\in\xi(s-)$, then the
ring at time $T^{(\Delta)}_m$ is called \textit{legal} if
%
%e84 #&#
\begin{equation}
\label{eq:uniformbis} U_m^{(\Delta)} \leq\frac{ \lambda_\ell( x'-x)}{\|\lambda\|_\infty}
\end{equation}
and in this case we set $\xi(s) \cap\bar\mathcal{C}:=(\xi(s-) \cap\bar
\mathcal{C}) \setminus\{x\}$, otherwise we set $\xi(s) \cap\bar
\mathcal{C}= \xi(s-) \cap\bar\mathcal{C}$. In the first case we say
that $x$ is erased and that the domain $[x,x']$ has incorporated the
domain on its left. Similarly, if $s \in[0,t]$ equals $ \bar T
^{(\Delta)}_m$ with $\Delta=[x,x'] \in\mathcal{C}$ and $x,x'
\in\xi(s-)$, then the ring at time $\bar T^{(\Delta)}_m$ is called
\textit{legal} if
%
%e85 #&#
\begin{equation}
\label{eq:uniform2bis} \bar U_m^{(\Delta)} \leq\frac{ \lambda_r ( x'-x)}{\|\lambda\|
_\infty}
\end{equation}
and in this case we set $\xi(s) \cap\bar\mathcal{C}:=(\xi(s-) \cap \bar
\mathcal{C}) \setminus\{x'\}$, otherwise we set $\xi(s) \cap\bar
\mathcal{C}= \xi(s-) \cap\bar\mathcal{C}$. Again, in the first case we
say that $x'$ is erased and that the domain $[x,x']$ has incorporated
the domain on its right. Finally, if $s \in[0,t]$ equals $
\tilde{T}^{(\Delta)}_m$ with $\Delta=[x,x ' ] \in\mathcal{C}$ and $x,x'
\in\xi(s-)$, then the ring at time $\tilde{T}^{(\Delta)}_m$ is called
\textit{legal} if
%
%e86 #&#
\begin{equation}
\label{eq:uniform3bis} \tilde{U}_m^{(\Delta)} \leq\frac{ \lambda_a ( x'-x)}{\|\lambda\|
_\infty}
\end{equation}
and in this case we set $\xi(s) \cap\bar\mathcal{C}:=(\xi(s-) \cap
\bar \mathcal{C}) \setminus\{x,x'\}$, otherwise we set $\xi(s) \cap\bar
\mathcal{C}=
\xi(s-) \cap\bar\mathcal{C}$. Again, in the first case we say that
$x$ and
$x'$ are erased\vadjust{\goodbreak} and that the domain $[x,x']$ has incorporated both the
domains from its right on its left.

We point out that $\bar\mathcal{C}\cap\bar\mathcal{C}'= \varnothing $
if $ \mathcal{C}$ and $ \mathcal{C}'$ are distinct clusters in
$\mathcal{W}_t[\omega,\zeta ]$. On the other hand, it could be
$\bar\mathcal{C}\cap\overline{ \mathcal{V}_t[\omega,\zeta]}\neq
\varnothing$. Let $x $ a~point in the intersection, and suppose, for
example, that $[a,x] \in\mathcal{C}$ while
$[x,b]\in\mathcal{V}_t[\omega,\zeta] $. Then, by definition of
$\mathcal{W}_t [\omega, \zeta]$, one easily derives that the Poisson
processes associated to the domains $[a,x] $ and $[x,b]$ do not
intersect $[0,t]$, while at least one of the Poisson processes
associated to the domain on the left of $[a,x] $ intersects $[0,t]$. In
particular, $x \in\xi(s) \cap\bar \mathcal{C}$ for all $s \in[0,t]$, in
agreement with~\eqref{heidi}. The same conclusion is reached if
$[a,x]\in\mathcal{V}_t[\omega,\zeta] $ and $[x,b] \in\mathcal{C}$. This
allows us to conclude that the definition of the path
$\{\xi(s)\}_{s\ge0}$ up to time $t$ is well posed. We point out that
this definition is $t$-dependent. The reader can easily check that,
increasing $t$, the resulting paths coincide on the intersection of
their time domains. Joining these paths together we get
$\{\xi(s)\}_{s\ge0}$.

Given a~configuration $\zeta\in\mathcal{N}(d_{\min})$, the law of the
corresponding random path $\{\xi(s)\}_{s\ge0}$ is that of the OCP with
initial condition $\zeta$. The advantage of the above construction is
that all OCP's, obtained by varying the initial configuration, can be
realized on the same probability space. Given a~probability measure
$\mathcal{Q}$ on $\mathcal{N}(d_{\min})$, the OCP with initial
distribution $\mathcal{Q}$ can be realized by the random path
$\{\xi^{\cdot} (s,\cdot)\}_{s\ge 0}$ defined on the product space
$\Omega\times\mathcal{N}(d_{\min})$ endowed with the probability
measure $P\times\mathcal{Q}$.

The next result (similar to \cite{S}, Lemma 2.2) is an immediate
consequence of the above construction and of the metric defined on
$\mathcal{N}(d_{\min})$. We omit the proof.

%le8.1 #&#
\begin{Lemma} \label{hollanda}
For any $(\zeta,\omega) \in\mathcal{N}(d_{\min}) \times\mathcal {B}$,
the function $[0,\infty) \ni s \mapsto\xi^\zeta(s,\omega) \in\mathcal
{N}(d_{\min}) $ is c\`adl\`ag. In other words,
$\{\xi^\zeta(s,\omega)\}_{s\ge0}$ belongs to the Skohorod space
$D ( [0,\infty), \mathcal{N}(d_{\min}) )$.
\end{Lemma}

%%%%%%%%%%%%%%%%%%%%%%%%%%%%%%%%%%%%%%%%%%%%%%%%%%%%%%%%%%%%
%%%%%%%%%%%%%%%%%%%%%%%%%%%%%%%%%%%%%%%%%%%%%%%%%%%%%%%%%%%%
%%%%%%%%%%%%%%%%%%%%%%%%%%%%%%%%%%%%%%%%%%%%%%%%%%%%%%%%%%%%

%s9 #&#
\section{\texorpdfstring{Proof of Theorem \protect\ref{amico_marco_zero}}
{Proof of Theorem 2.9}}\label{infinitesimi}

This section is dedicated to the construction and the analysis of
the Markov generator $\mathcal{L}$ %\club
of the OCP.
%Since such a~construction is partially standard, we omit some details
%and focus on the relevant proofs.
We first introduce the Markov semigroup associated to the graphical
construction of Section~\ref{universal} and then introduce the
pregenerator $\mathbb{L}$.

If points (domain extremes) belong always to a~given countable subset
of $\mathbb{R}$ (e.g., points belong to $\mathbb {Z}$), then one can
directly apply the methods developed for interacting particle systems
on countable space \cite{L}, identifying each domain extreme with
a~particle. In the general case, we have introduced a~lattice structure
(see Section~\ref{sec_OCP}) which strongly simplifies the problem of
the Markov generator from an analytic viewpoint, and allows us to use
again the methods described in~\cite{L}. Endowing the space
$\mathcal{N}$ of locally finite subset of $\mathbb{R}$ of the vague
topology, the map $\mathcal{N}\ni\xi\mapsto\xi\cap[a,b]\in\mathcal{N}$
is not continuous, hence the above discretization requires some special
care.

%s9.1 #&#
\subsection{Markov semigroup and pregenerator}\label{semigruppo_OCP}

Given an initial configuration $\xi\in\mathcal{N}(d_{\min})$, define
the path $\{\xi(s)\}_{s\ge0} = \{\xi^\xi(s)\}_{s\ge0}$ as in
Section~\ref{universal}, with $\zeta=\xi$. Note that the dependence on
the element $\omega\in \Omega$ is understood. In what follows we will
alternatively use the notation $\{\xi^\xi(s,\omega)\}_{s \geq0}$, or
$\{\xi^\xi(s)\}_{s \geq0}$ or $\{\xi(s)\}_{s \geq0}$, depending on the
context.
% for the process starting from $\zeta=\xi$.

Let $\mathbb{P}_\xi$ be the law of the OCP starting from
$\xi\in\mathcal{N}(d_{\min}) $,
% {\it i.e.} the probability measure
%on $D([0,\infty), \cN(d_{\min}))$ defined by
\[
\mathbb{P}_\xi(A) = P \bigl( \bigl\{ \omega\in\Omega\dvtx  \bigl\{ \xi(s)
\bigr\}_{s \geq0} \in A \bigr\} \bigr) \qquad\forall A \subset D\bigl([0,
\infty), \mathcal{N}(d_{\min})\bigr)\mbox{ Borel}.
\]
We write $\mathbb{E}_\xi$ for the corresponding expectation. Then, for
any $f \in\mathbb{B}$, we set
\[
P_t f(\xi) = \mathbb{E}_\xi \bigl( f \bigl( \xi(t)
\bigr) \bigr) \qquad\forall t \geq0.
\]

Below we shall prove that $(P_t)_{t \geq0}$ is a~Markov semigroup on
$\mathbb{B}$ in the sense of the following definition \cite{L}:
%
%de9.1 #&#
\begin{definition}[(Markov semigroup)] \label{markov}
A family of linear operators $(S_t)_{t \geq0}$ on $\mathbb{B}$ is
called a~Markov semigroup if it is Feller, that is, $S_t f
\in\mathbb{B}$ for all $f\in \mathbb{B}$, and satisfies the following
properties:
\begin{longlist}[(iii)]
\item[(i)] $S_0 = \mathbh{1}_\mathbb{B}$, the identity
    operator on $\mathbb{B}$;

\item[(ii)] for any $f \in\mathbb{B}$, $\lim_{t \to0} \| S_t f
    - f \| =0$;

\item[(iii)] for any $s,t \geq0$, any $f \in\mathbb{B}$$,
    S_{t+s} f = S_t ( S_s f)$;

\item[(iv)] for any $t \geq0$, $S_t \mathbh{1} = \mathbh{1}$;

\item[(v)] for any $f \in\mathbb{B}$, $f \geq0 \Rightarrow P_t
    f   \geq0$.
\end{longlist}
\end{definition}

Before moving to the proof of the fact that $(P_t)_{t \geq0}$ is
a Markov semigroup, we need to introduce some operators and to fix some
notation.

Given $s>0$ we consider the operator $L_s$ on $\mathbb{B}$ defined as
%
%e87 #&#
\begin{eqnarray}\label{santinumi}
L_s f(\xi) &=& \mathop{\sum
_{[x,x+d]\ \mathrm{domain}}}_{\mathrm{in\  \xi\cap(-s,s)}}  \bigl\{ \lambda_\ell(d)
\bigl[ f \bigl(\xi\setminus\{x\} \bigr)- f(\xi) \bigr]\nonumber
\\[-4pt]
&&\hspace*{59pt}{} + \lambda_r (d) \bigl[ f \bigl(\xi\setminus\{x+d\} \bigr)- f(\xi) \bigr]
\\
&&\hspace*{59pt}{}+ \lambda_a (d) \bigl[ f \bigl(\xi\setminus\{x, x+d\}
\bigr)- f(\xi) \bigr] \bigr\}.\nonumber
\end{eqnarray}
Since $\xi$ is locally finite, the RHS is given by a~finite sum and
therefore is well defined. Given an integer $n\in\mathbb{N}_+$, we
define the operator $\mathbb{L}_n$ on $\mathbb{B}$ as
%
%e88 #&#
\begin{equation}
\label{ln} \mathbb{L}_n f (\xi):= d_{\min}^{-1}
\int_{n d_{\min} }^{(n+1) d_{\min} } L_s f(\xi) \,ds.
\end{equation}
Note that, given $\xi\in\mathcal{N}(d_{\min} )$, the integrand is
a~bounded stepwise function with a~finite family of jumps. Hence, it is
integrable.

Recall the notation at the beginning of Section~\ref{prova_differisco}.
Given $k \in\mathbb{Z}$ and $\xi\in \mathcal{N}(d_{\min})$, let
\begin{eqnarray*}
 c_k( \xi)&:=& \mathbh{1} \bigl( |\xi\cap I_k|=1 \bigr) \bigl[
\lambda_r \bigl( d_{z_k } ^\ell\bigr)+
\lambda_\ell\bigl( d_{z_k} ^r\bigr) \bigr],
\\
 c_{k,k'} (\xi)&:=& \mathbh{1} \bigl(|\xi\cap I_k|=1, |\xi
\cap I_{k'}|=1, |\xi\cap I_r|=0\ \forall r\dvtx  k
<r<k' \bigr) \lambda_a ( z_{k'}-z_k
),
\end{eqnarray*}
where for any $\xi$ such that $ |\xi\cap I_k|=1$ we set $z_k:= \xi \cap
I_k$ [due to the definition of $\mathcal{N}(d_{\min})$ each interval
$I_k$ contains at most one point of $\xi$].

Finally, for any $f \in\mathbb{D}$ [recall~\eqref{ddd}], set
%
%e89 #&#
\begin{equation}
\label{preistorico} \mathbb{L}f (\xi): = \sum_{r \in\mathcal{R}}
c_r (\xi) \nabla_r f(\xi), \qquad \xi\in
\mathcal{N}(d_{\min}).
\end{equation}
Since the rates $\lambda_\ell,\lambda_r,\lambda_a$ are bounded, for all
$f \in\mathbb{D}$ the series on the RHS of~\eqref{preistorico} is
absolutely convergent; hence $\mathbb{L}f (\xi)$ is well defined.

%le9.2 #&#
\begin{Lemma}\label{giostra}
The following holds:
\begin{longlist}[(iii)]
\item[(i)] For each $n \in\mathbb{N}_+$, $\mathbb{L}_n$ is a~bounded
    operator from $\mathbb{B}$ to $\mathbb{B}$.

\item[(ii)] For each $f \in\mathbb{D}$, $\mathbb{L}f \in\mathbb{B}$
    and $\mathbb{L}f = \lim_{n \to\infty}\mathbb{L}_n f$. In
    particular, $\mathbb{L}$ is an operator with domain
    $\mathcal{D}(\mathbb{L}):=\mathbb{D}$ into $\mathbb{B}$.

\item[(iii)] $\mathbb{B}_{\mathrm{loc}} \subset\mathbb{D}$
    ($\mathbb{B}_{\mathrm{loc}}$ being the set of local functions
    $f\in\mathbb{B}$).
\end{longlist}
\end{Lemma}

%re9.3 #&#
\begin{remark}\label{linguini}
Observe that, for any $f \in\mathbb{B}_{\mathrm{loc}}$ and any $\xi\in
\mathcal{N}(d_{\min})$$,  \mathbb{L}f(\xi)$ equals the RHS of
\eqref{santi}. On the other hand, we point out that the operator
$\mathbb{L}\dvtx  \mathbb {B}\supset \mathbb{D}\to\mathbb{B}$ is
a~\emph{Markov pregenerator} as defined in \cite{L}, Chapter~1,
Definition 2.1. Indeed, it holds \textup{(i)} $\mathbh{1}
\in\mathbb{D}$, \textup{(ii)} $\mathbb{D}$ is dense in $\mathbb{B}$
since it contains the subset $\mathbb{B}_{\mathrm{loc}}$ which we know
by Lemma~\ref{lemino} to be dense and finally \textup{(iii)} if $f
\in\mathbb{D}$ and $f(\xi)= \min \{ f(\xi ')\dvtx  \xi'
\in\mathcal{N}(d_{\min}) \}$ then $\mathbb{L}f (\xi) \geq0$. Due to
\cite{L}, Chapter~1, Proposition~2.2, these conditions ensure that
$\mathbb{L}$ is a~Markov pregenerator.
\end{remark}

\begin{pf} Without loss of generality, for simplicity of notation we
take \mbox{$d_{\min}=1$}.

We consider part \textup{(i)}. Let $\xi_k \to\xi$ in
$\mathcal{N}(d_{\min})$.
%Let us prove that $\bbL_s f (\xi_k)\to\bbL f(\xi)$. To this aim we
We set $R:=\{ s\in[n,n+1]\dvtx   \xi\cap\{-s,s\} = \varnothing\}$.
%Then $R \cap[n,n+1]$ is finite.
We claim that $L_s f(\xi_k) \to L_s f(\xi)$ for $s \in R$. To this aim
we apply Lemma~\ref{gnocchi}\textup{(ii)}. For $k$ large, it holds that
$\xi \cap(-s,s)$ and $\xi_k\cap(-s,s)$ have the same finite cardinality
$N$. Writing $x_j$ and $x_j^{(k)}$ for their $j$th point (from the
left), we can write
\begin{eqnarray*}
L_s f (\xi) &=& \sum_{j=1}^{N-1}
 \bigl\{ \lambda_\ell(x_{j+1}-x_j) \bigl[ f
\bigl(\xi \setminus\{x_j\} \bigr)- f(\xi) \bigr]
\\[-4pt]
&&\hspace*{20pt}{} + \lambda_r (x_{j+1}-x_j) \bigl[ f
\bigl(\xi\setminus\{x_{j+1} \} \bigr)- f(\xi) \bigr]
\\
&&\hspace*{20pt}{} + \lambda_a (x_{j+1}-x_j) \bigl[ f
\bigl(\xi\setminus\{x_j,x_{j+1} \} \bigr)- f(\xi) \bigr]
\bigr\}
\end{eqnarray*}
and a~similar expression for $ L_s f (\xi_k)$. The thesis then follows
from (a) the convergence $x_j \to x_j ^{(k)}$ as $k \to\infty$ due
to\vspace*{-2pt} Lemma~\ref{gnocchi}\textup{(ii)}, (b) the continuity
of the jump rates, (c) the convergence $\xi_k\setminus\{x^{(k)}_j\}
\to\xi\setminus\{x _j\}$, $\xi_k\setminus\{x^{(k)}_{j+1}\}
\to\xi\setminus\{x _{j+1}\}$ and $\xi_k\setminus\{x^{(k)}_j,
x^{(k)}_{j+1} \} \to\xi\setminus \{x_j, x_{j+1}\}$ as $k \to\infty$ for
$j\dvtx  1\leq j < N$, (d) the continuity of $f$.

We can now prove that $\mathbb{L}_n f$ belongs to $\mathbb{B}$. To this
aim it is enough to apply the dominated convergence theorem together
with the above claim and the following observations: (a) $R
\setminus[n,n+1]$ is finite, (b) due to the definition of
$\mathcal{N}(d_{\min})$ the function $L_s f $ has uniform norm bounded
by $C s\|f\|$, $C$ being independent from $s$.

Let us now prove part \textup{(ii)}. Since we already now that
$\mathbb{L}_n f \in \mathbb{B}$, it is enough to show that
$\sup_{\xi\in\mathcal {N}(d_{\min})} | \mathbb{L}f(\xi) -\mathbb{L}_n
f(\xi)|$ converges to zero as $n \to \infty$. By the boundedness of the
rates it holds
\[
\bigl| \mathbb{L}f(\xi) -\mathbb{L}_n f(\xi)\bigr|\leq C \mathop{\sum
_{r\in\mathcal{R}}}_{\operatorname{supp}(r) \not\subset
[-n,n]} \Delta_f(r),
\]
where the support of $r$ is defined as $\operatorname{supp}(r)= k$ if
$r=k $ and $\operatorname{supp}(r)= \{k,k'\}$ if $r=(k,k')$. The above
estimate and the fact that $f \in\mathbb{D}$ allow us to conclude.

Part \textup{(iii)} is obvious.
\end{pf}

%pr9.4 #&#
\begin{Proposition} \label{note}
The family of linear operators $(P_t)_{t \geq0}$ is a~Markov semigroup
on $\mathbb{B}$ given by contraction maps (i.e., $\| P_t f \|
\leq\|f\|$ for all $f \in\mathbb{B}$). Moreover, for any $f
\in\mathbb{B}_{\mathrm{loc}}$, it holds
% and any $\xi\in\cN(d_{\min})$, the function
%$t \mapsto P_tf(\xi)$ is differentiable at $t=0$ with derivative
%$\bb
%$$
%$$
%and the limit holds uniformly in $\xi$:
%e90 #&#
\begin{equation}
\lim_{t \downarrow0} \biggl\llVert \frac{P_t f -f}{t} - \mathbb{L} f
\biggr\rrVert = 0.
\end{equation}
\end{Proposition}

\begin{pf}
We focus on the only point that is not standard, namely the Feller
property. The rest is either a~direct consequence of the graphical
construction, or can be easily derived using the arguments presented in
\cite{S}, Chapter~2. Details are left to the reader.

Let us prove the Feller property. Fix $f \in\mathbb{B}$ and
$\varepsilon>0$.
%$\varepsilon>0$, $t >0$ and $\eta\in\cN(d_{\min})$. Our aim is
%to prove that there exists some $\delta>0$ such that for any $\eta'
Thanks to Lemma~\ref{lemino}, setting $f_N(\xi) =\int_N^{N+1}f(\xi
\cap(-r,r)) \,dr$, we are guaranteed that $f _N
\in\mathbb{B}_{\mathrm{loc}}$ and $\lim_{N \to\infty}\|f-f_N\|=0$.
Since $\|P_t f-P_t f_N\|_\infty\leq\|f-f_N\|_\infty$ as functions on
$\mathcal {N}(d_{\min} )$, approximating $f$ by $f_N$ we conclude that
it is enough to show that $P_tf_N\in\mathbb{B}$, or equivalently that
$P_t f \in\mathbb {B}$ for any $f \in\mathbb{B}_{\mathrm{loc}}$.

Let us fix $f \in\mathbb{B}_{\mathrm{loc}}$ and suppose that $f$ has
support inside $(-N,N)$ for some $N\geq1$. For simplicity of notation
we take $d_{\min}\geq1$ (the general case is completely similar). Since
$\mathcal{N}(d_{\min})$ is compact, $f$ is uniformly continuous. Hence,
there exists $\delta_0>0$ such that
%
%e91 #&#
\begin{equation}
\label{fuoriberlusconi} m(\zeta,\eta)<\delta_0 \quad\Longrightarrow\quad
\bigl|f(\zeta)-f(\eta)\bigr| \leq\varepsilon.
\end{equation}
Recall the universal coupling discussed in Section~\ref{universal} and the notation introduced therein. Depending on
$\varepsilon$, we can fix $\gamma> 10 $
large enough
such that $P(\mathcal{C}) \geq1- \varepsilon$ where $\mathcal{C}$
is the event given
by the elements $\omega\in\Omega$ for which there exist integers $k,k'$
with $ 10N \leq k,k'\leq\gamma N$ and
\[
[0,t]\cap \bigl( \mathcal{T}^{(j)} \cup\bar{\mathcal{T}}^{(j)}
\cup\tilde{\mathcal{T}}^{(j)
} \bigr) = \varnothing\qquad\forall
j=k,k-1,-k',-\bigl(k'-1\bigr).
\]
Given a~generic configuration $\zeta\in\mathcal{N}( d_{\min})$, all the
points $x$ of $\zeta\cap(-N,N)$ have index $N(x,\zeta)$ belonging to
$[-N,N]$ due to our assumption $d_{\min} \geq1$.

We claim that, if $\omega\in\mathcal{C}$, then the configuration
$\xi^{\zeta}(t)[\omega] $ inside $(-N,N)$ is univocally determined
knowing $ \mathcal{T}^{(j)}$$,  \bar{\mathcal{T}}^{(j)}$$,  \tilde
{\mathcal{T}}^{(j)}$, $(U^{(j)}_m)_{m \geq0}$, $(\bar{ U}^{(j)}_m)_{m
\geq0}$,\break  $(\tilde{U}^{(j)}_m)_{m \geq0}$ for $j \in[-\gamma N, \gamma
N]$. In order to explain this, suppose, for example, that $\zeta$ is
unbounded from the left and from the right. Then the Poisson processes
associated to the domains $[x,y]$ and $[y,z]$ do not have any time
inside $[0,t]$, where $N(x,\zeta)=k-1$, $N(y,\zeta)=k$. In particular,
both these domains can be incorporated but cannot incorporate other
domains. This implies that the point $y$ survives for all times in
$[0,t]$. Hence, whatever has happened on the right of $y$ up to time
$t$ has not influenced the dynamics on the left of $y$. The same
argument holds observing the domains $[x',y']$ and $[y',z']$ with
$N(x',\zeta)=-k'$, $N(y',\zeta)=-(k'-1)$. If $\zeta$ is bounded from
the left of from the right, the proof of our claim is even simpler.

Due to the above claim, for each $\zeta$ it holds
\[
\xi^{\zeta}(t)[\omega] \cap(-N,N)=\xi^{\bar\zeta}(t)[\omega]
\cap(-N,N)
\]
if $\omega\in\mathcal{C}$ and if $ \bar\zeta$ is the configuration obtained
from $\zeta$ by removing all the points $x \in\zeta$ with $|N(x,
\zeta)|> \gamma N$.
Recalling that $f$ has support in $(-N,N)$, we have
%
%e92 #&#
\begin{eqnarray}
\label{cantare}
\bigl| P_t f (\zeta)-P_t f ( \bar\zeta) \bigr|\nonumber
&=& \biggl|\int P(d\omega) f \bigl(\xi^{\zeta }(t)[\omega]\cap(-N,N) \bigr)
\\
&&\hspace*{3pt} - \int P(d\omega) f \bigl(\xi^{\bar\zeta
}(t)[\omega]\cap(-N,N) \bigr)
\biggr|
\\
&\leq& 2 P\bigl(\mathcal{C}^c\bigr) \|f\|\leq 2 \|f\| \varepsilon.\nonumber
\end{eqnarray}
Fix now $\zeta$. Let us suppose for simplicity that $\zeta$ is
unbounded from the left and from the right (the other cases can be
treated similarly). Then $\bar\zeta$ contains all the points $x\in
\zeta$ with index $N(x,\zeta)\in[-\gamma N,\gamma N]$. We have $\bar
\zeta= \zeta\cap(-a,b)$ for suitable $a,b>0$. Due to Lemma
\ref{gnocchi}, one can prove that there exists $\delta>0$ [smaller than
$\delta_0$, defined in~\eqref{fuoriberlusconi}] such that if
$\eta\in\mathcal {N}(d_{\min} )$ and $m(\zeta,\eta) \leq\delta$, then
$\eta\cap(-a,b) $ has the same cardinality of $\zeta\cap(-a, b)$. In
particular, $\eta\cap (-a,b)$ is given by all the points $x$ of $\eta$
with index $N(x, \eta) \in[-\gamma N,\gamma N]$. This implies that for
all $\eta\in\mathcal {N}( d_{\min})$ such that $m( \zeta,\eta)
\leq\delta$ it holds $ \bar\eta= \eta\cap(-a,b) $. Fix $\delta_1
>0$. Taking $\delta$ smaller if necessary, we can assume that
if $m(\zeta,\eta) \leq\delta$, then any two points $x\in\zeta$
and $x'
\in\eta$ with $N(x,\zeta)=N(x',\eta)$ satisfy
$|x-x'|\leq\delta_1$.\vadjust{\goodbreak}

Since~\eqref{cantare} has been obtained for any configuration in $
\mathcal{N} (d _{\min} )$, we conclude that
%
%e93 #&#
%e94 #&#
\begin{eqnarray}\label{incendio}
\qquad \bigl| P_t f (\zeta)-P_t f ( \eta) \bigr|\leq2 \|f
\| \varepsilon+ \bigl|P_t f \bigl( \zeta\cap(-a,b) \bigr) -P_t
f\bigl(\eta\cap(-a,b)\bigr) \bigr|
\nonumber\\[-8pt]\\[-8pt]
\eqntext{\forall \eta\dvtx m(\zeta, \eta) \leq\delta.}
\end{eqnarray}
Hence, in order to prove that $\zeta\mapsto P_t f (\zeta)$ is
continuous, it remains to prove that $|P_t f ( \zeta\cap(-a,b) ) -P_t
f(\eta\cap(-a,b)) |$ is small with $\varepsilon$. Fix an integer $L$
that will
be chosen later and $\eta$ so that $m(\zeta, \eta) < \delta$. Then we
decompose the expectation according to the event that the total
(random) number $X$ of clock rings
inside
$(-a,b)$, up to time $t$, is smaller or larger than $L$. Namely
%
%e95 #&#
\begin{eqnarray}
\label{corrovia}
&& \bigl|P_t f \bigl( \zeta\cap(-a,b) \bigr)
-P_t f\bigl(\eta\cap(-a,b)\bigr) \bigr|\nonumber
\\
&&\qquad \leq \bigl|\mathbb{E} \bigl(f\bigl(\xi^{\zeta\cap(-a,b)}(t)\bigr)
\mathbh{1}_{X \leq L}\bigr) - \mathbb{E} \bigl(f\bigl(\xi^{\eta\cap(-a,b)}(t)
\bigr) \mathbh{1}_{X \leq L}\bigr) \bigr|
\\
&&\quad\qquad{} + 2 \|f \| P(X \geq L),\nonumber
\end{eqnarray}
where $X$ is the cardinality of the set
\[
\bigl\{ s \in[0,t]\dvtx  s\in\mathcal{T}^{(k)} \cup\bar
\mathcal{T}^{(k)} \cup\tilde \mathcal{T}^{(k)}\mbox{ for
some } k \in\mathbb{Z}\cap[-\gamma N, \gamma N] \bigr\}.
\]
Let $t_1< \cdots<t_X $ be clock rings in the above set. Consider the
first ring $t_1$. Either this ring is legal/not legal
[see~\eqref{eq:uniformbis},~\eqref{eq:uniform2bis},
\eqref{eq:uniform3bis}] for both processes [i.e., the dynamics starting
from $\zeta\cap(-a,b)$ and the dynamics starting from
$\eta\cap(-a,b)$], or it is legal for one process and not legal for the
other one. In the first case one easily sees that
$m(\xi^{\zeta\cap(-a,b)}(t_1) \cap(-N,N)),\xi^{\eta \cap(-a,b)}(t_1)
\cap(-N,N))<\delta$ (and thus $m(\xi^{\zeta \cap(-a,b)}(s)
\cap(-N,N)),\xi^{\eta\cap(-a,b)}(s) \cap (-N,N))<\delta$ for any $s
\in[0,t_2)$). The second case takes place with probability bounded by
\[
c(\delta_1):=\sup_{ i=a,\ell,r}\ \sup_{d,d'\geq0\dvtx  |d-d'| \leq
2\delta_1}
\frac{ |\lambda_i(d)-\lambda_i(d') |}{\|\lambda\|}.
\]
By assumption, the jump rates $\lambda_a, \lambda_\ell,\lambda_r$ are
continuous functions with support in $[0,d_{\max}]$, and hence they are
uniformly continuous and thus $\lim_{\delta_1\downarrow0}
c(\delta_1)=0$. Iterating the above argument, we end up with
%
%e96 #&#
\begin{eqnarray}\label{scotta}
&& \bigl| \mathbb{E} \bigl(f\bigl(\xi^{\zeta\cap(-a,b)}(t)\bigr)
\mathbh{1}_{X \leq L}\bigr) - \mathbb{E} \bigl(f\bigl(\xi^{\eta\cap(-a,b)}(t)
\bigr) \mathbh{1}_{X \leq L}\bigr) \bigr|\nonumber
\\
&&\qquad  \leq Lc( \delta_1)
+ \mathbb{E} \bigl( \bigl|f\bigl(\xi^{\zeta\cap(-a,b)}(t) \bigr) - f_N
\bigl(\xi^{\eta\cap(-a,b)}(t) \bigr)\bigr|
\nonumber\\[-8pt]\\[-8pt]
&&\hspace*{85pt}{}\times  \mathbh{1}_{m(\xi^{\zeta\cap(-a,b)}(t) \cap(-N,N)),\xi^{\eta\cap
(-a,b)}(t) \cap(-N,N))<\delta} \bigr)\nonumber
\\
&&\qquad \leq L c(\delta_1) + \varepsilon,\nonumber
\end{eqnarray}
where in the last line we used~\eqref{fuoriberlusconi} (together with
the fact that $\delta< \delta_0$).

In remains to estimate the deviation $P(X \geq L) $ with $X$ a~Poisson
variable of mean $3tM$, where $M$ is the cardinality of $ [-\gamma
N,\gamma N] \cap\mathbb{Z}$. Since $E( e^{ X}) = \exp \{ (e-1) 3tM \}
$, setting $L= \kappa t M$ by Chebyshev inequality we get
%
%e97 #&#
\begin{equation}
\label{sole} P(X \geq\kappa t M ) \leq\exp \bigl\{ 3tM (e-1)- \kappa t M
\bigr\} \leq e^{-\kappa t M/2 }
\end{equation}
for $\kappa\geq\kappa_0$. Summing up the above estimates
[see~\eqref{incendio}, \eqref {corrovia},~\eqref{scotta},~\eqref{sole}]
we finally get the following. Fixed $\delta_1>0$ and $
\kappa>\kappa_0$, for $\delta$ small enough the bound
$m(\zeta,\eta)<\delta$ implies
\[
\bigl|P_tf(\zeta)-P_tf(\eta)\bigr| \leq2\|f\| \varepsilon+ \kappa
t M c(\delta_1)+ \varepsilon+ \|f \| e^{-\kappa t M/2}.
\]
Choosing $\kappa$ large enough, and then $\delta_1$ small enough
amounts to the desired result.
\end{pf}

%s9.2 #&#
\subsection{\texorpdfstring{Proof of Theorem \protect\ref{amico_marco_zero}}
{Proof of Theorem 2.9}}

By definition, the Markov generator $\mathcal{L}\dvtx
\mathbb{B}\supset\mathcal{D}(\mathcal{L}) \to\mathbb{B}$, associated to
the Markov semigroup $\{ P_t\dvtx  t \geq0\}$ acting on the space
$\mathbb{B}$, has domain $\mathcal{D}(\mathcal{L})$ given by
\[
\mathcal{D}(\mathcal{L}):= \biggl\{ f\in\mathbb{B}\dvtx  \lim_{t\downarrow0}
\frac{ P_t f
- f}{t}\mbox{ exists in }\mathbb{B} \biggr\}.
\]
Moreover, given $f \in\mathcal{D}(\mathcal{L})$, one sets $\mathcal
{L}f:= \lim_{t\downarrow0}\frac{ P_t f - f}{t}$. We stress that the
above limits are thought w.r.t. the uniform norm. In addition, we
recall that the space $\mathbb{B}$ depends on the parameter $d_{\min}$,
although omitted. Note that, when speaking of Markov generators, we do
not follow the definition given in \cite{L}, Chapter~1 (even if,
invoking the Hille--Yosida theorem, the two definitions coincide).

Our aim is to prove the following theorem, which corresponds to Theorem~\ref{amico_marco_zero}:
%
%th9.5 #&#
\begin{Theorem}\label{amico_marco}
The subspaces $\mathbb{B}_{\mathrm{loc}}$ and $\mathbb{D}$ are a~core
of the Markov generator~$\mathcal{L}$; that is, $\mathcal{L}$ is the
closure of the operator $\mathbb{L}\dvtx  \mathbb{D}\ni f
\mapsto\mathbb{L}f\in\mathbb{B}$, and of its restriction to
$\mathbb{B}_{\mathrm{loc}}$. Moreover, if $f \in \mathbb{D}$,
$\mathcal{L}f(\xi)$ equals the absolutely convergent series on the RHS
of~\eqref{santi}.
\end{Theorem}

We need some preparation. Our first target is to prove that the image
of $\mathbh{1} -\lambda\mathbb{L}$ (where $\mathbh{1}$ is the identity
operator) is dense in $\mathbb{B}$ for $\lambda$ sufficiently small.
% and that $P_t f \in\bbD$ for all $f \in\bbD$ and $t\geq0$.
To this aim, we follow a~strategy similar to the one adopted for
particle systems in \cite{L}, Chapter~1.
%Indeed, we first start with
%some key estimate (compare with \cite{L}, Ch. 1, Lemma 3.4:
Set $\| c\|_\infty:=\sup_{r \in\mathcal{R}} \|c_r \|_\infty$ and
note that,
by boundedness of the rates, $\| c\|_\infty< \infty$.

%le9.6 #&#
\begin{Lemma}\label{pixar:deigeni}
Suppose that $f \in\mathbb{D}$ and $f-\lambda\mathbb{L}f = g$ for
some $\lambda\geq
0$. Then for any $r \in\mathcal{R}$ it holds
%
%e98 #&#
\begin{equation}
\label{wall-e} \Delta_f (r) \leq\Delta_g (r) +
\lambda\sum_{ r'\in\mathcal
{R},  r' \neq r } \gamma\bigl(r,r'
\bigr) \Delta_f \bigl(r'\bigr),
\end{equation}
where $\gamma(r,r'):=\sup_{\xi\in\mathcal{N}(d_{\min})}
|c_{r'}(\xi^r)-c_{r'}(\xi) |$.
\end{Lemma}

\begin{pf}
Fix $\varepsilon>0$ and a~finite subset $\hat\mathcal{R}\subset
\mathcal{R}$. Take $\xi \in\mathcal{N}(d _{\min})$ such that $ \Delta_f
(r)\leq \varepsilon+ |\nabla_r f(\xi) |$. Since the map $\xi\mapsto
f(\xi^r)$ can be discontinuous we cannot avoid the error $\varepsilon$
(the setting in \cite{L} is different due to continuity). We first
consider the case that $|\nabla_r f (\xi)|= \nabla_r f (\xi)$. Then it
holds
%
%e99 #&#
\begin{eqnarray}
\label{nanomech}
\Delta_f (r) &\leq& \varepsilon+\nabla_r
f(\xi)=\varepsilon+ \nabla _r g(\xi)+ \lambda\mathbb{L}f \bigl(
\xi^r\bigr) -\lambda\mathbb{L}f (\xi)
\nonumber\\[-8pt]\\[-8pt]
&\leq& \varepsilon+
\Delta_g (r)+ \lambda\mathbb{L}f \bigl(\xi^r\bigr) -
\lambda \mathbb{L}f (\xi).\nonumber
\end{eqnarray}
Since $\xi^r = \xi\setminus I_r$ ($I_r:=I_k \cup I_{k'} $ if
$r=(k,k')$), $c_r(\xi^r)=0$ and $\nabla_r f(\xi) \geq0$, we have
%
%e100 #&#
\begin{eqnarray}
\label{rum} \mathbb{L}f \bigl( \xi^r\bigr)-\mathbb{L}f (\xi) & =&
\sum_{ r'\in\mathcal
{R} } \bigl\{ c_{r'}\bigl(
\xi^r\bigr) \nabla_{r'} f \bigl( \xi^r \bigr)-
c_{r'} (\xi) \nabla_{r'} f (\xi) \bigr\}
\nonumber\\[-8pt]\\[-8pt]
& \leq&\sum_{ r'\in\mathcal{R},  r' \neq r } \bigl\{ c_{r'}\bigl(
\xi^r\bigr) \nabla_{r'} f \bigl( \xi^r \bigr)-
c_{r'} (\xi) \nabla_{r'} f (\xi) \bigr\}.\nonumber % \leq2 \sum_{r'\in\cR,  r' \not= r } \| c_{r'}\|_\infty\D_f (r')
\end{eqnarray}
%
%Hence, in the case $ \D_f (r)\leq\e+ \nabla_r f(\xi) $ w
By our choice of $\xi$ we can write
\[
f \bigl( \bigl(\xi^{r'}\bigr)^r \bigr) - f \bigl(
\xi^{r'} \bigr) \leq\Delta _f(r) \leq\varepsilon+ f\bigl(
\xi^r\bigr) -f(\xi),
\]
thus implying that $\nabla_{r'} f (
\xi^r  ) \leq\varepsilon+\nabla_{r'} f(\xi) $. In particular,
it holds
%
%e101 #&#
\begin{eqnarray}
\label{stima1}
\qquad c_{r'}\bigl(\xi^r\bigr)
\nabla_{r'} f \bigl( \xi^r \bigr)- c_{r'} (\xi)
\nabla_{r'} f (\xi) & \leq& \bigl[ c_{r'}\bigl(
\xi^r\bigr) - c_{r'} (\xi) \bigr]\nabla_{r'} f (
\xi)+ \varepsilon\|c\|_\infty
\nonumber\\[-8pt]\\[-8pt]
& \leq&\gamma\bigl(r,r'\bigr)\Delta_f
\bigl(r'\bigr) + \varepsilon\|c\|_\infty.\nonumber
\end{eqnarray}
On the other hand, we have the trivial bound
%
%e102 #&#
\begin{equation}
\label{stima2} c_{r'}\bigl(\xi^r\bigr)
\nabla_{r'} f \bigl( \xi^r \bigr)- c_{r'} (\xi)
\nabla_{r'} f (\xi) \leq2\|c\|_\infty\Delta_f
\bigl(r'\bigr).
\end{equation}
Combining~\eqref{nanomech},~\eqref{rum} and using~\eqref{stima1} for
$r'\in\hat\mathcal{R}$ and~\eqref{stima2} for $r' \in\mathcal
{R}\setminus\hat \mathcal{R}$, we get
%
%e103 #&#
\begin{eqnarray}\label{eve}
\Delta_f (r) &\leq& \varepsilon+\Delta_g
(r)+ \lambda\sum_{r' \in
\hat\mathcal{R}\dvtx  r'\neq r} \gamma\bigl(r,r'
\bigr)\Delta_f\bigl(r'\bigr)+ \lambda\varepsilon\|c \|_\infty|\hat\mathcal{R}|
\nonumber\\[-8pt]\\[-8pt]
&&{}   +2 \lambda\| c\|_\infty\sum
_{ r'\in \mathcal{R}\setminus\hat\mathcal{R} } \Delta_f \bigl(r'\bigr).\nonumber
\end{eqnarray}
It is simple to check, by similar arguments, that the above bound
\eqref{eve} holds also in the case $|\nabla_r f (\xi)|= -\nabla_r f
(\xi)$. Note moreover that, since $f \in\mathbb{D}$, the last series in
\eqref{eve} is finite and converges to zero as $\hat\mathcal
{R}\nearrow\mathcal{R}$.
Taking first the limit $\varepsilon\downarrow0$ and then the
limit $\hat\mathcal{R}\nearrow\mathcal{R}$ we get the thesis.
\end{pf}

We can finally prove our first target:

%le9.7 #&#
\begin{Lemma}
The image $\{ f -\lambda\mathbb{L}f\dvtx   f \in\mathbb{D}\}$ is dense
in $\mathbb{B}$ for $\lambda\geq0$ small enough.
\end{Lemma}

\begin{pf}
Part of the proof is similar to the proof of \cite{L}, Chapter~1,
Theorem 3.9. We give it for completeness. Without loss of generality,
for simplicity of notation we take $d_{\min}=1$. Consider the operator
$ \mathbb{L}_n$ defined in~\eqref{ln}. As already observed in Lemma
\ref{giostra}, $\mathbb{L}_n$ is a~bounded operator $\mathbb{L}_n\dvtx
\mathbb{B}\to\mathbb{B}$. It is simple to check that $\mathbb{L}_n$ is
a Markov pregenerator; see the criterion in Remark~\ref{linguini}.
Being $\mathbb{L}_n$ a~bounded Markov pregenerator, the image of
$\mathbh{1} - \lambda\mathbb{L}_n$ is the entire space $\mathbb{B}$ for
each $\lambda\geq0$; see \cite{L}, Chapter~1, Proposition 2.8. Hence,
fixed $g \in\mathbb{D}$ we can find $f_n \in\mathbb{B}$ such that
\[
f_n - \lambda\mathbb{L}_n f_n =g.
\]
Take $ s \in(n, n+1)$. Fix $r \in\mathcal{R}$. If $\nabla_r
f_n(\xi)\geq0$, we can bound
\[
L_s f_n \bigl( \xi^r
\bigr)-L_s f_n (\xi) \leq \mathop{\sum
_{r'\in\mathcal{R}\dvtx  r' \neq r,}}_
{\operatorname{supp}(r')\subset[-n-1,n+1) } \mathcal{U} \bigl( c_{r'}
\bigl(\xi^r\bigr) \nabla_{r'} f_n \bigl(
\xi^r \bigr)- c_{r'} (\xi) \nabla_{r'}
f_n (\xi) \bigr),
\]
where $\mathcal{U}(x)=x \mathbh{1}_{\{x \geq0\}}$. Hence, averaging
over $s$, the same estimate holds for $\mathbb{L}_n$ instead of $L_s$.
Using this observation and the same arguments used in the proof of
Lemma~\ref{pixar:deigeni}, we get
%
%e104 #&#
\begin{equation}
\label{wall-eeee} \Delta_{f_n} (r) \leq\Delta_g (r) +
\lambda\mathop{\sum_{r'\in\mathcal{R}\dvtx  r' \neq r,}}_
{\operatorname{supp}(r')\subset[-n-1,n+1)
} \gamma
\bigl(r,r'\bigr) \Delta_{f_n} \bigl(r'\bigr).
\end{equation}
Introduce now the bounded operator $\Gamma\dvtx  \ell_1 ( \mathcal{R})
\to\ell_1 (\mathcal{R})$ as
\[
(\Gamma\underline x) (r)= \sum_{ r'\in
\mathcal{R}\dvtx  r' \neq r } \gamma
\bigl(r,r'\bigr) \underline x \bigl(r'\bigr), \qquad
\underline x \in\ell_1 (\mathcal{R}).
\]
The operator is bounded since $\gamma(r,r') $ is bounded by $\|c\|
_\infty$ and is zero if the supports of $r$ and $r'$ are at distance
larger than a~suitable constant depending on $d_{\min }$ and $d_{\max}$
only (recall that that rates $\lambda_\ell,\lambda_r,\lambda_a$ are
zero when evaluated at $d \geq d_{\max}$). Then bound~\eqref{wall-eeee}
implies that $ [\mathbh{1}- \lambda \Gamma] \Delta_{f_n}
\leq\Delta_{g}$. If $\lambda$ is small enough, the operator
$\mathbh{1}- \lambda\Gamma$ can be inverted, and therefore we get
%
%e105 #&#
\begin{equation}
\label{capitano} \Delta_{f_n} \leq[\mathbh{1}- \lambda
\Gamma]^{-1} \Delta_g.
\end{equation}
Let us define $g_n:= f_n - \lambda\mathbb{L}f_n$. Then
\begin{eqnarray*}
\|g- g_n\| &=& \lambda\bigl\|( \mathbb{L}- \mathbb{L}_n)f_n
\bigr\|  \leq \mathop{\sum_{r \in\mathcal{R}\dvtx }}_
{\operatorname{supp}(r) \not\subset(-n,n)} \|
c_r\|_\infty\Delta_{f_n} (r)
\\
&\leq& \|c\|_\infty \mathop{\sum_{r \in\mathcal{R}\dvtx }}_{\operatorname{supp}( r)
\not\subset(-n,n)}
[\mathbh{1}- \lambda\Gamma]^{-1} \Delta_g (r).
\end{eqnarray*}
Since $[\mathbh{1}- \lambda\Gamma]^{-1} \Delta_g \in\ell_1
(\mathcal{R})$, the above bound implies that $\lim_{n \to\infty} \|g-
g_n \|=0$. Recalling that $g \in\mathbb{D}$ and that $g_n$ belongs to
the image of $\mathbh{1}- \lambda \mathbb{L}$, we conclude that the
image of this last operator is dense in $\mathbb{D}$ and therefore in~$\mathbb{B}$.
\end{pf}

As a~consequence of the above result and Remark~\ref{linguini}, we get
that the closure $\bar\mathbb{L}$ of $\mathbb{L}$ is a~Markov generator
in the sense of \cite{L}, Chapter~1, Definition~2.7 (briefly, we will
say that $\mathbb{L}$ is an $L$-Markov generator).

%le9.8 #&#
\begin{Lemma}\label{cavallino}
If $f\in\mathbb{D}$, then there exists a~sequence $f_n\in\mathbb
{B}_{\mathrm{loc}}$ such that
% setting $f_n (\xi):= f(\xi\cap(-n d_{\min},n
%d_{\min}))$, then $f_n \in\bbB_{\mathrm{loc}}$,
$f_n \to f$ and $\mathbb{L}f_n \to\mathbb{L}f$ in $\mathbb{B}$.
\end{Lemma}

\begin{pf}
Given $n$ set $f_n(\xi):= \int_n^{n+1} f( \xi\cap(-s,s)) \,ds$. Due to
Lemma~\ref{lemino}, we know that $\|f-f_n\|\to0$ and $f_n \in
\mathbb{B}_{\mathrm{loc}}$. Let us prove that $\| \mathbb{L}f_n-\mathbb
{L}f \|\to0$. To this aim, setting $\xi_s:= \xi\cap(-s,s)$ and
observing that $(\xi_s)^r= (\xi^r)_s$ for all $r \in\mathcal{R}$, for
any integer $N$ we can write
%
%e106 #&#
%e107 #&#
\begin{eqnarray}
&& \bigl| \mathbb{L} f(\xi)-\mathbb{L}f_n (\xi) \bigr|\nonumber
\\
&&\qquad = \biggl| \int
_n^{n+1} \sum_{r\in\mathcal{R}}
c_r(\xi) \bigl( \nabla _r f(\xi) - \nabla_r f
(\xi_s)\bigr) \,ds \biggr|\nonumber
\\
&&\qquad \leq \biggl| \int_n^{n+1} \mathop{\sum
_{r\in\mathcal{R}\dvtx }}_{\operatorname{supp}(r)\not
\subset
[-N,N]} c_r(\xi) \bigl(
\nabla_r f(\xi) - \nabla_r f (\xi_s) \bigr) \,ds
\biggr|\nonumber
\\
&&\quad\qquad {}+ \biggl| \int_n^{n+1} \mathop{\sum
_{r\in\mathcal{R}\dvtx }}_{\operatorname{supp}(r) \subset[-N,N]} c_r(\xi) \bigl(
\nabla_r f(\xi) - \nabla_r f (\xi_s) \bigr)\,ds
\biggr|\nonumber
\\
&&\qquad  \leq 2\|c\|_\infty \mathop{\sum_{r\in
\mathcal{R}\dvtx  }}_{\operatorname{supp}(r)\not\subset[-N,N]}
\Delta_f(r) \label{mentana1}
\\
&&\quad\qquad {}+ 2 \|c\|_\infty\bigl|\bigl\{r\in\mathcal{R}\dvtx \operatorname{supp}(r) \subset
[-N,N]\bigr\}\bigr| \cdot\|f-f_n\|.\label{mentana2}
\end{eqnarray}
Given $\varepsilon>0$ we choose $N$ large enough that~\eqref{mentana1}
is smaller than $\varepsilon$ (this is possible since
$f\in\mathbb{D}$).
%$\sum_{r\in\cR\dvtx  {\rm supp}(r)\not\subset[-N,N]}\Delta_f(r)\leq
Afterwards, for $n$ large enough~\eqref{mentana2} is smaller
than~$\varepsilon$ (recall that $f_n \to f$ in $\mathbb{B}$). Then we
conclude that $\|\mathbb{L}f - \mathbb{L}f_n\|\leq2 \varepsilon$ for
$n$ large enough.
\end{pf}

We can finally prove Theorem~\ref{amico_marco}.

\begin{pf*}{Proof of Theorem~\ref{amico_marco}}
In Proposition~\ref{note} we have already showed that $\mathcal
{L}f=\mathbb{L}f $ if $f \in \mathbb{B}_{\mathrm{loc}}$. As observed
after Lemma~\ref{giostra}, in this case $\mathcal{L}f$ must
equal~\eqref{santi}. By Lemma~\ref{cavallino}, $\bar\mathbb{L}$ is the
closure of the restriction of $\mathbb{L}$ to
$\mathbb{B}_{\mathrm{loc}}$. Hence, $\mathbb{B}_{\mathrm{loc}}$ is
a~core of $\bar\mathbb{L}$. By Lemma~\ref{giostra}\textup{(i)}, given
$f\in\mathbb {D}$ the value $\mathbb{L}f (\xi)$ equals the RHS of
\eqref{santi} which is an absolutely convergent series.

It remains to prove that $\bar\mathbb{L}=\mathcal{L}$. Since
$\mathbb{L}f=\mathcal{L}f$ for all $f \in\mathbb{B}_{\mathrm{loc}}$,
Lemma~\ref{cavallino} and the closure of $\mathcal{L}$ implies that
$f\in\mathcal{D}(\mathcal{L})$ and $\mathbb {L}f=\mathcal{L}f $ for all
$f \in\mathbb{D}$ (the fact that $\mathcal{L}$ is close is a~standard
fact: combine Definition 2.1 in \cite{L}, Chapter~1, with the
Hille--Yosida theorem as stated in Theorem 2.9 in \cite{L}, Chapter~1,
leading to the fact that $\mathcal {L}$ is an $L$-Markov generator, and
therefore closed). This observation implies that $\mathcal{L}$ is an
extension of $\bar\mathbb{L}$. It is a~general fact that this implies
that $\mathcal{L}= \bar\mathbb{L}$; cf. \cite{S}, Proposition~3.13,
together with the Hille--Yosida theorem as stated in Theorem 2.9 in
\cite{L}, Chapter~1.
\end{pf*}

%
%. It is a~general fact that any Markov pregenerator is a~%closable operator, which is again a~Markov pregenerator (cf.
%(\bar\bbL)\to\bbB$ for the closure of $\bbL$.
%

\section*{Acknowledgements}
We warmly thank T. Kuna and F. Martinelli for useful discussions. We
acknowledge the anonymous referee for his/her careful reading of the
manuscript and for his/her useful suggestions. We thank the Laboratoire
de Probabilit\'{e}s et Mod\`{e}les Al\'{e}atoires of the University
Paris VII and the Department of Mathematics of the University of Roma
Tre for the support and the kind hospitality.


%suskaldyti doi

% imsref loaded by linak, 2013-11-19 10:37:06

\printaddresses

\end{document}